\documentclass[preprint,12pt,3p, proof]{elsarticle}
\makeatletter
\def\ps@pprintTitle{%
 \let\@oddhead\@empty
 \let\@evenhead\@empty
 \def\@oddfoot{}%
 \let\@evenfoot\@oddfoot}
\makeatother

\usepackage{tikz}
\usepackage{longtable, booktabs,multirow,lipsum,array}
\usepackage{natbib}
\usepackage{amssymb}
\usepackage{amsmath}
\usepackage{subfig}
\usepackage{enumerate}
\usepackage{mwe}
\usepackage{appendix}
\usepackage{mathtools}
\usepackage{graphicx}
\usepackage{color}
\usepackage{hyperref}
\usepackage{amsfonts}
\usepackage{graphicx}
\usepackage[utf8]{inputenc}
\usepackage[all]{xy}
\usepackage{csquotes}
\usepackage{ulem}
\usepackage{amsmath}
\usepackage{amsthm}

\newcommand\independent{\protect\mathpalette{\protect\independenT}{\perp}}
\def\independenT#1#2{\mathrel{\rlap{$#1#2$}\mkern2mu{#1#2}}}

\newcommand{\specialcell}[2][c]{%
\begin{tabular}[#1]{@{}c@{}}#2\end{tabular}}

\newtheorem{lemma}{Lemma}

\newtheorem{theorem}{Theorem}
\newtheorem{corollary}{Corollary}

\usetikzlibrary{shapes,decorations,arrows,calc,arrows.meta,fit,positioning}
\tikzset{
    -Latex,auto,node distance =1 cm and 1 cm,semithick,
    state/.style ={ellipse, draw, minimum width = 0.7 cm},
    point/.style = {circle, draw, inner sep=0.04cm,fill,node contents={}},
    bidirected/.style={Latex-Latex,dashed},
    el/.style = {inner sep=2pt, align=left, sloped}
}

\begin{document}

\begin{frontmatter}
\title{Causal inference with limited resources: proportionally-representative interventions}

\author[1]{Aaron L. Sarvet  \corref{cor1}}
\author[1]{Kerollos N. Wanis}
\author[1,2]{Jessica Young}
\author[3]{Roberto Hernandez-Alejandro}
\author[1,4,5]{Miguel A. Hern\'{a}n}
\author[1,6]{Mats J. Stensrud}

\address[1]{Department of Epidemiology, Harvard T. H. Chan School of Public Health, USA}
\address[2]{Department of Population Medicine, Harvard Medical School and Harvard Pilgrim Health Care Institute, Boston, Massachusetts, USA}
\address[3]{Division of Transplantation and Hepatobiliary Surgery, University of Rochester, Rochester, NY, USA}
\address[4]{Department of Biostatistics, Harvard T.H. Chan School of Public Health, Boston, Massachusetts}
\address[5]{Harvard-MIT Division of Health Sciences and Technology, Cambridge, Massachusetts}
\address[6]{Department of Biostatistics, University of Oslo, Norway}
\cortext[cor1]{\textbf{Contact information for corresponding author:}\\
Aaron L. Sarvet, Department of Epidemiology, Harvard T. H. Chan School of Public Health, USA. \url{asarvet@g.harvard.edu}}

\begin{abstract}
Investigators often evaluate treatment effects by considering settings in which all individuals are assigned a treatment of interest, assuming that an unlimited number of treatment units are available. However, many real-life treatments are of limited supply and cannot be provided to all individuals in the population. For example, patients on the liver transplant waiting list cannot be assigned a liver transplant immediately at the time they reach highest priority because a suitable organ is not likely to be immediately available. In these cases, investigators may still be interested in the effects of treatment strategies in which a finite number of organs are available at a given time, that is, treatment regimes that satisfy resource constraints. Here, we describe an estimand that can be used to define causal effects of treatment strategies that satisfy resource constraints:  proportionally-representative interventions for limited resources. We derive a simple class of inverse probability weighted estimators, and apply one such estimator to evaluate the effect of restricting or expanding utilization of ‘increased risk’ liver organs to treat patients with end-stage liver disease. Our method is designed to evaluate policy-relevant interventions in the setting of finite treatment resources.
\end{abstract}
\end{frontmatter}

\date{December 2019}

\newpage


\section{Introduction}
\label{sec: Introduction}

The average treatment effect is consistently estimated with data from an ideal randomized trial. This effect is identified because the expected outcome in the treatment arm is identical to the expected potential outcome had \textit{everyone} been provided treatment (and likewise for control). Observational studies can provide estimates with the same interpretation, when the data are used to emulate a randomized trial \cite{robins1986new, WhatIf}. Results of these studies are often used to justify decisions by policy makers interested in population outcomes under hypothetical policies. However, these results are not directly relevant when treatment is not available for \textit{everyone} in the target population, for example due to practical limitations on treatment resources.

Limitations on treatment resources represent important conditions in nearly every conceivable health policy setting. Consider studies that aim to assess the effects of liberal vs.\ conservative strategies for surgical blood transfusion  \cite{mazer2017restrictive}, or of care reception at hospitals with high vs.\ low procedural volume \cite{vemulapalli2019procedural}. The corresponding randomized trial implicitly considers infeasible policies, in which: \textit{all} patients follow a liberal transfusion strategy; or \textit{all} patients receive the procedure at a hospital with high procedural volume. Neither study considers policies that are implementable (due to lack of sufficient blood supply, or surgeons), yet their results form the basis of medical guidelines and national policies. For example, implementing an apparently favorable policy may lead to unexpected adverse outcomes due to increases in waiting times that are unavoidable in limited resource settings. These mechanisms are not appreciated by conventional methods.

Studies aiming to inform real-world policy-making should consider treatment strategies that are compatible with real-world resource constraints. We describe a new class of estimands that is relevant for policy makers who are interested in causal questions in settings where treatment resources are limited. We refer to estimands of this class as expected potential outcomes under \textit{proportionally-representative interventions} that constrain treatment resources.. We present an inverse probability weighting estimator for these estimands that is easy to implement with standard statistical software. To fix ideas and facilitate their introduction, the main text presents the observed data structure, identification conditions, and estimation strategies in the following common setting: patients are waiting to receive a single dose of one of two treatments types, and policy makers are considering eliminating the suspected inferior treatment type, versus continuing to use both treatment types as usual (i.e.\ no policy change). Such settings are common in medicine and public health, where e.g.\ patients are faced with decisions between treatment types with limited supply; from the patient's perspective, the decision is reducible to choosing between a plan of accepting the first treatment option that becomes available, and a plan of waiting for the suspected superior treatment option. For example, a patient with non-emergency indications for coronary artery bypass graft surgery might seek recommendations for choosing between procedure reception as soon as possible (including possibly at their local area hospital) or waiting to undergoing the procedure at a regional center of excellence. 

In general, proportionally-representative interventions permit consideration of settings where treatment resources are constrained to any hypothetical level. For example, investigators may be interested in a policy where utilization of the suspected superior treatment type is doubled (relative to extant levels), possibly corresponding to increased investment in the superior treatment resource supply. Therefore, we have included extensions to settings in which resources are arbitrarily constrained. We also present extensions to censoring, provide proofs of key results, and further discuss relationships between these new limited resource estimands and classical estimands. Finally we illustrate the method by estimating the effects of different policies for liver transplantation, where treatment resource limitations are severe.

\section{Data Structure}
\label{sec: Data Structure}

We consider a study in which $n$ individuals are followed for $k \in \{0,1,\dots,K\}$ equally spaced discrete time intervals. The individuals $i \in \{1,\dots, n\}$ are independent and identically distributed at baseline, and thus we generally  omit the $i$ subscript on the random variables.

In each interval $k$, an individual is a candidate for receiving treatments $B_k$ and $H_k$, where $B_k$ is an indicator for receiving the suspected superior treatment (e.g.\ a ``high-quality'' organ transplant) and $H_k$ is an indicator for receiving the suspected inferior treatment (e.g. a ``low-quality'' organ transplant). Let $L_k$ be a vector of the individual's covariates with support $\mathcal{L}_k$ and $Y_k$ a binary survival outcome by the end of interval $k$.

We use over-lines (e.g. $\overline{B}_k$) to indicate the history of variables, underlines (e.g. $\underline{B}_k$) to indicate the future trajectory of variables, and superscripts to indicate potential outcomes under some policy, e.g. $Y_k^{g_z}$ is the potential outcome in interval $k$ under policy $g_z$. We supplement superscripts with the `plus' symbol (+) to differentiate the natural values of treatment under some policy (e.g. $B_k^{g_z}$) from the values that treatment takes immediately following intervention on that variable under that policy (e.g. $B_k^{g_z+}$). This is coherent with the notation in section 5.1 of Richardson and Robins, 2013 \cite{richardson2013single} and see Young et. al (2014) \cite{young2014identification} for a thorough review of the distinction between natural and post-intervention treatment variables. We define a topological order within each interval as $\Big(L_k, B_k, H_k, Y_k\Big)$. 

By definition, all individuals are alive and untreated in interval 0, so $B_0=H_0=Y_0=0$, $L_0$ is equal to the empty set $\emptyset$, and individuals may only possibly receive a single treatment, such that if $B_k=1$ then $\underline{H}_k=\underline{B}_{k+1}=0$, and if $H_k=1$ then $\underline{B}_{k+1}=\underline{H}_{k+1}=0$. For notational convenience, we define the indicator functions $R_k=I\{Y_{k-1}=\overline{H}_{k-1}=\overline{B}_{k-1}=0\}$ and $S_k=I\{\overline{B}_{k}=Y_{k-1}= \overline{H}_{k-1}=0\}$. Thus, by the above definitions, $R_k$ and $S_k$ indicate treatment eligibility (i.e.\ ``being at risk'' of receiving treatment) in an interval, for suspected superior and inferior treatments, respectively. It follows immediately that $R_1=1$.
    
\section{Proportionally-representative interventions}
\label{sec: Regime}

In this section, we describe policy-relevant regimes $g_z$, within the class of proportionally-representative interventions. As elaborated below, proportionally-representative interventions are stochastic interventions defined by observed conditional treatment densities, and by \textit{a priori} user-specified resource constraints. To fix ideas, we present specific policy-relevant regimes $g_1$ and $g_0$, where $g_1$ corresponds to a policy of abolishing the suspected-inferior treatment, and  $g_0$ corresponds to a policy of no intervention.  

First, we define the treatment resource constraints that motivate the policy-relevant regimes. Let $q_k^{g_z}$ and $m_k^{g_z}$, respectively, be the multiplicative factors by which the population utilization of superior and inferior treatment units, respectively, are changed (relative to the unintervened world) in interval $k$ under regime \textit{$g_z$}. For regime $g_1$: 

\begin{align}
    & q_k^{g_1} = 1 \\
    & m_k^{g_1} = 0,
\end{align}

and for regime $g_0$:

\begin{align}
    & q_k^{g_1} = 1 \\
    & m_k^{g_1} = 1,
\end{align}

for all $k \in \{1, \dots, K\}$.     

In words, the constraints are specified such that: 1) the population utilization of suspected superior treatment units used under both regimes $g_1$ and $g_0$ in interval $k$ ($q_k^{g_1}\times P(B_k=1)$ and $q_k^{g_0}\times P(B_k=1)$) are equal to the number of superior treatment units actually used in the absence of intervention ($P(B_k=1)$); and 2) the population utilization of suspected inferior treatment units used under regime $g_1$ in interval $k$ ($m_k^{g_1}\times P(H_k=1)$) is set to 0 (corresponding to a policy in which the use of such treatment units is abolished), whereas the number of such units used under regime $g_0$ in interval $k$ ($m_k^{g_0}\times P(H_k=1)$) is maintained equal to the number of inferior treatment units actually used in the absence of intervention ($P(H_k=1)$, corresponding to a policy of no intervention). 
    
For notational convenience, we re-express the constraints for regime $g_1$ as

\begin{align}
    & P(B_k^{g_1+}=1) =  P(B_k=1) \label{eq; c1.1}\\
    & P(H_k^{g_1+}=1) = 0, \label{eq; c2.1}
\end{align}

and for regime $g_0$ as

\begin{align}
    & P(B_k^{g_0+}=1) =  P(B_k=1) \label{eq; c1.2}\\
    & P(H_k^{g_0+}=1) = P(H_k=1), \label{eq; c2.2}
\end{align}
for all $k \in \{1, \dots, K\}$. 
    
In general, we define regime $g_z$ to be a stochastic intervention on $H_k$ and  $B_k$ for all $k \in \{1,...,K\}$, where $z \in \{0,1\}$ indicates regime, such that intervention distributions are defined as follows

\begin{align}
    &  f_{B_k^{g_z+} \mid R_k^{g_z+}, \overline{L}_k^{g_z}}(1 \mid R_k, \overline{L}_k) 
        =      \alpha_{k}(z) \times f_{B_k \mid R_k, \overline{L}_k}(1 \mid R_k, \overline{L}_k) \label{eq: int Bk} \\
    & f_{H_k^{g_z+} \mid S_k^{g_z+}, \overline{L}_k^{g_z}}(1 \mid S_k, \overline{L}_k) 
        =      \beta_k(z) \times f_{H_k \mid S_k, \overline{L}_k}(1 \mid S_k, \overline{L}_k), \label{eq: int Hk}
\end{align}

with probability 1, where $f_{B_k \mid R_k, \overline{L}_k}(\cdot \mid \cdot)$ and $f_{B_k^{g_z+} \mid R_k^{g_z+}, \overline{L}_k^{g_z}}(\cdot \mid \cdot)$ are the probability density functions for receiving the suspected superior treatment unit under the observed data generating mechanism and under regime $g_z$, respectively, and likewise for $f_{H_k \mid S_k, \overline{L}_k}(\cdot \mid \cdot)$ and $f_{H_k^{g_z+} \mid S_k^{g_z+}, \overline{L}_k^{g_z}}(\cdot \mid \cdot)$, with respect to the suspected inferior treatment. Note that $R_k^{g_z+}=I\{Y_{k-1}^{g_z}= \overline{H}_{k-1}^{g_z+}=\overline{B}_{k-1}^{g_z+}=0\}$ and likewise $S_k^{g_z+}$ are defined in terms of post-intervention (as opposed to natural) treatment variables. Moreover, define $\alpha_k(z)$ and $\beta_k(z)$ as regime $z$-specific scaling functions that satisfy certain resource constraints,

\begin{align}
    &  \alpha_k(z) =\frac{P(B_k^{g_z+}=1)}{P(B_k^{g_z}=1)} \label{eq; alphak},\\ 
    &  \beta_k(z) = \frac{P(H_k^{g_z+}=1)}{P(H_k^{g_z}=1)} \label{eq; betak},
\end{align}

for all $k \in \{1,\dots,K\}$. By plugging in the resource constraints defined by expressions \eqref{eq; c1.1}-\eqref{eq; c2.2}, we have that

\begin{align}
    &  \alpha_k(1) =\frac{P(B_k=1)}{P(B_k^{g_1}=1)} \\
    &  \beta_k(1) = 0,
\end{align}

and

\begin{align}
    &  \alpha_k(0) =\frac{P(B_k=1)}{P(B_k^{g_0}=1)} \\
    &  \beta_k(0) = \frac{P(H_k=1)}{P(H_k^{g_0}=1)},
\end{align}
where we deliberately use the natural values of treatment, $H_k^{g_z}$ and $B_k^{g_z}$. We show in Lemma \ref{lemma; appA} of Appendix A that $\alpha_k(z)$ and $\beta_k(z)$ will always take some value between 0 an 1 under the constraints in \eqref{eq; c1.1}-\eqref{eq; c2.2}.

We remind the reader that regime $g_z$ involves \textit{proportionally-representative interventions} on the conditional likelihoods of treatment reception. Specifically regime $g_z$ assigns an eligible individual treatment in interval $k$ probability equal to $\alpha_k(z) \times P(B_k=1 \mid R_k$, that is, with probability equal to some constant ($\alpha_k(z)$) times the factual likelihood of that individual receiving that treatment given their covariate and treatment history. The proportionally-representative interventions are stochastic, because they randomly assign treatment according to some pre-specified, non-degenerate distribution \cite{young2014identification}. They are representative \cite{young2014identification}, because this distribution is chosen, for each individual, to be a function of the treatment distribution of the observed data generating mechanism, conditional on that individual's treatment and confounder history (e.g. $f_{B_k \mid R_k, \overline{L}_k}(\cdot \mid \cdot)$). Thus, proportionally-representative interventions are considered to be dynamic interventions, with respect to treatment and confounder history, in that the treatment assignment rule for each individual is a function of these variables. Proportionally-representative interventions constrain resources because $\alpha_k(z)$ and $\beta_k(z)$ are chosen specifically so that marginal treatment utilization will be equal in expectation to some value consistent with a set of pre-specified treatment limitations, as in expressions \eqref{eq; c1.1} and \eqref{eq; c1.2}. In Appendix B we prove that the resource constraints are satisfied under these interventions.  

The above expressions represent a subset of proportionally-representative interventions for limited resources: interventions in which treatments are either abolished or constrained such that marginal treatment utilization in the intervened world will be equal in expectation to that in the unintervened world. In general, proportionally-representative interventions could be employed to constrain resource utilization to any arbitrary level, including levels greater than those in the unintervened world, and we provide flexible definitions of these regimes in Appendix A. We also show in Appendix C that, when all treatment types are either abolished, or are assumed unlimited and provided to every individual with a particular covariate and treatment history, then -- for all covariate and treatment histories and all time points -- the joint counterfactual treatment density under such a regime is exactly equal to the analogous density under a particular traditional dynamic deterministic regime. This equivalence demonstrates that any dynamic deterministic regime can be understood as a special (and often unrealistic) case of a proportionally representative intervention, in which treatment resources are either abolished or assumed to be practically unlimited.

\section{Identification:} 
\label{sec: Identification}

To identify expected potential outcomes under regime \textit{$g_z$}, $\mathbb{E}[Y_k^{g_z}]$, for all $k \in \{1, \dots, K\}$ from observed data distributions, the following identification conditions are sufficient:

\subsection{Exchangeability 1}

\begin{align}
    \underline{Y}^{g_z}_t \independent I(B_t^{g_z}=b_t) \mid \overline{L}_t^{g_z}=\overline{l}_t, \overline{Y}_{t-1}^{g_z}=0, \overline{H}_{t-1}^{g_z}=\overline{h}_{t-1}, \overline{B}_{t-1}^{g_z}=\overline{b}_{t-1}\label{ex1B}
\end{align}

for $\{\overline{b}_{t}, \overline{l}_t, \overline{h}_{t-1}  \mid P(\overline{B}_{t}^{g_z+}=\overline{b}_{t},  \overline{L}_{t}^{g_z}=\overline{l}_{t}, \overline{Y}_{t-1}^{g_z}=0, \overline{H}_{t-1}^{g_z+}=\overline{h}_{t-1})>0\}$,  $t\in\{1,\dots,k\}$, and
\\

\begin{align}
    \underline{Y}^{g_z}_t \independent I(H_t^{g_z}=h_t) \mid \overline{B}_{t}^{g_z}=\overline{b}_{t}, \overline{L}_t^{g_z}=\overline{l}_t, \overline{Y}_{t-1}^{g_z}=0, \overline{H}_{t-1}^{g_z}=\overline{h}_{t-1}, \label{ex1H}
\end{align}

for $\{\overline{h}_{t}, \overline{b}_{t}, \overline{l}_{t}  \mid P(\overline{H}_{t}^{g_z+}=\overline{h}_{t},  \overline{B}_{t}^{g_z+}=\overline{b}_{t}, \overline{L}_{t}^{g_z}=\overline{l}_{t}, \overline{Y}_{t-1}^{g_z}=0)>0\}$, $t\in\{1,\dots,k\}$.
\\

These sequential exchangeability conditions are implied by the standard ``no unmeasured confounding'' for the outcomes $\underline{Y}_t$, with respect to past treatment.

\subsection{Exchangeability 2}

\begin{align}
    \underline{B}^{g_z}_t \independent I(B_{t-1}^{g_z}=b_{t-1}) \mid \overline{L}_{t-1}^{g_z}=\overline{l}_{t-1}, \overline{Y}_{t-2}^{g_z}=0, \overline{H}_{t-2}^{g_z}=\overline{h}_{t-2}, \overline{B}_{t-2}^{g_z}=\overline{b}_{t-2}, \label{ex2BB}
\end{align}

for $\{\overline{b}_{t-1}, \overline{l}_{t-1}, \overline{h}_{t-2}  \mid P(\overline{B}_{t-1}^{g_z+}=\overline{b}_{t-1}, \overline{L}_{t-1}^{g_z}=\overline{l}_{t-1}, \overline{Y}_{t-2}^{g_z}=0, \overline{H}_{t-2}^{g_z+}=\overline{h}_{t-2})>0\}$,  $t\in\{1,\dots,k\}$,
\\
and

\begin{align}
    \underline{B}^{g_z}_t \independent I(H_{t-1}^{g_z}=h_{t-1}) \mid \overline{B}_{t-1}^{g_z}=\overline{b}_{t-1}, \overline{L}_{t-1}^{g_z}=\overline{l}_{t-1}, \overline{Y}_{t-2}^{g_z}=0, \overline{H}_{t-2}^{g_z}=\overline{h}_{t-2}, \label{ex2BH}
\end{align}

for $\{\overline{h}_{t-1}, \overline{b}_{t-1}, \overline{l}_{t-1}  \mid P(\overline{H}_{t-1}^{g_z+}=\overline{h}_{t-1}, \overline{B}_{t-1}^{g_z+}=\overline{b}_{t-1}, \overline{L}_{t-1}^{g_z}=\overline{l}_{t-1}, \overline{Y}_{t-2}^{g_z}=0)>0\}$,  $t\in\{1,\dots,k\}$,
\\

These exchangeability conditions are implied by the assumption of ``no unmeasured confounding'' for natural treatments $\underline{B}^{g_z}_t$, with respect to past treatment, as in Young et al. \cite{young2014identification}. The reason why we do not need analogous exchangeability conditions for $\underline{H}^{g_z}_t$ is the specific regimes under consideration, $g_1$ and $g_0$, specified by the constraints of expressions \eqref{eq; c1.1}-\eqref{eq; c2.2}. For proportionally  representative interventions that arbitrarily constrain resources, additional exchangeability conditions for $\underline{H}^{g_z}_t$ are needed, as outlined in Appendix A. 

While we have noted that the above conditions are implied by standard ``no unmeasured confounding'' conditions, Exchangeabilities 1 and 2 are weaker because they are restricted to treatment levels within covariate and treatment histories that are plausible under the regime of interest, \textit{$g_z$}; Exchangeabilities 1 and 2 can together be interpreted as a time-varying generalization to the conditional exchangeability condition C2 in Haneuse and Rotnitzky, 2013 \cite{haneuseEstimationEffectInterventions2013}, who considered feasible interventions on an individual's surgical operating time that depended on that individual's operating time they would have received under the observed data generating mechanism. Haneuse and Rotnitzky noted that the usual ``no unmeasured confounding'' assumptions imply that any individual could have received a surgical operating time of any of the lengths considered within a particular covariate level, and that these surgical operating time were as good as randomized within this group. This assumption is unreasonable when some ranges of operating times are infeasible for a subpopulation defined by the particular covariate level. in the context of the applied example of section \ref{sec: Example}, we highlight why the stronger, traditional exchangeability condition is unreasonable in many settings with limited resources.  

Furthermore, note that Exchangeability 1 and Exchangeability 2 together are implied by the Exchangeability condition of Theorem 31 in Richardson and Robins, 2013 \cite{richardson2013single}. Thus, Exchangeabilities 1 and 2 are easy to check in a Single World Intervention Template (SWIT) for conditional d-separations between the nodes corresponding to $\underline{Y}^{g_z}_k$ and past natural treatment variables $B_k^{g_z}$ and $H_k^{g_z}$. We provide an example SWIT in which these hold in Figure \ref{fig: SWIT1}. 

\subsection{Consistency}

\begin{align}
    & \text{if }  \overline{B}_{t}=\overline{B}_{t}^{g_z+} \text{ and }  \overline{H}_{t}=\overline{H}_{t}^{g_z+} \nonumber\\
    & \text{then } Y_t = Y_t^{g_z}, L_{t+1} = L_{t+1}^{g_z} \text{, and } B_{t+1} = B^{g_z}_{t+1}
    \label{ass: consistency B and Y}  \text{, and: }
\end{align} 

\begin{align}
    & \text{if } \overline{H}_{t}=\overline{H}_{t}^{g_z+} \text{ and } \overline{B}_{t+1}=\overline{B}_{t+1}^{g_z+} \nonumber\\
    & \text{then } H_{t+1} = H^{g_z}_{t+1}
    \label{ass: consistency H},
\end{align}

for all $t\in\{0,\dots,k\}$. The consistency assumptions state that if an individual whose observed treatment history up until some interval  equals  their assigned treatment history under regime $g_z$, then the values of all future observed variables that are non-descendants of future treatments (minimally, the immediately subsequent outcome, covariates, and the natural value of subsequent treatment) are equal to the value they would naturally take had that individual actually followed regime $g_z$.

\subsection{Positivity}

\begin{align}
    & f_{R_t^{g_z+}, \overline{L}_t^{g_z}}(1, \overline{L}_t)>0\text{ and } f_{B_t^{g_z+} \mid R_t^{g_z+}, \overline{L}_t^{g_z}}(B_t \mid 1, \overline{L}_t)>0 \implies \nonumber  \\
    & \quad f_{B_t \mid R_t, \overline{L}_t}(B_t \mid 1, \overline{L}_t)>0\text{, w.p.1} \label{eq: positivityB},
\end{align}

and

\begin{align}
    & f_{S_t^{g_z+}, \overline{L}_t^{g_z}}(1, \overline{L}_t)>0 \text{ and } f_{H_t^{g_z+} \mid S_t^{g_z+}, \overline{L}_t^{g_z}}(H_t \mid 1, \overline{L}_t)>0 \implies \nonumber  \\
    &   \quad f_{H_t \mid S_k, \overline{L}_t}(H_t \mid 1, \overline{L}_t)>0\text{, w.p.1},  \label{eq: positivityH} 
\end{align}

for all $t \in \{1, \dots, t\}$. That is, if there exists in interval $t$ some treatment-eligible individuals with covariate history $\overline{l}_t$ who are assigned treatment $b_t$ under regime ${g_z}$, then there must exist individuals with the same covariate and treatment history under the observed data generating mechanism. Note that most proportionally-representative interventions guarantee this positivity condition, \eqref{eq: positivityB}, since by definition of the regimes in expressions \eqref{eq: int Bk} and \eqref{eq: int Hk},  $f_{B_t \mid R_t, \overline{L}_t}(\cdot \mid \cdot)=0$ implies $f_{B_t^{g_z+} \mid R_t^{g_z}, \overline{L}_t^{g_z}}(\cdot \mid \cdot) = 0$, and therefore lack of positivity would contradict the definition of the regimes.      

\subsection{Identification formulae} \label{subsec: gfom}

When conditions hold, we can identify $\mathbb{E}( Y^{g_z}_K )$ from the non-extended g-formula of Robins (1986) for $Y_K$, $f^{g_z}_{Y_K}(1)$,

\begin{align}
  f^{g_z}_{Y_K}(1) =
   & \sum_{\overline{l}_K} \sum_{\overline{h}_K} \sum_{\overline{b}_K} \sum_{k=1}^K P(Y_k=1 \mid \overline{H}_k=\overline{h}_k, \overline{B}_k=\overline{b}_k, \overline{L}_k=\overline{l}_k, Y_{k-1}=0) \label{eq; gformY} \\
    & \times \prod_{j=1}^{k} \Big\{f_{H_j^{g_z+} \mid \overline{B}_{j}^{g_z+}, \overline{L}_j^{g_z}, Y_{j-1}^{g_z}, \overline{H}_{j-1}^{g_z+}}(h_j \mid \overline{b}_{j}, \overline{l}_j, 0, \overline{h}_{j-1})\nonumber \\
    & \times f_{B_j^{g_z+} \mid \overline{L}_j^{g_z}, Y_{j-1}^{g_z}, \overline{H}_{j-1}^{g_z+}, \overline{B}_{j-1}^{g_z+}}(b_j \mid \overline{l}_j, 0, \overline{h}_{j-1}, \overline{b}_{j-1})\nonumber \\
    & \times P(L_j=l_j \mid Y_{j-1}=0, \overline{H}_{j-1}=\overline{h}_{j-1}, \overline{B}_{j-1}=\overline{b}_{j-1}, \overline{L}_{j-1}=\overline{l}_{j-1}) \nonumber\\
    & \times P(Y_{j-1}=0 \mid \overline{H}_{j-1}=\overline{h}_{j-1}, \overline{B}_{j-1}=\overline{b}_{j-1}, \overline{L}_{j-1}=\overline{l}_{j-1}, Y_{j-2}=0)\Big\} \nonumber,
\end{align}

where  

\begin{align}
    f_{B_j^{g_z+} \mid \overline{L}_j^{g_z}, Y_{j-1}^{g_z}, \overline{H}_{j-1}^{g_z+}, \overline{B}_{j-1}^{g_z+}}(b_j \mid \overline{l}_j, 0, \overline{h}_{j-1}, \overline{b}_{j-1}) =
    & \Big(\alpha_j(z) \times f_{B_j \mid \overline{L}_j, Y_{j-1}, \overline{H}_{j-1}, \overline{B}_{j-1}}(1 \mid \overline{l}_j, 0, \overline{h}_{j-1}, \overline{b}_{j-1})\Big)^{b_j} \nonumber \\
    \times &  \Big(1-\alpha_j(z) \times f_{B_j \mid \overline{L}_j, Y_{j-1},  \overline{H}_{j-1}, \overline{B}_{j-1}}(1 \mid \overline{l}_j, 0, 
\overline{h}_{j-1}, \overline{b}_{j-1})\Big)^{1-b_j} \nonumber .
\end{align}

and 

\begin{align}
    f_{H_j^{g_z+} \mid \overline{B}_{j}^{g_z+}, \overline{L}_j^{g_z}, Y_{j-1}^{g_z}, \overline{H}_{j-1}^{g_z+}}(h_j \mid \overline{b}_{j}, \overline{l}_j, 0, \overline{h}_{j-1}) =
    & \Big(\beta_j(z) \times f_{H_j \mid \overline{B}_{j}, \overline{L}_j, Y_{j-1}, \overline{H}_{j-1}}(1 \mid \overline{b}_{j}, \overline{l}_j, 0, \overline{h}_{j-1})\Big)^{h_j} \nonumber \\
    \times &  \Big(1-\beta_j(z) \times f_{H_j \mid \overline{B}_{j}, \overline{L}_j, Y_{j-1}, \overline{H}_{j-1}}(1 \mid \overline{b}_{j}, \overline{l}_j, 0, \overline{h}_{j-1})\Big)^{1-h_j} \nonumber .
\end{align}

Since all individuals are alive and untreated in interval 0, regardless of regime, we can trivially apply the consistency condition in expression \eqref{eq; c1.1} to find that $\alpha_1(z)$ is identified by $\frac{P(B_1^{g_z+}=1)}{P(B_1=1)}$. Thus $\alpha_j(z)$ and $\beta_{j-1}(z)$ are identified recursively for $j \in \{ 2, \dots, K+1 \}$, respectively, so that

\begin{align}
    \alpha_j(z) =
    \frac{P(B_j^{g_z+}=1)}{
        f^{g_z}_{B_j}(1)
    } \nonumber
\end{align}

and 

\begin{align}
    \beta_{j-1}(z) = 
    \frac{P(H_{j-1}^{g_z+}=1)}{
        f^{g_z}_{H_{j-1}}(1)
    }, \nonumber
\end{align}

where $P(B_j^{g_z+}=1)$ and $P(H_{j-1}^{g_z+}=1)$ are determined by the constraints, and $f^{g_z}_{B_j}(1)$ is the non-extended g-formula for $B_j$,

\begin{align}
    f^{g_z}_{B_j}(1) = 
    & \sum_{\overline{l}_j} P(B_j=1 \mid \overline{L}_j=\overline{l}_j, R_j=1) \label{eq; gformB} \\
    & \times \prod_{m=1}^{j} \Big\{  P(L_m=l_m \mid R_m=1, \overline{L}_{m-1}=\overline{l}_{m-1}) \nonumber\\
    & \times P(Y_{m-1}=0 \mid H_{m-1}=0, S_{m-1}=1, \overline{L}_{m-1}=\overline{l}_{m-1}) \nonumber \\
    & \times  \big(1-\beta_{m-1}(z) \times f_{H_{m-1} \mid S_{m-1}, \overline{L}_{m-1}}(1 \mid 1, \overline{l}_{m-1})\big) \nonumber \\
    & \times  \big(1-\alpha_{m-1}(z) \times f_{B_{m-1} \mid R_{m-1}, \overline{L}_{m-1}}(1 \mid 1, \overline{l}_{m-1})\big) \Big\} \nonumber,  
\end{align}

and $f^{g_z}_{H_j}(1)$ is the non-extended g-formula for $H_j$,

\begin{align}
    f^{g_z}_{H_j}(1)=
    & \sum_{\overline{l}_j} P(H_j=1 \mid \overline{L}_j=\overline{l}_j, S_j=1) \label{eq; gformH} \\
    & \times \prod_{m=1}^{j} \Big\{ \big(1-\alpha_{m}(z) \times f_{B_{m} \mid R_{m}, \overline{L}_{m}}(1 \mid 1, \overline{l}_{m})\big)  \nonumber \\
    & \times  P(L_m=l_m \mid R_m=1, \overline{L}_{m-1}=\overline{l}_{m-1}) \nonumber\\
    & \times P(Y_{m-1}=0 \mid H_{m-1}=0, S_{m-1}=1, \overline{L}_{m-1}=\overline{l}_{m-1}) \nonumber \\
    & \times  \big(1-\beta_{m-1}(z) \times f_{H_{m-1} \mid S_{m-1}, \overline{L}_{m-1}}(1 \mid 1, \overline{l}_{m-1})\big) \Big\}\nonumber.
\end{align}

Note that the sets of densities in $f^{g_z}_{B_j}(1)$ and $f^{g_z}_{H_j}(1)$, excluding the observed conditional densities of receiving the suspected superior, and suspected inferior treatment resource in interval $j$, that is, $P(B_j=1 \mid \overline{L}_j=\overline{l}_j, R_j=1)$, and $P(H_j=1 \mid \overline{L}_j=\overline{l}_j, S_j=1)$, respectively,  are strict subsets of the set of densities in  $f^{g_z}_{Y_K}(1)$. As such, $f^{g_z}_{B_j}(1)$ and $f^{g_z}_{H_j}(1)$ are simply truncated versions of the same g-formula as $f^{g_z}_{Y_K}(1)$ that identify distributions of different outcomes -- the natural values of treatment -- at time $j$.

Furthermore, note that when $z=0$, $P(B_1^{g_0+}=1)=P(B_1=1)$, then $\alpha_1(0)=1$. Thus, $f^{g_0}_{H_1}(1)=P(H_1)$. Since $P(H_1^{g_0+}=1)=P(H_1=1)$, then $\beta_1(0)=1$, as well. Since $P(B_k^{g_0+}=1)=P(B_k=1)$ and $P(H_k^{g_0+}=1)=P(H_k=1)$ for all $k \in \{1,\dots, K\}$, as in the definitions of the constraints in \eqref{eq; c1.2} and \eqref{eq; c2.2}, arguing iteratively, then $\alpha_k(0)=1$ and $\beta_k(0)=1$ for all $k \in \{1,\dots, K\}$, so that $\mathbb{E}[Y_K^{g_0}] = f^{g_z}_{Y_K}(1)=\mathbb{E}[Y_K]$.

We provide a proof of the above identification results in the more general case where treatment resources are arbitrarily constrained and subjects may be lost to follow-up, in Appendix D.

\subsubsection{Alternative representations of the g-formula}\
\\

The g-formula \eqref{eq; gformY} may be expressed in many ways, and its representation is arbitrary when computation is non-parametric. However, alternative representation will impact computational complexity and statistical properties when parametric estimation is chosen in high-dimensional settings with finite samples, as is often the case. A particular representation naturally motivates a class of inverse probability weighted (IPW) estimators that are easily computed with off-the-shelf software. Therefore we describe equivalent representations of g-formula in \eqref{eq; gformY} below, where we define $V$ to be any subset  of $L_1$, including possibly the emptyset, $\emptyset$.

\begin{align}
   f^{g_z}_{Y_K}(1) = & \sum_{v}\sum_{k=1}^K \lambda_{Y,k}^{g_z}(v) \prod_{j=1}^{k-1}[1-\lambda_{Y,j}^{g_z}(v)]f(v), \label{eq; altgformY}
\end{align}

where 

\begin{align}
    & \lambda_{Y,k}^{g_z}(V) = \frac{
        \mathbb{E}\big[Y_k(1-Y_{k-1})W_{H,k}^{g_z}W_{B,k}^{g_z} \mid V\big]
                    }{
        \mathbb{E}\big[(1-Y_{k-1})W_{H,k}^{g_z}W_{B,k}^{g_z}\mid V \big]           
                    } \label{eq; lambda}
\end{align}

and 

\begin{align}
        W_{B,k}^{g_z}=    \prod_{j=1}^{k}\frac{
                     \Big( \alpha_j(z) \times f_{B_{j} \mid R_{j}, \overline{L}_{j}}(1 \mid R_{j}, \overline{L}_{j})\Big)^{B_{j}}  
                    \times \Big(1- \alpha_j(z) \times f_{B_{j} \mid R_{j}, \overline{L}_{j}}(1 \mid R_{j}, \overline{L}_{j})\Big)^{1-B_j}
                }{
                f_{B_j \mid R_j, \overline{L}_j}(B_j \mid R_j, \overline{L}_j)
                } \label{eq; weight B}
\end{align}

and 

\begin{align}
        W_{H,k}^{g_z}=    \prod_{j=1}^{k}\frac{
                     \Big( \beta_j(z) \times f_{H_{j} \mid S_{j}, \overline{L}_{j}}(1 \mid S_{j}, \overline{L}_{j})\Big)^{H_{j}}  
                    \times \Big(1- \beta_j(z) \times f_{H_{j} \mid S_{j}, \overline{L}_{j}}(1 \mid S_{j}, \overline{L}_{j})\Big)^{1-H_j}
                }{
                f_{H_j \mid S_j, \overline{L}_j}(H_j \mid S_j, \overline{L}_j)
                } \label{eq; weight H},
\end{align}

for $k \in \{0, \dots, K\}$. We provide a proof of the equivalence between the representations in \eqref{eq; gformY} and \eqref{eq; altgformY} in Appendix E, which only depends on the positivity assumption. 

Similarly, the g-formula of expression \eqref{eq; gformB} may be expressed as

\begin{align}
   f^{g_z}_{B_j}(1) =  & \pi_{B,j}^{g_z}, \label{eq; altgformB}
\end{align}

where

\begin{align}
    & \pi_{B,j}^{g_z} = \mathbb{E}\big[B_jW_{H,j-1}^{g_z}W_{B,j-1}^{g_z}\big],
\end{align}

and the g-formula of expression \eqref{eq; gformH} as

\begin{align}
   f^{g_z}_{H_j}(1) =  & \pi_{H,j}^{g_z}, \label{eq; altgformH}
\end{align}

where

\begin{align}
    & \pi_{H,j}^{g_z} = \mathbb{E}\big[H_jW_{B,j}^{g_z}W_{H,j-1}^{g_z}\big].
\end{align}

Here, $\lambda_{Y,k}^{g_z}$ represents the discrete time hazard for death in interval $k$, and $\pi_{B,k}^{g_z}$ and $\pi_{H,k}^{g_z}$ are the marginal probabilities of receiving a suspected superior treatment and suspected inferior treatment unit, respectively, under the regime ${g_z}$, characterized by the proportionally-representative interventions of \eqref{eq: int Bk} and \eqref{eq: int Hk} that realistically constrain treatment resources. These alternative expressions will motivate our presentation of a particular marginal structural model (MSM) \cite{hernanMarginalStructuralModels2000}, introduced in Section \ref{sec: ipw}.

\section{Inverse Probability Weighted Estimation of Risk under Proportionally Representative Interventions}
\label{sec: ipw}

In low-dimensional settings, large samples would allow estimation of $\mathbb{E}  ( Y^{g_z}_K )$ by non-parametrically estimating the components of the g-formulae in expressions \eqref{eq; gformY} - \eqref{eq; gformH}, or equivalently the alternative formulations in \eqref{eq; altgformY}, \eqref{eq; altgformB} and \eqref{eq; altgformH}. However, in high dimensional settings, e.g.\ when $L_k$ takes many levels and/or when $K$ is large, non-parametric estimation will often be practically infeasible. One approach to overcoming this issue is to supplement the non-parametric identification assumptions in Section \ref{sec: Identification} with restrictions on the observed data-generating mechanism in the form of parametric modelling assumptions such that parametric estimators can be used.

In particular, it is possible to impose parametric modelling assumptions on some or all of the components of expression \eqref{eq; gformY} ($f^{g_z}_{Y_K}(1)$), e.g.\ using the parametric g-formula estimator of Robins \cite{robins1986new}, which would likely include modelling assumptions for the $K$ conditional treatment densities that comprise the proportionally-representative intervention distributions for regime ${g_z}$, \eqref{eq: int Bk} and \eqref{eq: int Hk}. These estimators are often computed via Monte Carlo simulations using estimated model coefficients to approximate the estimator's high-dimensional integral/sum, but are prone to bias from model-misspecification due to the large number of conditional densities for which parametric assumptions are made. 

Alternatively, we might impose modelling assumptions on $\lambda_{Y,k}^{g_z}(V)$ of expression \eqref{eq; altgformY}, corresponding to an MSM as in Young 2018 \cite{youngInverseProbabilityWeighted2018}, and on the components of $W_{B,k}^{g_z}$ and $W_{H,k}^{g_z}$, defined in expressions \eqref{eq; weight B} and \eqref{eq; weight H}. Without knowledge about the structural functions of the observed data generating mechanism, as is almost always the case, this approach would usually be more desirable than a parametric g-formula estimator for proportionally representative interventions, since the union of the sets of components of $W_{B,K}^{g_z}$ and $W_{H,k}^{g_z}$, respectively, defined by the $2\times K$ conditional treatment densities, is a strict subset of the set of components of the g-formula for which modelling assumptions must be made in the parametric g-formula. While MSMs are also subject to bias from model misspecification and can be inefficient, the strictly smaller number of models that must be specified, and relative computational ease of such an estimator, involving only a triplet of simple regression models (for outcome, and the pair of treatments), provide specific advantages. 

\subsection{Marginal Structural Models for Proportionally Representative Interventions}\label{subsec; MSM}\ 
\\

A MSM parameterizes the contrast of marginal means under different regimes indexed by $z$. The joint consideration of laws under different regimes requires joint consideration of pseudo-data from generating mechanisms corresponding to these regimes. When a single observed dataset can be partitioned into subsets following the different regimes under consideration, then the joint consideration naturally follows by re-weighting each individual in the single observed dataset according to the regime under which their observed treatment was consistent. Alternatively, in the setting of proportionally-representative interventions, many individuals receive treatment consistent with both regimes $g_0$ and $g_1$; that is, there is a non-empty intersection between the terms in $f^{g_0}_{Y_K}(1)$ and $f^{g_1}_{Y_K}(1)$ (the g-formulae for $Y_K$ under regimes $g_0$ and $g_1$, respectively). A general approach that allows joint consideration of pseudo-data from worlds under both of these regimes, is to consider perfectly cloned copies of the observed data (one for each regime under consideration), and introduce an administrative variable, $Z$, which indexes the particular regime that the cloned copy corresponds to. We use this approach for the remainder of the manuscript, and therefore we re-express $\lambda_{Y,k}^{g_z}(V)$, with respect to the cloned dataset,as

\begin{align}
    & \lambda_{Y,k}^{g_z}(V) = \lambda_{Y,k}(Z, V) = \frac{
        \mathbb{E}\big[Y_k(1-Y_{k-1})W_{H,k}(Z)W_{B,k}(Z) \mid V, Z\big]
                    }{
        \mathbb{E}\big[(1-Y_{k-1})W_{H,k}(Z)W_{B,k}(Z)\mid V, Z \big]           
                    },  \label{eq; altgformYZ}
\end{align}

where 

\begin{align}
        W_{B,k}(Z) =  \prod_{j=1}^{k}\frac{
                     \Big( \alpha_j(Z) \times f_{B_{j} \mid R_{j}, \overline{L}_{j}}(1 \mid R_{j}, \overline{L}_{j})\Big)^{B_{j}}  
                    \times \Big(1- \alpha_j(Z) \times f_{B_{j} \mid R_{j}, \overline{L}_{j}}(1 \mid R_{j}, \overline{L}_{j})\Big)^{1-B_j}
                }{
                f_{B_j \mid R_j, \overline{L}_j}(B_j \mid R_j, \overline{L}_j)
                } \label{eq; weight BZ}
\end{align}

and

\begin{align}
        W_{H,k}(Z)=    \prod_{j=1}^{k}\frac{
                     \Big( \beta_j(Z) \times f_{H_{j} \mid S_{j}, \overline{L}_{j}}(1 \mid S_{j}, \overline{L}_{j})\Big)^{H_{j}}  
                    \times \Big(1- \beta_j(Z) \times f_{H_{j} \mid S_{j}, \overline{L}_{j}}(1 \mid S_{j}, \overline{L}_{j})\Big)^{1-H_j}
                }{
                f_{H_j \mid S_j, \overline{L}_j}(H_j \mid S_j, \overline{L}_j)
                } \label{eq; weight HZ}.
\end{align}

Similarly, we reexpress $\pi_{B,j}^{g_z}$ and $\pi_{H,j}^{g_z}$, respectively, as

\begin{align}
    & \pi_{B,j}(Z) = \mathbb{E}\big[B_jW_{H,j-1}(Z)W_{B,j-1}(Z)\big] \label{eq; altgformBZ}
\end{align}

and 

\begin{align}
    & \pi_{H,j}(Z) = \mathbb{E}\big[H_jW_{B,j}(Z)W_{H,j-1}(Z)\big]. \label{eq; altgformHZ}
\end{align}

Now, assume that there exists a real-valued parameter $\psi^{*}$, such that, for a user-specified $V$, the following holds for all values of $k$, $Z$, and $V$,

\begin{align}
    & h\{\gamma(k,Z,V; \psi)\}  = \frac{
        \mathbb{E}\big[Y_k(1-Y_{k-1})W_{H,k}(Z)W_{B,k}(Z) \mid V, Z\big]
                    }{
        \mathbb{E}\big[(1-Y_{k-1})W_{H,k}(Z)W_{B,k}(Z)\mid V, Z \big]           
                    }, \label{eq; MSM}
\end{align}
where $\psi=\psi^{*}$; $h\{\cdot\}$ is a known link function constrained between 0 and 1 (e.g.\ the logit or probit functions); and $\gamma(\cdot)$ is defined to be some function of $k$, $Z$, and $V$ that is differentiable with respect to $\psi$ and is not a function of $Z$ when $\psi=0$ (so that the sharp null hypothesis of no effect of regime $g_1$ versus $g_0$ on survival at any time $k$, $Y_k$, along with exchangeability conditions \eqref{ex1B}-\eqref{ex2BH}, implies $\psi^*=0$).  Then, \eqref{eq; MSM} is an MSM for the discrete-time hazard of death at $k$, conditional on V, and under regimes indexed by $Z$, defined by the proportionally-representative interventions  of \eqref{eq: int Bk} and \eqref{eq: int Hk} that realistically constrain treatment resources.

Given the MSM in \eqref{eq; MSM} holds, \eqref{eq; altgformY} may be re-written as 

\begin{align}
   & \sum_{v}\sum_{k=1}^K h\{\gamma(k,z,v; \psi)\} \prod_{j=1}^{k-1}[1-h\{\gamma(k,z,v; \psi)\}]f(v), \label{eq; altgformYMSM}
\end{align},

when $\psi=\psi^*$.

\subsection{Inverse Probability Weighted Estimation} \label{subsec: IPWest} \
\\ 

Let $\hat{\psi}$ be the solution to the estimating equation

\begin{align}
   & \sum_{i=1}^n\sum_{z}\sum_{k=1}^K U_{i,k}(\psi, \hat{\eta}_B, \hat{\eta}_H)=0, \label{eq; esteq}
\end{align},

with respect to $\psi$, where 

\begin{align}
   U_k(\psi, \hat{\eta}_B, \hat{\eta}_H) = & [Y_k-h\{\gamma(k,Z,V; \psi)\}] \label{eq; esteqU}\\
                           & \times (1-Y_{k-1})W_{B,k}(Z,\hat{\eta}_B)W_{H,k}(Z, \hat{\eta}_H). \nonumber 
\end{align}

Here, let

\begin{align}
        W_{B,k}(Z, \hat{\eta}) =  \prod_{j=1}^{k}\frac{
                     \Big( \alpha_j(Z, \hat{\eta}) \times f_{B_{j} \mid R_{j}, \overline{L}_{j}}(1 \mid R_{j}, \overline{L}_{j}; \hat{\eta}_B)\Big)^{B_{j}}  
                    \times \Big(1- \alpha_j(Z, \hat{\eta}) \times f_{B_{j} \mid R_{j}, \overline{L}_{j}}(1 \mid R_{j}, \overline{L}_{j}; \hat{\eta}_B)\Big)^{1-B_j}
                }{
                f_{B_j \mid R_j, \overline{L}_j}(B_j \mid R_j, \overline{L}_j; \hat{\eta}_B)
                } \label{eq; estweightB},
\end{align}

and 

\begin{align}
        W_{H,k}(Z, \hat{\eta})=    \prod_{j=1}^{k}\frac{
                     \Big( \beta_j(Z, \hat{\eta}) \times f_{H_{j} \mid S_{j}, \overline{L}_{j}}(1 \mid S_{j}, \overline{L}_{j}; \hat{\eta}_H)\Big)^{H_{j}}  
                    \times \Big(1- \beta_j(Z, \hat{\eta}) \times f_{H_{j} \mid S_{j}, \overline{L}_{j}}(1 \mid S_{j}, \overline{L}_{j}; \hat{\eta}_H)\Big)^{1-H_j}
                }{
                f_{H_j \mid S_j, \overline{L}_j}(H_j \mid S_j, \overline{L}_j; \hat{\eta}_H)
                }, \label{eq; estweightH}
\end{align}

where $\eta \equiv \{\eta_B, \eta_H\}$, $f_{B_j \mid R_j, \overline{L}_j}(B_j \mid R_j, \overline{L}_j; \eta_B)$ and $f_{H_j \mid S_j, \overline{L}_j}(H_j \mid S_j, \overline{L}_j; \eta_H)$ are models for the denominators of \eqref{eq; estweightB} and \eqref{eq; estweightH}, and $\hat{\eta}_B$ and $\hat{\eta}_H$ are the MLEs of $\eta_B$ and $\eta_H$, respectively.

Here, $\alpha_1(z,\hat{\eta})$ is estimated  by

\begin{align}
    \frac{P(B_1^{g_Z+}=1)}{\frac{1}{n}\sum_{i=1}^nB_{i,1}} \label{eq; alpha1est},
\end{align}

and $\alpha_j(Z, \hat{\eta})$ and $\beta_{j-1}(Z, \hat{\eta})$ are obtained recursively for $j = 2, \dots, K+1$, respectively, from

\begin{align}
    \alpha_j(Z, \hat{\eta}) =
    \frac{P(B_j^{g_Z+}=1)}{\hat{\pi}_{B,j}(Z, \hat{\eta})
    } \label{eq; alphajest}
\end{align}

and 

\begin{align}
    \beta_{j}(Z, \hat{\eta}) = 
    \frac{P(H_j^{g_Z+}=1)}{
        \hat{\pi}_{B,j}(Z, \hat{\eta})
    } \label{eq; betajest}
\end{align}

where 

\begin{align}
    & \hat{\pi}_{B,j}(Z, \hat{\eta}) = \frac{1}{n}\sum_{i=1}^n\big[B_{i,j}W_{H,i,j-1}(Z_i, \hat{\eta})W_{B, i, j-1}(Z_i, \hat{\eta})\big]
\end{align}

and 

\begin{align}
    & \hat{\pi}_{H,j}(Z, \hat{\eta}) = \frac{1}{n}\sum_{i=1}^n\big[H_{i,j}W_{B, i,j}(Z_i, \hat{\eta})W_{H, i,j-1}(Z_i, \hat{\eta})\big].
\end{align}

Here, $P(B_j^{g_Z+}=1)$ and $P(H_j^{g_Z+}=1)$ are defined by the user-specified proportionally-representative interventions, $\pi_{B,j}(Z, \eta)$ and $\pi_{H,j}(Z, \eta)$ are models for $\pi_{B,j}(Z)$ and $\pi_{H,j}(Z)$, and $\hat{\pi}_{B,j}(Z, \hat{\eta})$ and $\hat{\pi}_{H,j}(Z, \hat{\eta})$ are their MLEs.

Following Robins (1999) \cite{robins2000marginal}, if (i) the MSM of \eqref{eq; MSM} is correctly specified; and (ii), the models $f_{B_j \mid R_j, \overline{L}_j}(B_j \mid R_j, \overline{L}_j; \eta_B)$ and $f_{H_j \mid S_j, \overline{L}_j}(H_j \mid S_j, \overline{L}_j; \eta_H)$ of the denominators for \eqref{eq; estweightB} and \eqref{eq; estweightH} are correctly specified, then we have 

\begin{align}
   \mathbb{E}[U_k(\psi^*, \eta_B^*, \eta_H^*)] = 0 
\end{align}

for all $k$, with $\eta_B^*$ and $\eta_H^*$ the true values of $\eta_B$ and $\eta_H$ and the IPW estimator $\hat{\psi}$ consistent and asymptotically normal for $\psi^*$.

Here, we assume pooled logistic models for $h\{\gamma(k,Z,V; \psi)\}$, $f_{B_k \mid R_k, \overline{L}_j}(B_k \mid 1, \overline{L}_k; \eta_B)$ and $f_{H_k \mid S_k, \overline{L}_j}(H_k \mid 1, \overline{L}_k; \eta_H)$, that is,

\begin{align}
   h\{\gamma(k,Z,V; \psi)\} = \text{expit}\{\gamma(k,Z,V; \psi)\} \label{eq; pooledlogitY}
\end{align} 

and

\begin{align}
   f_{B_k \mid R_k, \overline{L}_k}(1 \mid 1, \overline{L}_k; \eta_B) = \text{expit}\{\phi_B(k, \overline{L}_k; \eta_B)\} \label{eq; pooledlogitB}
\end{align}

and

\begin{align}
   f_{H_k \mid S_k, \overline{L}_k}(1 \mid 1, \overline{L}_k; \eta_H) = \text{expit}\{\phi_H(k, \overline{L}_k; \eta_H)\} \label{eq; pooledlogitH}
\end{align}

with $\phi_B$ and $\phi_H$ specified functions of $(k, \overline{L}_k)$, differentiable with respect to $\eta_B$ and $\eta_H$, respectively, and $\text{expit}(\cdot)=\frac{exp(\cdot)}{1+exp(\cdot)}$. Note that $f_{B_k \mid R_k, \overline{L}_j}(0 \mid 0, \overline{L}_k; \eta_B) = f_{H_k \mid S_k, \overline{L}_j}(0 \mid 0, \overline{L}_k; \eta_B) = 1$ by definition (that is, previously treated or deceased individuals will be untreated in interval $k$, with probability 1).

A solution to \eqref{eq; esteq} is obtained through the following algorithm applied to a cloned subject-interval dataset with administrative variable Z, constructed such that each subject will have $K^*$ lines indexed by $k=1,\dots, K^*$, where $K^*=K$ if $Y_K=0$, else $K^*=\text{min}(\{j \mid Y_j=1\})$, so that $K^*=K$ when a subject is alive at the end of follow-up or equals the interval number during which a subject dies during follow-up. 

\subsubsection{IPW estimation algorithm for $\psi$} \label{subsubsec: IPW alg}

\begin{enumerate}
    \item [1.] Using subject-interval records with $Z=1$ and $R_k=1$, obtain $\hat{\eta}_B$ by fitting pooled logistic regression model \eqref{eq; pooledlogitB} with dependent variable $B_k$ and independent variables a specified function of $k=0,\dots, K$ and $\overline{L}_k$, corresponding to the choice of $\phi_B(\cdot)$.
    \item [2.] Similarly, using subject-interval records with $Z=1$ and $S_k=1$, obtain $\hat{\eta}_H$ by fitting a pooled logistic regression model \eqref{eq; pooledlogitH} with dependent variable $H_k$ and independent variables a specified function of $k=0,\dots, K$ and $\overline{L}_k$, corresponding to the choice of $\phi_H(\cdot)$.
    \item[3.] Set $\alpha_{0}(z, \hat{\eta})$ and $\beta_{0}(z, \hat{\eta})$ to 1. Using subject-interval records with $Z=1$, obtain $\alpha_1(1, \hat{\eta})$ by evaluating the ratio in \eqref{eq; alpha1est}: dividing $P(B_1^{g_1+}=1)$ (defined by the intervention), by the proportion of individuals with $B_1=1$, $\frac{1}{n}\sum_{i=1}^nB_{i,1}$. Likewise, using subject-interval records with $Z=0$, obtain $\alpha_1(0, \hat{\eta})$.
    \item[4.] For each subject's line 1, attach the suspected-superior treatment weight, $W_{B,1}$, calculated as 
    
    $$\frac{
                     \Big( \alpha_1(Z, \hat{\eta}) \times \text{expit}\{\phi_B(1, \overline{L}_1; \hat{\eta_B})\}\}\Big)^{B_{1}}  
                    \times \Big(1- \alpha_1(Z, \hat{\eta}) \times \text{expit}\{\phi_B(1, \overline{L}_1; \hat{\eta_B})\}\Big)^{1-B_1}
                }{
                \Big(\text{expit}\{\phi_B(1, \overline{L}_1; \hat{\eta_B})\}\Big)^{B_{1}}  
                    \times \Big(1-  \text{expit}\{\phi_B(1, \overline{L}_1; \hat{\eta_B})\}\Big)^{1-B_1}
                }$$
    \item[5.] Using subject-interval records with $Z=1$, obtain $\beta_1(1, \hat{\eta})$ by evaluating the ratio in \eqref{eq; betajest}: dividing $P(H_1^{g_1+}=1)$ (defined by the intervention), by the \textit{weighted} proportion of individuals with $H_1=1$, $\frac{1}{n}\sum_{i=1}^nH_{i,1}W_{B,i,1}$. Likewise, obtain $\beta_1(0, \hat{\eta})$.
    \item[6.] For each subject's line 1, attach the suspected-inferior treatment weight, $W_{H,1}$, calculated as
    
    $$\Bigg[\frac{
                     \Big( \beta_1(Z, \hat{\eta}) \times \text{expit}\{\phi_H(1, \overline{L}_1; \hat{\eta_H})\}\Big)^{H_{1}}  
                    \times \Big(1- \alpha_1(Z, \hat{\eta}) \times \text{expit}\{\phi_H(1, \overline{L}_1; \hat{\eta_H})\}\Big)^{1-H_1}
                }{
                \Big(\text{expit}\{\phi_H(1, \overline{L}_1; \hat{\eta_H})\}\Big)^{H_{1}}  
                    \times \Big(1-  \text{expit}\{\phi_H(1, \overline{L}_1; \hat{\eta_H})\}\Big)^{1-H_1}
                }\Bigg]^{S_1}$$
    
    \item[7.] For $z\in\{0, 1\}$, iterate from $k=2,\dots K$:
    \begin{enumerate}
       \item [7.1.] Using subject-interval records on line $k$ with $Z=z$, obtain $\alpha_{k}(z, \hat{\eta})$ by evaluating the ratio in \eqref{eq; alphajest}, dividing $P(B_{k}^{g_z+}=1)$ (defined by the intervention) by the weighted proportion of individuals with $B_{k}=1$, $\frac{1}{n}\sum_{i=1}^n\big[B_{i,k}W_{B, i,k-1}W_{H, i,k-1}\big]$.
       \item[7.2.] Using subject-interval records on line $k$ with $Z=z$, attach the suspected-superior treatment weight, $W_{B,k}$, calculated as
       
        $$\prod_{j=1}^k\Bigg[\frac{
                     \Big( \alpha_j(z, \hat{\eta}) \times \text{expit}\{\phi_B(j, \overline{L}_j; \hat{\eta_B})\}\Big)^{B_{j}}  
                    \times \Big(1- \alpha_j(z, \hat{\eta}) \times \text{expit}\{\phi_B(j, \overline{L}_j; \hat{\eta_B})\}\Big)^{1-B_j}
                }{
                \Big(\text{expit}\{\phi_B(j, \overline{L}_j; \hat{\eta_B})\}\Big)^{B_{j}}  
                    \times \Big(1-  \text{expit}\{\phi_B(j, \overline{L}_j; \hat{\eta_B})\}\Big)^{1-B_j}
                }\Bigg]^{R_j}$$  
                
        \item [7.3.] Using subject-interval records on line $k$ with $Z=z$, obtain $\beta_{k}(z, \hat{\eta})$ by evaluating the ratio in \eqref{eq; betajest}, dividing $P(H_{k}^{g_z+}=1)$ (defined by the intervention) by the weighted proportion proportion of individuals with $H_{k}=1$, $\frac{1}{n}\sum_{i=1}^n\big[H_{i,k}W_{B, i,k}W_{H, i,k-1}\big]$.
        \item[7.4.] Using subject-interval records on line $k$ with $Z=z$, attach the suspected-superior treatment weight, $W_{H,k}$, calculated as:  
            $$\prod_{j=1}^k\Bigg[\frac{
                     \Big( \beta_j(z, \hat{\eta}) \times \text{expit}\{\phi_H(j, \overline{L}_j; \hat{\eta_H})\}\Big)^{H_{j}}  
                    \times \Big(1- \alpha_j(z, \hat{\eta}) \times \text{expit}\{\phi_H(j, \overline{L}_j; \hat{\eta_H})\}\Big)^{1-H_j}
                }{
                \Big(\text{expit}\{\phi_H(j, \overline{L}_j; \hat{\eta_H})\}\Big)^{H_{j}}  
                    \times \Big(1-  \text{expit}\{\phi_H(j, \overline{L}_j; \hat{\eta_H})\}\Big)^{1-H_j}
                }\Bigg]^{S_j}$$         
    \end{enumerate}
\item[8.] Using all subject-interval records in the cloned dataset, obtain $\hat{\psi}$ by fitting a \textit{weighted} pooled logistic regression model, with weights $W_{B,k}$ and $W_{H,k}$ defined in the previous steps, dependent variable $Y_k$ and independent variables a specified function of $k=1,\dots,K$ and $(Z, V)$ corresponding to the choice of $\gamma(\cdot)$.
\end{enumerate}\
\\

Our final IPW estimate of the g-formula for the risk of death by $K$ under regime $g_z$, $f^{g_z}_{Y_K}(1)$ defined by the proportionally-representative interventions of \eqref{eq: int Bk} and \eqref{eq: int Hk} that constrain resources can then be obtained by the plug-in estimator

\begin{align}
   & \sum_{i=1}^{n}\sum_{k=1}^K \text{expit}\{\gamma(k,z,V_i; \hat{\psi})\} \prod_{j=1}^{k-1}[1-\text{expit}\{\gamma(k,z,V_i; \hat{\psi})\}] \label{eq; plugin}.
\end{align}

We provide extensions of the above algorithm to settings with arbitrary  resource constraints and loss-to-follow-up in Appendix F.

\section{Data Analysis}
\label{sec: Example}

Patients in need of a liver transplant can only receive a transplant if a suitable organ is available. The treatment (receive a transplant) is clearly a limited resource and the constraint imposed by organ availability should be satisfied in a policy-relevant study. Although the proportion of organ donors in the U.S. has risen considerably over the past two decades \cite{abaraCharacteristicsDeceasedSolid2019}, many of these organs are classified as ‘increased risk’ organs because their utilization might result in a higher probability of unintended transmission of HIV, hepatitis B, and/or hepatitis C to recipients \cite{seemPHSGuidelineReducing2013}.  Therefore, policy makers might be interested in learning how a restrictive policy of using only ‘standard risk’ organs (a suspected superior treatment resource) would impact survival of wait-listed patients compared with the current practice of using both `standard risk’ and `increased risk' organs (a suspected inferior treatment resource). Given the limited number of total livers available for transplantation, implementation of any restrictive policy will necessarily increase the length of time that candidates wait for a transplant. So, in other words, a first question of interest is whether increasing the time candidates are waiting for a standard-risk organ, and therefore the number of candidates who die while on the waiting list, is outweighed by a suspected longer survival among those who ultimately receive such a transplant. If this is false, policy makers might want to consider increasing utilization of increased risk organs. Already, some clinicians hesitate to use ‘increased risk’ organs \cite{kumarSurveyIncreasedInfectious2016,kucirkaProviderUtilizationHighRisk2009,volkPHSIncreasedRisk2017}, so it is possible that the system could increase the utilization of these organs by simply discarding less organs. 
Policy makers might be interested in learning how a policy that increases the utilization of ‘increased risk’ organs (a suspected inferior treatment resource) would impact survival of wait-listed patients compared with the current practice. So, in other words, a second question of interest is whether decreasing the time candidates wait for a standard-risk organ is outweighed by a suspected shorter survival among those who ultimately receive ‘increased risk’ organ transplants.

Motivated by this policy relevant question, we used data from the Scientific Registry of Transplant Recipients (SRTR) to study the causal effect of different transplantation policies. The SRTR data system is a repository of information for all candidates ever added to the United States Organ Procurement and Transplantation Network (OPTN) waiting list. The data were restricted to those from individuals aged 18 or older with no prior history of liver transplantation, who were eligible for liver transplantation and were added to the OPTN waiting list to receive a liver organ between 2005 to 2015, and were followed until  death,  loss  to  follow-up  (as  reported  by  individual transplant programs), or May 31st, 2016, whichever occurs first. The SRTR includes  data on wait-list candidate mortality via linkage to the National Death Index \cite{kim2019optn}.  Over the study period, n=93,812 transplant candidates were added to the waitlist, of whom 51,322 received livers from deceased donors. Data were coarsened into discrete 30-day intervals, where k=1 corresponds to an individual's first 30-day interval upon entering the waitlist. Data are used to estimate the 10-year ($Y_K$ where $K=120$) cumulative incidence of death under the following regimes:

\begin{enumerate}
    \item [$g_0.$] Current practice: ``standard risk'' and ``increased risk'' organs are utilized at current levels,
    \item [$g_1.$] Restrictive practice: ``standard risk'' organs are utilized at current levels and ``increased risk'' organ utilization is abolished,
    \item [$g_2.$] Expansive practice A: ``standard risk'' organs are utilized at current levels and ``increased risk'' organ utilization is increased by $25\%$,
    \item [$g_3.$] Expansive practice B: ``standard risk'' organs are utilized at current levels and ``increased risk'' organ utilization is increased by $50\%$,
\end{enumerate}

where each of these regimes targets outcomes under the ``natural course'', that is, under hypothetical interventions that abolish censoring, and for simplicity, under interventions that abolish utilization of transplants other than ``standard risk'' or ``increased risk'' organs (organs from living donors,cardiac-death donors, donors who are not HIV, hepatitis C, and hepatitis B seronegative, or donors with unknown risk status). We used the IPW estimation algorithm described in Section \ref{subsubsec: IPW alg}, adapted for censoring weights and weights for ``other'' transplant types, as in Appendix F.  

We defined $B_k$ to indicate reception of a ``standard risk'' organ, and $H_k$ to indicate reception of an ``increased risk'' organ, in interval $k$. We defined $V$ to be the empty set, $\emptyset$. For $k = 2,\dots, K$, interval $k$ confounders $L_k$ included waiting-list priority in the form of model for end-stage liver disease [MELD] score, MELD score exception, and urgent-need status. Interval 1 confounders $L_1$ additionally included year of listing to the waiting list, gender, race, age, height, weight,  willingness to accept a less optimal organ (i.e. a liver segment, a organ from an incompatible blood type donor, or a donor with hepatitis B or C), need for life support, functional status, primary diagnosis leading to liver failure, history of complications or procedures related to liver failure (i.e.\ spontaneous bacterial peritonitis, portal vein thrombosis, transjugular intrahepatic portosystemic shunt). 

Regimes are comprised of proportionally representative interventions that assign treatment according to the general expressions in Appendix A (for which expressions in \eqref{eq: int Bk} and \eqref{eq: int Hk} are a special case), and limit treatment resource utilization according to the following constraints, for $k=1,\dots K$:

\begin{enumerate}
    \item [$g_0.$] Current practice: $P(B_k^{g_0+}=1)=P(B_k=1)$ and $P(H_k^{g_0+}=1)=P(H_k=1)$,
    \item [$g_1.$] Restrictive practice: $P(B_k^{g_1+}=1)=P(B_k=1)$ and $P(H_k^{g_1+}=1)=0$,
    \item [$g_2.$] Expansive practice A: $P(B_k^{g_2+}=1)=P(B_k=1)$ and $P(H_k^{g_0+}=1)=\text{min}\big(1.25 \times P(H_k=1), P(S_k^{g_2}=1)\big)$,
    \item [$g_3.$] Expansive practice B: $P(B_k^{g_3+}=1)=P(B_k=1)$ and $P(H_k^{g_3+}=1)=\text{min}\big(1.5 \times P(H_k=1), P(S_k^{g_3}=1)\big)$.
\end{enumerate}

All regimes, $g_z$ included hypothetical interventions to abolish censoring - that is, we targeted potential outcomes under proportionally representative interventions that would occur during the natural course where no individual was censored. We discuss extensions to censoring and resulting procedures in appendices D and F.

\subsection{Marginal Structural Models}

We assumed the following functional form for $\gamma$ of the MSM in \eqref{eq; MSM}. Specifically,

\begin{align}
   \gamma(k,Z,V; \psi) = \psi_0 + \psi_1^Tg(k) + \psi_2Z + \psi_3Zk \label{eq; gammaspec},
\end{align} 

and we specify the weight models in Appendix G. 

\subsection{Results}

The utilization of ‘standard risk’ and ‘increased risk’ organs under the hypothetical treatment regimes over a follow-up period of 10 years are shown in figures \ref{fig: graftuse1} and \ref{fig: graftuse23}. The utilization of `standard risk' organs is identical in the weighted data under all treatment strategies, which visually confirms that the treatment resource constraints are satisfied.

The estimated 10-year cumulative incidence of death is 59.3$\%$ (95$\%$ confidence interval [CI]: 58.8 to 59.9$\%$) under regime $g_0$, corresponding to the current practice of using both ‘standard risk’ and ‘increased risk’ organs, in contrast to 64.9$\%$ (95$\%$ CI: 62.7 to 67.9$\%$) under regime $g_1$, corresponding to the restrictive practice of using only ‘standard risk’ organs. That is, an estimated difference of 5.6 percentage points (95$\%$ CI: 3.4 to 8.6 percentage points). The cumulative incidence curves for death under regimes $g_0$ and $g_1$ are displayed in figure \ref{fig: cuminc1}.

The estimated 10-year cumulative incidence of death is 58.4$\%$ (95$\%$ CI: 57.8 to 58.9$\%$) under regime $g_2$, corresponding to the expansive practice of the increasing utilization of `increased risk’ organs by a factor of 1.25, an incidence 1.0 percentage points lower than what would be observed under regime $g_0$ (95$\%$ CI: -1.4, -0.7 percentage points), and the estimated cumulative incidence is 57.4$\%$ (95$\%$ CI: 56.6 to 58.1$\%$) under regime $g_3$, corresponding to the expansive practice of the increasing utilization of `increased risk’ organs by a factor of 1.5. That is, an incidence 1.9 percentage points lower than what would be observed under regime $g_0$ (95$\%$ CI: -- to --$\%$ percentage points). The cumulative incidence for death under regimes $g_0$, $g_2$, and $g_3$ are displayed in figure \ref{fig: cuminc23}. 

In other words, we estimated that, despite the concerns regarding infectious disease transmission and organ inferiority, a policy of abolishing utilization of ‘increased risk’ organs would have increased the cumulative incidence of death by 5.6 percentage points at 10 years after addition to the waiting list, and that increasing utilization of ‘increased risk’ organs would actually increase overall survival, compared to current practice.

All analyses were performed using \texttt{R} version 3.5.2 and \href{https://github.com/KerollosWanis/finite_resources_transplant}{\color{blue}code is available} to reproduce these results, as well as to flexibly estimate the effects of other proportionally-representative interventions on `standard risk' and `increased risk' organ utilization. Ninety-five percent confidence intervals were obtained from the $97.5^{th}$ and $2.5^{th}$ percentiles of the distribution of point estimates obtained by repeating the IPW algorithm on 500 nonparametric bootstrap samples.

\section{Conclusion}
  
In this article we have presented a novel class of policy-relevant estimands: expected potential outcomes under \textit{proportionally-representative interventions} that constrain treatment resources. These estimands are policy-relevant because they (i) incorporate substantive knowledge to specify limits on treatment utilization that are feasibly achieved under an actual policy and (ii) coarsely preserve features of the observed joint distribution between treatment and covariates - specifically, the rank-order of conditional treatment probabilities, which would naturally be unperturbed in settings where the intervention is a manipulation of treatment resource scarcity. 

Our estimands stand in contrast to the classical population-level estimand - equal to the average treatment effect of deterministic regimes - which we show is a special case of proportionally-representative interventions where treatment resources are assumed (often absurdly) to be practically unlimited. To facilitate the concerns of policy-makers, we provide simple inverse probability-weighted estimators for proportionally-representative interventions, implementable with off the shelf-software. These estimators are consistent under mildly stronger identification assumptions than those typically needed for most causal estimands. 

We demonstrated the use of the proportionally representative interventions in a study of organ transplantation policies, where treatment resource limitations are severe. Evidence from our marginal structural models supports the continued (and possible expanded) use of the suspected inferior treatment resource (the so-called ``increased-risk'' liver grafts), suggesting that the increased scarcity imposed by elimination of these grafts outweighs the suspected inferiority of receiving these grafts, with respect to the cumulative incidence of death in the population of transplant eligible individuals.

Our class of estimands generalizes previously used unlimited resource estimands and are easy to implement. However, the \textit{representative} nature of the stochastic interventions may not always correspond to a policy that is of substantive interest. Specifically, policy makers might be interested in effects of alternative prioritization mechanisms, for example, hypothetical policies corresponding to alternative algorithms for rank determination on the national waiting list for liver transplants. Such estimands fall outside of the class defined by proportionally-representative interventions because the \textit{representative} nature of the stochastic interventions is tantamount to preserving the natural (observed) prioritization of individuals with different covariate values under the regimes considered. We will study extensions of our results to a broader class of realistic, policy-relevant estimands in future work. 

\clearpage

\clearpage


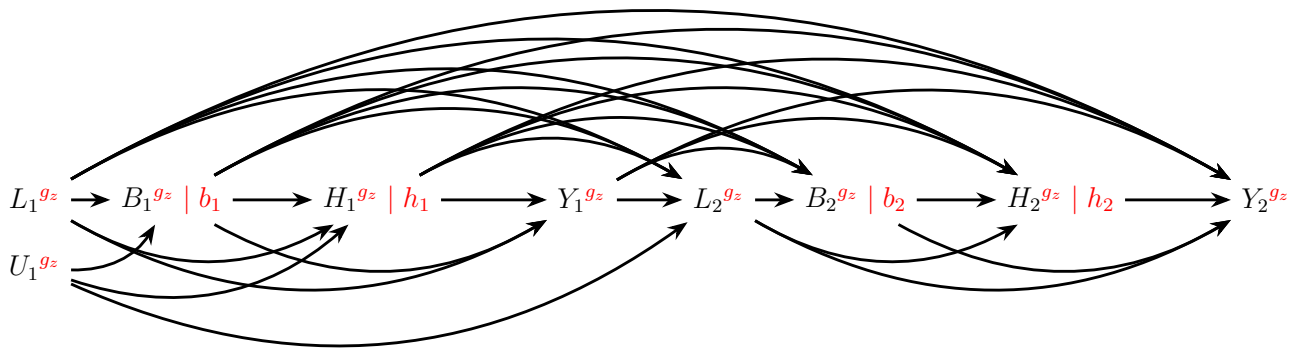
\begin{figure}\scalebox{.9}{
\centering
\begin{tikzpicture}
\begin{scope}[every node/.style={thick,draw=none}]
    \node (L1) at (0,0) {$L_1 \color{red} ^{g_z}$};
    \node (U1) at (0,-1) {$U_1 \color{red} ^{g_z}$};
	\node (B1) at (2,0) {$B_1 \color{red} ^{g_z} \mid b_1$};
    \node (H1) at (5,0) {$H_1 \color{red} ^{g_z} \mid h_1$};
    \node (Y1) at (8,0) {$Y_1 \color{red}  ^{g_z}$};
    \node (L2) at (10,0) {$L_2 \color{red} ^{g_z}$};
	\node (B2) at (12,0) {$B_2 \color{red} ^{g_z} \mid b_2$};
    \node (H2) at (15,0) {$H_2 \color{red} ^{g_z} \mid h_2$};
    \node (Y2) at (18,0) {$Y_2 \color{red} ^{g_z}$};
\end{scope}

\begin{scope}[>={Stealth[black]},
              every node/.style={fill=white,circle},
              every edge/.style={draw=black,very thick}]
    \path [->] (B1) edge (H1);
    \path [->] (B1) edge[bend right] (Y1);
    \path [->] (H1) edge (Y1);
    \path [->] (L1) edge[bend right] (Y1);
    \path [->] (L1) edge (B1);
    \path [->] (L1) edge[bend right] (H1);
    \path [->] (U1) edge[bend right] (B1);
    \path [->] (U1) edge[bend right] (H1);
    \path [->] (U1) edge[bend right] (L2);
    \path [->] (Y1) edge (L2);
    \path [->] (Y1) edge[bend left] (B2);
    \path [->] (Y1) edge[bend left] (H2);
    \path [->] (Y1) edge[bend left] (Y2);
    \path [->] (B1) edge[bend left] (L2);
    \path [->] (B1) edge[bend left] (B2);
    \path [->] (B1) edge[bend left] (H2);
    \path [->] (B1) edge[bend left] (Y2);
    \path [->] (H1) edge[bend left] (L2);
    \path [->] (H1) edge[bend left] (B2);
    \path [->] (H1) edge[bend left] (H2);
    \path [->] (H1) edge[bend left] (Y2);
    \path [->] (L1) edge[bend left] (L2);
    \path [->] (L1) edge[bend left] (B2);
    \path [->] (L1) edge[bend left] (H2);
    \path [->] (L1) edge[bend left] (Y2);
    
    \path [->] (B2) edge (H2);
    \path [->] (B2) edge[bend right] (Y2);
    \path [->] (H2) edge (Y2);
    \path [->] (L2) edge[bend right] (Y2);
    \path [->] (L2) edge (B2);
    \path [->] (L2) edge[bend right] (H2);
\end{scope}
\end{tikzpicture}
}
\caption{SWIT (K=2)in which Exchangeability 1 and 2 hold.}

\label{fig: SWIT1}
\end{figure}
\clearpage

\begin{figure}[ht]
\centering
\includegraphics[width=1\textwidth]{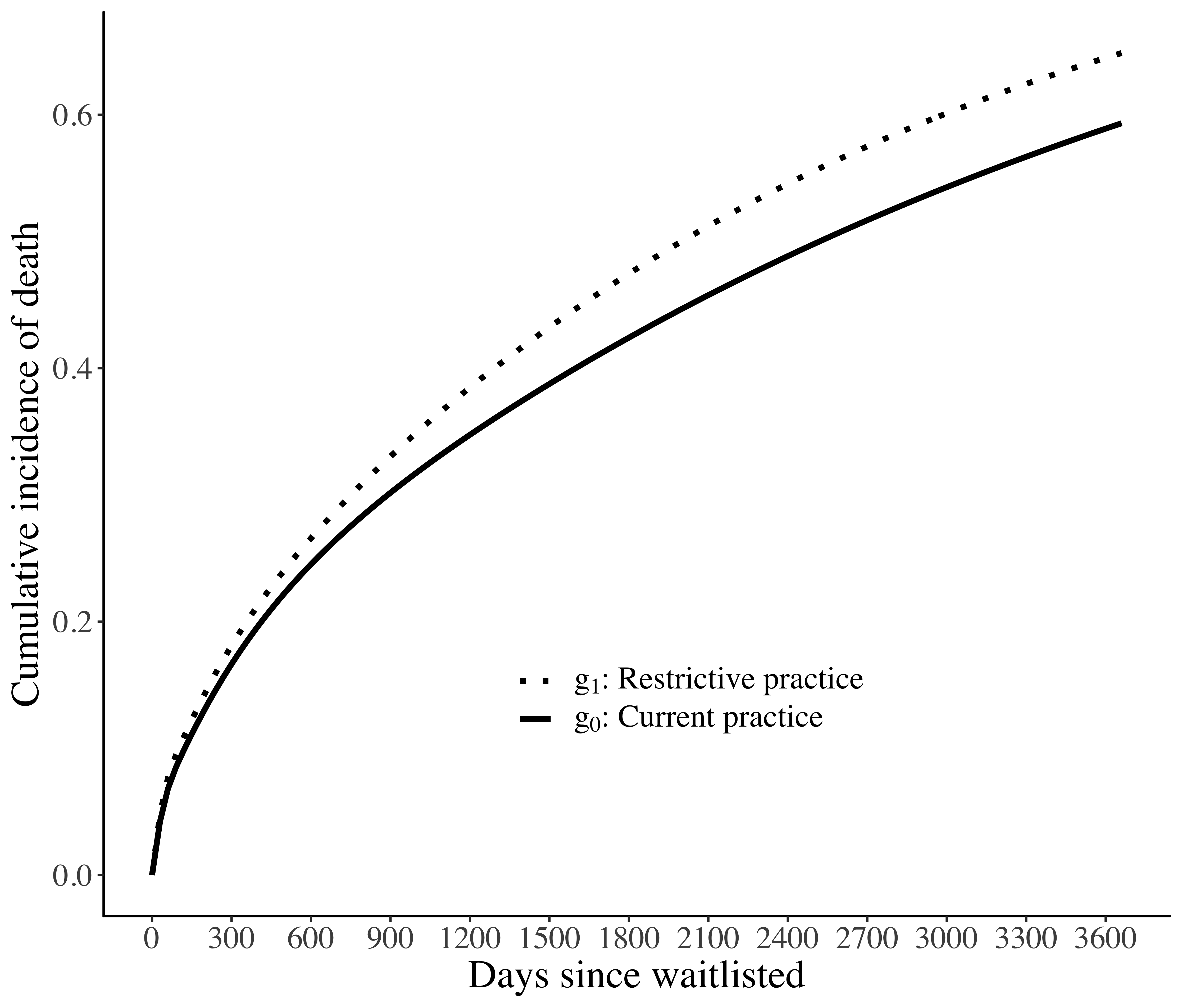}
\caption{Potential outcomes under regimes $g_0$ and $g_1$} \label{fig: cuminc1}
\end{figure} 

\begin{figure}[ht]
\centering
\includegraphics[width=1\textwidth]{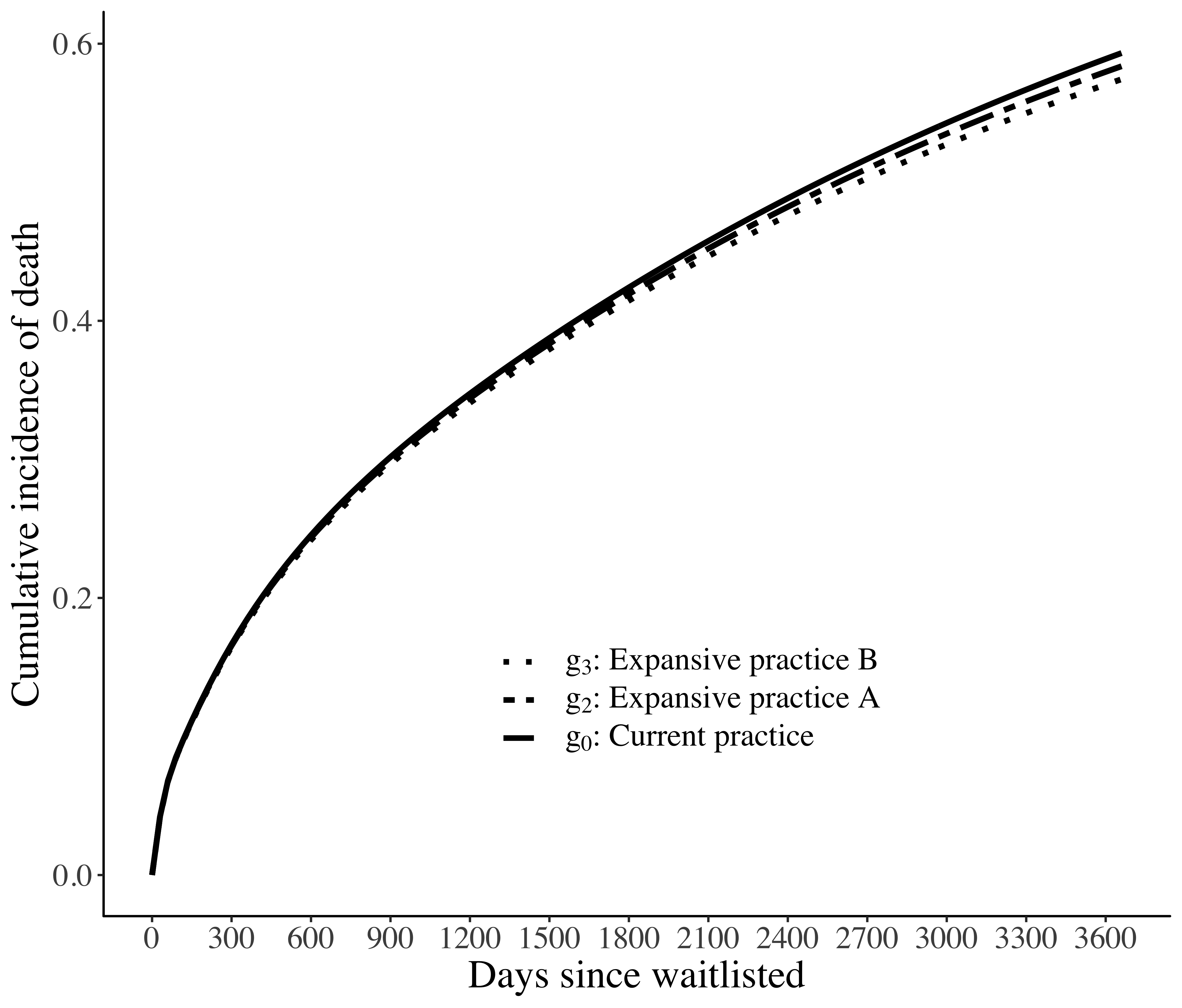}
\caption{Potential outcomes under regimes $g_0$, $g_2$, and $g_3$} \label{fig: cuminc23}
\end{figure} 

\clearpage

\begin{figure}[ht]
\centering
\includegraphics[width=1\textwidth]{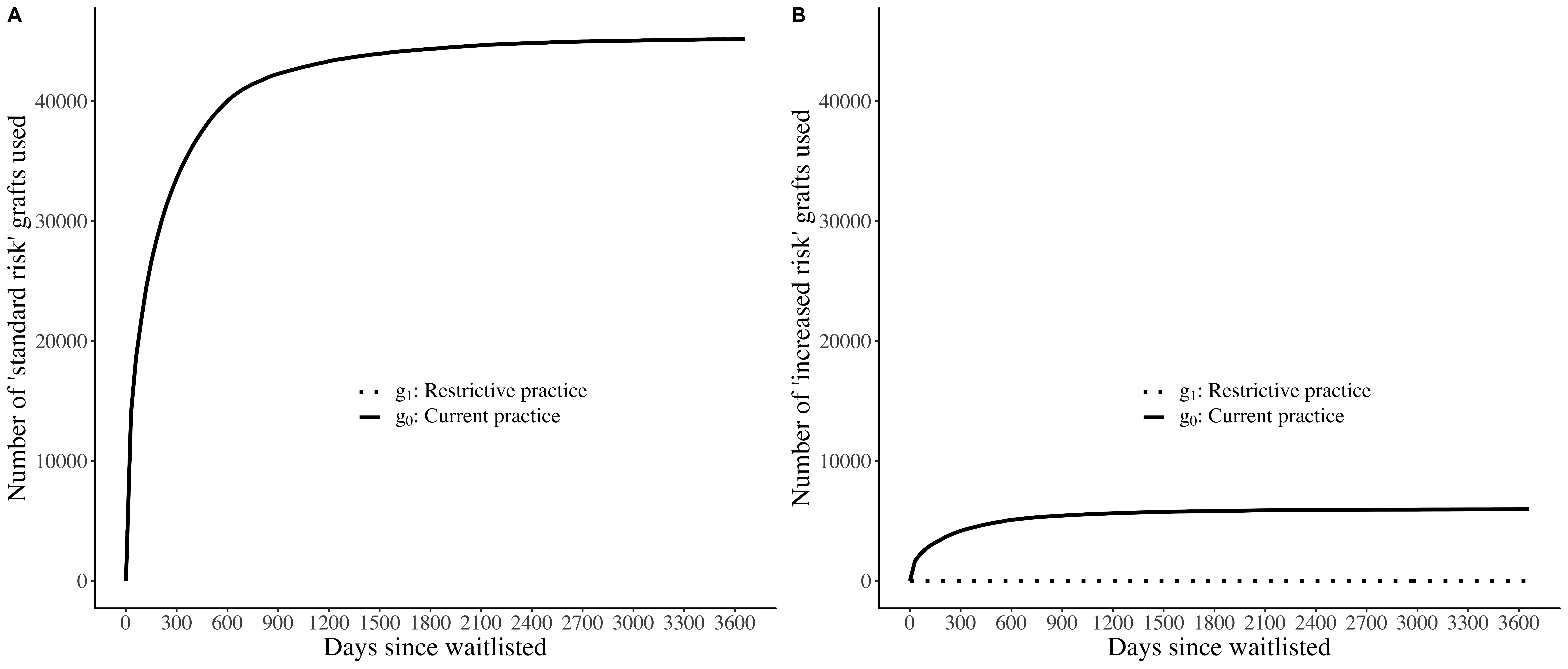}
\caption{Organ utilization under regimes $g_0$ and $g_1$} \label{fig: graftuse1}
\end{figure} 

\begin{figure}[ht]
\centering
\includegraphics[width=1\textwidth]{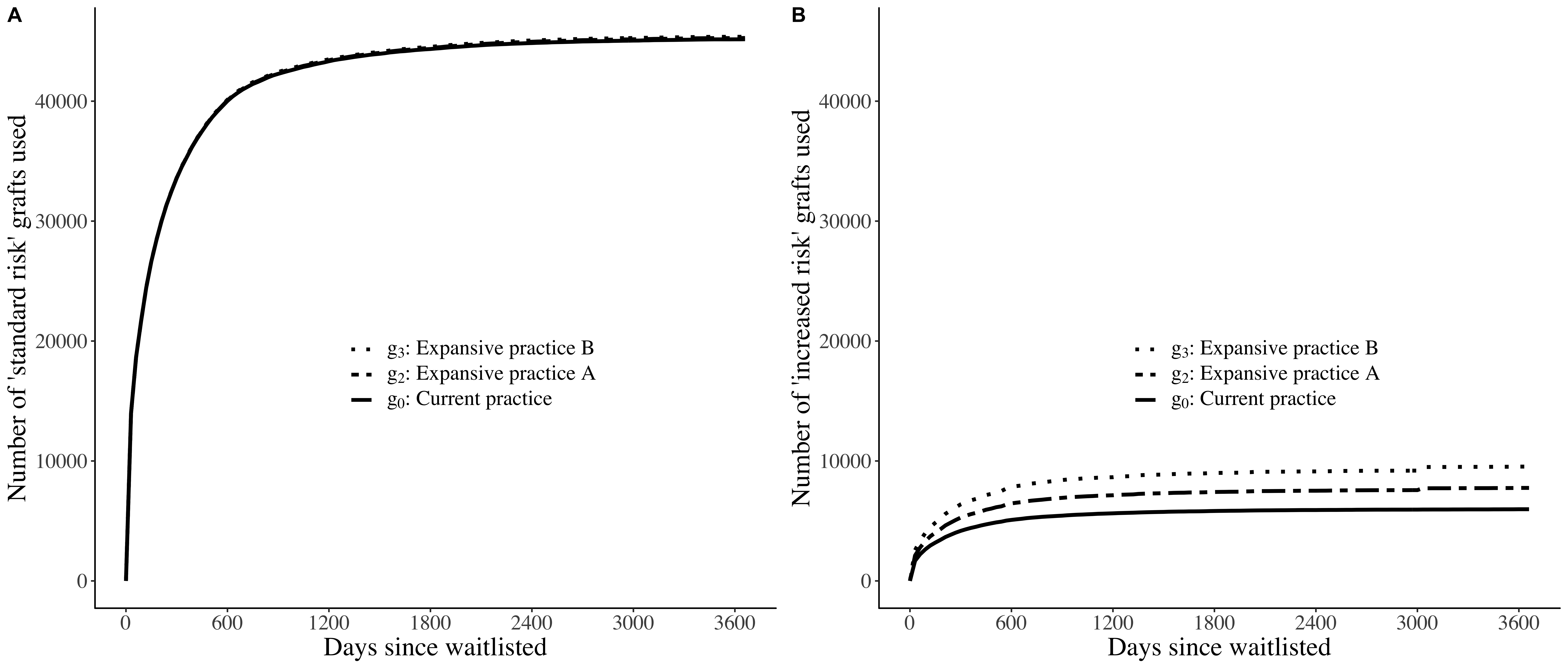}
\caption{Organ utilization under regimes $g_0$, $g_2$, and $g_3$} \label{fig: graftuse23}
\end{figure} 
\clearpage


\bibliographystyle{ieeetr}

\bibliography{Limited2.bib}
\clearpage


\section{Appendix A: Extension to arbitrary resource constraints}

In the main text, we present regimes and corresponding identification and estimation results for a subset of possible proportionally representative interventions, where treatment resources are either eliminated or are constrained to utilization equal to that under the observed data generating mechanism.  In this Appendix, we generalize results for the full class of proportionally representative interventions, where resource utilization is constrained to any arbitrary level. 

\subsection{Regimes}

As before, regimes are determined by the user-defined treatment resource constraints $q_k^{g_z}$ and $m_k^{g_z}$. However, they must be adapted to respect the additional constraint imposed by the finite treatment population - a regime could not possibly use more treatment units than there are treatment-eligible individuals under that regime in a given interval. As such, constraints are re-expressed as:

\begin{align}
    Q_k^{g_z} = \text{min}\Big(q_k^{g_z}, \frac{P(R_k^{g_z+}=1)}{P(B_k=1)} \Big),
\end{align}

and
    
\begin{align}
    M_k^{g_z} = \text{min}\Big(m_k^{g_z}, \frac{P(S_k^{g_z+}=1)}{P(H_k=1)}\Big)
\end{align}
    
for all $k \in \{1, \dots, K\}$. 

In words, we have adapted the constraints so that the resources are utilized to the highest possible extent that does not exceed the user-specified resource limitation: either that limit, or the number of treatment eligible individuals if that limit exceeds the latter.   
    
For notational convenience, we re-express the constraints in terms of probabilities and in terms of variables $\aleph^{g_z}_{B,k}$ and $\aleph^{g_z}_{H,k}$, that indicate whether the user-specified resource limitations exceed the expected number of treatment eligible individuals under regime $g_z$ within a particular interval:

\begin{align}
    P(B_k^{g_z+}=1) = 
        & \Bigg(q_k^{g_z}\times P(B_k=1)\Bigg)^{\aleph^{g_z}_{B,k}}  \label{eq; czB}\\
        \times & \Bigg(P(R_k^{g_z+}=1)\Bigg)^{1-  \aleph^{g_z}_{B,k}}
\end{align}
 
and:
    
\begin{align}
    P(H_k^{g_z+}=1) = 
        & \Bigg(m_k^{g_z}\times P(H_k=1)\Bigg)^{\aleph^{g_z}_{H,k}} \label{eq; czH} \\
        \times & \Bigg(P(S_k^{g_z+}=1)\Bigg)^{1-  \aleph^{g_z}_{H,k}}
\end{align}
    
In particular, $\aleph^{g_z}_{H,k}$ and $\aleph^{g_z}_{B,k}$ are precisely defined as
    
        \begin{align}
    \aleph^{g_z}_{B,k} =  I\Big(q_k^{g_z}\times P(B_k=1) \leq   P(R_k^{g_z+}=1)\Big) \label{eq; alephB}
        \end{align}
    
and 
    
        \begin{align}
    \aleph^{g_z}_{B,k} =  I\Big(m_k^{g_z}\times P(H_k=1) \leq   P(S_k^{g_z+}=1) \Big) \label{eq; alephH}
        \end{align}

In words, the marginal probability of post-intervention suspected-inferior treatment utilization under regime $g_z$, $P(B_k^{g_Z+}=1)$, is specified to be equal to $q_k^{g_z}\times P(B_k=1)$, or the marginal probability of suspected-inferior treatment \textit{eligibility} under regime $g_z$, $P(R_k^{g_Z+}=1)$, if the latter exceeds the former. These were not needed in the main text, since $\aleph^{g_z}_{B,k}$ and $\aleph^{g_z}_{H,k}$ are guaranteed to evaluate to 1 for all $k \in \{1,\dots, K\}$ under regimes $g_1$ and $g_0$ specified in expressions \eqref{eq; c1.1} - \eqref{eq; c2.2}, as we proove in Lemma \eqref{lemma; appA} at the end of Appendix A. 

Now, as in expressions \eqref{eq; c1.1} - \eqref{eq; c2.2} for particular regimes $g_1$ and $g_0$, we have fully specified the marginal resource constraints to be realized under an arbitrary regime $g_z$. However, since all of the conditions of Lemma \eqref{lemma; appA} are not met under an arbitrary regime $g_z$, specifically that $P(B^{g_z+}_k=1)$ and $P(H^{g_z+}_k=1)$ are not less than or equal to $P(B_k=1)$ and $P(H_k=1)$, respectively, for all $k \in \{1,\dots, K\}$, then $\alpha_k(z)$ and $\beta_k(z)$ as defined in expressions \eqref{eq; alphak} and \eqref{eq; betak} in the main text may not be less than or equal to 1 for all $k \in \{1,\dots, K\}$. By consequence, there is no guarantee that the conditional intervention  density function as defined in expressions  \eqref{eq: int Bk} and \eqref{eq: int Hk} in the main text will always evaluate to some number between 0 and 1 and thus may not be a valid probability density function. As such, we redefine the conditional intervention distributions of the proportionally representative interventions to preclude this possibility.

First, we define additional variables $\beth^{g_z}_{B,k}$ and $\beth^{g_z}_{H,k}$, which indicate when the pre-specified post-intervention marginal treatment probability exceeds the natural marginal treatment probability, under regime $g_z$ for both the suspected inferior and superior treatments:

        \begin{align}
    \beth^{g_z}_{B,k} =  I\Big(P(B_k^{g_z}=1) > P(B_k^{g_z+}=1)\Big)  \label{eq; bethB}
        \end{align}
    
    and 
    
        \begin{align}
    \beth^{g_z}_{H,k} =  I\Big(P(H_k^{g}=1) > P(H_k^{g_z+}=1)\Big) \label{eq; bethH} 
        \end{align}
    
We remind the reader that, under a particular regime, the natural treatment values are the values we would observe at the moment before treatment intervention.  Therefore, for example, $P(B_2^{g_z}=1)$ is the probability of receiving a suspected superior treatment unit under regime $g_z$, where we have intervened on suspected superior and suspected inferior treatment in interval 1, but we have not yet intervened on suspected superior treatment in interval 2.

Thus we the define the arbitrary regime $g_z$ to be a stochastic intervention on $H_k$ and  $B_k$ for all $k \in \{1,...,K\}$, such that intervention distributions are defined as follows: 
    
\begin{align}
    f_{B_k^{g_z+} \mid R_k^{g_z+}, \overline{L}_k^{g_z}}(1 \mid R_k, \overline{L}_k) 
        =      & \big(\alpha_{k}(z, R_k) \times f_{B_k \mid R_k, \overline{L}_k}(1 \mid R_k, \overline{L}_k)\big)
                    ^{\beth^{g_z}_{B,k}} \label{eq; newintB} \\
        \times & \big(1- \alpha_{k}(z, R_k) \times f_{B_k \mid R_k, \overline{L}_k}(0 \mid R_k, \overline{L}_k)\big)
                    ^{1-\beth^{g_z}_{B,k}} \nonumber,
\end{align}

and

\begin{align}
    f_{H_k^{g_z+} \mid S_k^{g_z+}, \overline{L}_k^{g_z}}(1 \mid S_k, \overline{L}_k) 
        =      & \big(\beta_{k}(z, S_k) \times f_{H_k \mid S_k, \overline{L}_k}(1 \mid S_k, \overline{L}_k)\big)
                    ^{\beth^{g_z}_{H,k}} \label{eq; newintH}\\
        \times & \big(1- \beta_{k}(z, S_k) \times f_{H_k \mid S_k, \overline{L}_k}(0 \mid S_k, \overline{L}_k)\big)
                    ^{1-\beth^{g_z}_{H,k}} \nonumber,
\end{align}

with probability 1, where $f_{B_k \mid R_k, \overline{L}_k}(\cdot \mid \cdot)$ and $f_{B_k^{g_z+} \mid R_k^{g_z+}, \overline{L}_k^{g_z}}(\cdot \mid \cdot)$ are defined as in the main text. Additionally, $\alpha_k(z)$ and $\beta_k(z)$ are regime-specific specific variables defined to satisfy resource constraints, adapted for the arbitrary regime $g_z$ as follows:

\begin{align}
    \alpha_k(z, R_k) =\Bigg[ & \Bigg(\frac{P(B^{g_z+}_k=1)}{P(B_k^{g_z}=1)}\Bigg)
                        ^{\beth^{g_z}_{B,k}} \label{eq; alphanew}\\ 
           \times & \Bigg(\frac{1-\frac{P(B^{g_z+}_k=1)}{P(R^{g_z+}_k=1)}}{1-\frac{ P(B^{g_z}_k=1)}{P(R_k^{g_z+}=1)}}\Bigg)
                        ^{(1-\beth^{g_z}_{B,k})} \nonumber \\
           \times & \Bigg(\aleph^{g_z}_{B,k}\Bigg) \Bigg]^{R_k}, \nonumber
\end{align}

and
    
\begin{align}
    \beta_k(z, S_k) =\Bigg[ & \Bigg(\frac{P(H^{g_z+}_k=1)}{P(H_k^{g_z}=1)}\Bigg)
                        ^{\beth^{g_z}_{H,k}} \label{eq; betanew}\\
           \times & \Bigg(\frac{1-\frac{P(H^{g_z+}_k=1)}{P(S^{g_z+}_k=1)}}{1-\frac{ P(H^{g_z}_k=1)}{P(S_k^{g_z+}=1)}}\Bigg)
                        ^{(1-\beth^{g_z}_{H,k})} \nonumber \\
           \times & \Bigg(\aleph^{g_z}_{H,k}\Bigg) \Bigg]^{S_k}, \nonumber
\end{align}
    
for all $k \in \{1,\dots,K\}$.

Here, we see that $\alpha_k(z, R_k)$ and $\beta_k(z, S_k)$ are random variables only with respect to $R_k$ and $S_k$, respectively. Thus, $\alpha_k(z, R_k)$ is a scaling function that:

\begin{enumerate}
    \item [1.] Evaluates to 1 whenever $R_k=0$ (that is, it does not scale one's probability of receiving or not receiving the suspected-superior treatment if one is not eligible to receive the suspected-superior treatment), else: 
    \item [2.] Evaluates to 0 whenever $\aleph^{g_z}_{B,k}=0$ (in which case, $\beth^{g_z}_{H,k}=0$), and so the conditional probabilities of \textit{not} receiving an organ are set to 0, else:
    \item [3.] Evaluates to some number between 1 and 0, and are applied to the conditional probabilities of receiving or not receiving the suspected-superior treatment, depending on the value of $\beth^{g_z}_{B,k}$. We leave proof of this particular claim to the reader, as it follows naturally from definitions of the indicator variables.
\end{enumerate}

In the context of the expression of the weights above,  $\beth^{g_z}_{B,k}$ indicates whether or not constraint satisfaction demands an intervention that scales up or scales down the natural likelihood of suspected-superior treatment reception, whereas  $\aleph^{g_z}_{B,k}$ indicates whether or not constraint satisfaction demands an intervention that assigns the suspected-superior treatment to \textit{all} eligible individuals. 

To provide additional intuition about the expressions for $\alpha_{k}(z, R_k)$ we note that $P(B_k^{g_z+}=1)=P(B_k^{g_z+}=1, R_k^{g_z+}=1)$, due to the determinism between $B_k^{g_z+}$ and $R_k^{g_z+}$. And so:
    
\begin{align}
            & \frac{P(B_k^{g_z+}=1)}{P(B_k^{g_z}=1)} 
            = \frac{P(B_k^{g_z+}=1 \mid R_k^{g_z+}=1)}{P(B_k^{g_z}=1 \mid R_k^{g_z+}=1)}  \nonumber
\end{align}

    and:
    
\begin{align}
            & \frac{1- \frac{P(B_k^{g_z+}=1)}{P(R_k^{g_z+}=1)}}{1-\frac{ P(B_k^{g_z}=1)}{P(R_k^{g_z+}=1)}}  
            =  \frac{P(B_k^{g_z+}=0 \mid R_k^{g_z+}=1)}{P(B_k^{g_z}=0 \mid R_k^{g_z+}=1)} \nonumber
\end{align}

Therefore, when $\aleph^{g_z}_{B,k}=1$ then $\alpha_{k}(z, R_k)$ evaluates to a ratio of the conditional probabilities of post-intervention and natural treatment reception under regime $g_z$.

\subsection{Identification}

In addition to identification conditions \eqref{ex1B} - \eqref{eq: positivityH}, we require the following exchangeability conditions for $H_k^{g_z}$, analogous to Exchangeability 2 conditions \eqref{ex2BB} and \eqref{ex2BH}:

\begin{align}
    \underline{H}^{g_z}_t \independent I(B_{t}^{g_z}=b_{t}) \mid \overline{L}_{t}^{g_z}=\overline{l}_{t}, \overline{Y}_{t-1}^{g_z}=0, \overline{H}_{t-1}^{g_z}=\overline{h}_{t-1}, \overline{B}_{t-1}^{g_z}=\overline{b}_{t-1}, \label{ex2HB}
\end{align}

for $\{\overline{b}_{t}, \overline{l}_{t}, \overline{h}_{t-1}  \mid P(      
    \overline{B}_{t-1}^{g_z+}=\overline{b}_{t}, 
    \overline{L}_{t}^{g_z}=\overline{l}_{t}, 
    \overline{Y}_{t-1}^{g_z}=0, 
    \overline{H}_{t-1}^{g_z+}=\overline{h}_{t-1})>0\}$,  $t\in\{1,\dots,k\}$,
\\

and:

\begin{align}
    \underline{H}^{g_z}_t \independent I(H_{t-1}^{g_z}=h_{t-1}) \mid  \overline{B}_{t-1}^{g_z}=\overline{b}_{t-1}, \overline{L}_{t-1}^{g_z}=\overline{l}_{t-1}, \overline{Y}_{t-2}^{g_z}=0, \overline{H}_{t-2}^{g_z}=\overline{h}_{t-2},\label{ex2HH}
\end{align}

for $\{\overline{h}_{t-1},  \overline{B}_{t-1}^{g_z}=\overline{b}_{t-1},  \overline{l}_{t-1}  \mid P(\overline{H}_{t-1}^{g_z+}=\overline{h}_{t-1}        
        \overline{B}_{t-1}^{g_z+}=\overline{b}_{t-1}, 
        \overline{L}_{t-1}^{g_z} =\overline{l}_{t-1}, 
        \overline{Y}_{t-2}^{g_z}=0)>0\}$,  $t\in\{1,\dots,k\}$.
\\

These additional exchangeability conditions are required to identify $\beta_k(z, S_k)$, which cannot be trivially identified as evaluating to either 0 or 1, as in regimes $g_1$ and $g_0$ described in the main text.

\subsubsection{Identification formulae}

As in Section \ref{subsec: gfom}, when all identification conditions hold, including \eqref{ex2HB} and \eqref{ex2HH}, we can identify $\mathbb{E}( Y^{g_z}_K )$ from the same non-extended g-formula of Robins (1986) for $Y_K$, $f^{g_z}_{Y_K}(1)$: equal to \eqref{eq; gformY}, except now $\alpha_k(z, R_k)$ and $\beta_k(z, S_k)$ are identified by the functionals in expressions \eqref{eq; alphanew} and  \eqref{eq; betanew}, and $\aleph$ and $\beth$ indicator functions are identified by the functionals in expressions \eqref{eq; alephB} - \eqref{eq; bethH}, except replacing $P(B_k^{g_z}=1)$ with $f^{g_z}_{B_k}(1)$ of \eqref{eq; gformB}, $P(H_k^{g_z}=1)$ with $f^{g_z}_{H_k}(1)$ of \eqref{eq; gformH}, and $P(R_k^{g_z+}=1)$ with $f^{g_z}_{R_k}(1)$ and $P(S_k^{g_z+}=1)$ with $f^{g_z}_{S_k}(1)$, where $f^{g_z}_{R_k}(1)$ and $f^{g_z}_{S_k}(1)$ are the g-formulae for $R_k$ and $S_k$, respectively:

\begin{align}
    f^{g_z}_{R_k}(1) = 
    & \sum_{\overline{l}_{k-1}} P(Y_{k-1}=0 \mid \overline{H}_{k-1}=0, \overline{B}_{k-1}=0, \overline{L}_{k-1}=\overline{l}_{k-1}, Y_{k-2}=0) \label{eq; gformR} \\
    & \times  \prod_{m=1}^{k-1} \Big\{f_{H_{m}^{g_z+} \mid S_{m}^{g_z+}, \overline{L}_{m}^{g_z}}(0 \mid 1, \overline{l}_{m}) \nonumber \\
    & \times  f_{B_{m}^{g_z+} \mid R_{m}^{g_z+}, \overline{L}_{m}^{g_z}}(0 \mid 1, \overline{l}_{m})  \nonumber \\
    & \times P(L_m=l_m \mid R_m=1, \overline{L}_{m-1}=\overline{l}_{m-1}) \nonumber\\
    & \times P(Y_{m-1}=0 \mid H_{m-1}=0, S_{m-1}=1, \overline{L}_{m-1}=\overline{l}_{m-1})\Big\} \nonumber,
\end{align}

and

\begin{align}
    f^{g_z}_{S_k}(1) = 
    & \sum_{\overline{l}_{k}} f_{B_{k}^{g_z+} \mid R_{k}^{g_z+}, \overline{L}_{k}}(1 \mid 1, \overline{l}_{k}) \label{eq; gformS} \\
    & \times   P(L_k=l_k \mid R_k=1, \overline{L}_{k-1}=\overline{l}_{k-1}) \\
    & \times \prod_{m=1}^{k-1} \Big\{P(Y_{m}=0 \mid H_{m}=0, S_{m}=1, \overline{L}_{m}=\overline{l}_{m}) \nonumber \\
    & \times  f_{H_{m}^{g_z+} \mid S_{m}^{g_z+}, \overline{L}_{m}^{g_z}}(0 \mid 1, \overline{l}_{m}) \nonumber \\
    & \times  f_{B_{m}^{g_z+} \mid R_{m}^{g_z+}, \overline{L}_{m}}(0 \mid 1, \overline{l}_{m})  \nonumber \\
    & \times   P(L_m=l_m \mid R_m=1, \overline{L}_{m-1}=\overline{l}_{m-1})\Big\} \nonumber.
\end{align}

\subsubsection{Alternative g-formulae representation}

As in the main text, the g-formula can be represented as in expressions \eqref{eq; altgformY} and \eqref{eq; lambda}, but where $W^{g_z}_{B,k}$ and $W^{g_z}_{H,k}$ are expressed as follows:

\begin{align}
        W_{B,k}^{g_z}=    \prod_{j=1}^{k}
            \begin{bmatrix*}[l]
            &  \begin{pmatrix*}[l]  \frac{
                     \Big( \alpha_j(z, R_j) \times f_{B_{j} \mid R_{j}, \overline{L}_{j}}(1 \mid R_{j}, \overline{L}_{j})\Big)^{B_{j}}  
                    \times \Big(1- \alpha_j(z, R_j) \times f_{B_{j} \mid R_{j}, \overline{L}_{j}}(1 \mid R_{j}, \overline{L}_{j})\Big)^{1-B_j}
                }{
                f_{B_j \mid R_j, \overline{L}_j}(B_j \mid R_j, \overline{L}_j)
                } 
                \end{pmatrix*}^{\beth^{g_z}_{B,k}} \\
    \times &  \begin{pmatrix*}[l]  \frac{
                     \Big(1 - \alpha_j(z, R_j) \times f_{B_{j} \mid R_{j}, \overline{L}_{j}}(0 \mid R_{j}, \overline{L}_{j})\Big)^{B_{j}}  
                    \times \Big(\alpha_j(z, R_j) \times f_{B_{j} \mid R_{j}, \overline{L}_{j}}(0 \mid S_{j}, \overline{L}_{j})\Big)^{1-B_j}
                }{
                f_{B_j \mid R_j, \overline{L}_j}(B_j \mid R_j, \overline{L}_j)
                } 
                \end{pmatrix*}^{(1-\beth^{g_z}_{B,k})} \\
            \end{bmatrix*} \label{eq; newweight B}.
\end{align}

and 

\begin{align}
        W_{H,k}^{g_z}=    \prod_{j=1}^{k}
            \begin{bmatrix*}[l]
            &  \begin{pmatrix*}[l]  \frac{
                     \Big( \beta_j(z, S_j) \times f_{H_{j} \mid S_{j}, \overline{L}_{j}}(1 \mid S_{j}, \overline{L}_{j})\Big)^{H_{j}}  
                    \times \Big(1- \beta_j(z, S_j) \times f_{H_{j} \mid S_{j}, \overline{L}_{j}}(1 \mid S_{j}, \overline{L}_{j})\Big)^{1-H_j}
                }{
                f_{H_j \mid S_j, \overline{L}_j}(H_j \mid S_j, \overline{L}_j)
                } 
                \end{pmatrix*}^{\beth^{g_z}_{H,k}} \\
    \times &  \begin{pmatrix*}[l]  \frac{
                     \Big(1 - \beta_j(z, S_j) \times f_{H_{j} \mid S_{j}, \overline{L}_{j}}(0 \mid S_{j}, \overline{L}_{j})\Big)^{H_{j}}  
                    \times \Big(\beta_j(z, S_j) \times f_{H_{j} \mid S_{j}, \overline{L}_{j}}(0 \mid S_{j}, \overline{L}_{j})\Big)^{1-H_j}
                }{
                f_{H_j \mid S_j, \overline{L}_j}(H_j \mid S_j, \overline{L}_j)
                } 
                \end{pmatrix*}^{(1-\beth^{g_z}_{H,k})} \\
            \end{bmatrix*} \label{eq; newweight H}.
\end{align}

The g-formulae for $B_k$ and $H_k$ are expressed as in \eqref{eq; altgformB} and \eqref{eq; altgformB}, and the additional g-formulae for $R_k$ and $S_k$ are expressed as:

\begin{align}
   f^{g_z}_{R_j}(1) =  & \pi_{R,j}^{g_z}, \label{eq; altgformR}
\end{align}

where

\begin{align}
    & \pi_{R,j}^{g_z} = \mathbb{E}\big[R_jW_{H,j-1}^{g_z}W_{B,j-1}^{g_z}\big]
\end{align}

and the g-formula of expression \eqref{eq; gformH} as:

\begin{align}
   f^{g_z}_{S_j}(1) =  & \pi_{S,j}^{g_z}, \label{eq; altgformS}
\end{align}

where

\begin{align}
    & \pi_{S,j}^{g_z} = \mathbb{E}\big[S_jW_{B,j}^{g_z}W_{H,j-1}^{g_z}\big].
\end{align}

\subsection{Inverse Probability Weighted Estimation of Risk under Proportionally Representative Interventions}

\subsubsection{Marginal Structural Models}
\label{subsec; AppAMSM}

Consider a cloned dataset as in subsection \ref{subsec; MSM} of the main text, except there are as many cloned copies as their are regimes under consideration, where $\mathcal{Z}$ is the support of $Z$, and whose instantiations $z$ index a particular regime $g_z$. Then $\lambda_{Y,k}^{g_z}(V)$, is likewise re-expressed as in \eqref{eq; altgformYZ} where

\begin{align}
        W_{B,k}(Z)=    \prod_{j=1}^{k}
            \begin{bmatrix*}[l]
            &  \begin{pmatrix*}[l]  \frac{
                     \Big( \alpha_j(Z, R_j) \times f_{B_{j} \mid R_{j}, \overline{L}_{j}}(1 \mid R_{j}, \overline{L}_{j})\Big)^{B_{j}}  
                    \times \Big(1- \alpha_j(Z, R_j) \times f_{B_{j} \mid R_{j}, \overline{L}_{j}}(1 \mid R_{j}, \overline{L}_{j})\Big)^{1-B_j}
                }{
                f_{B_j \mid R_j, \overline{L}_j}(B_j \mid R_j, \overline{L}_j)
                } 
                \end{pmatrix*}^{\beth_{B,k}(Z)} \\
    \times &  \begin{pmatrix*}[l]  \frac{
                     \Big(1 - \alpha_j(Z, R_j) \times f_{B_{j} \mid R_{j}, \overline{L}_{j}}(0 \mid R_{j}, \overline{L}_{j})\Big)^{B_{j}}  
                    \times \Big(\alpha_j(Z, R_j) \times f_{B_{j} \mid R_{j}, \overline{L}_{j}}(0 \mid R_{j}, \overline{L}_{j})\Big)^{1-B_j}
                }{
                f_{B_j \mid R_j, \overline{L}_j}(B_j \mid R_j, \overline{L}_j)
                } 
                \end{pmatrix*}^{(1-\beth_{B,k}(Z))} \\
            \end{bmatrix*} .
\end{align}

and 

\begin{align}
        W_{H,k}(Z)=    \prod_{j=1}^{k}
            \begin{bmatrix*}[l]
            &  \begin{pmatrix*}[l]  \frac{
                     \Big( \beta_j(Z, S_j) \times f_{H_{j} \mid S_{j}, \overline{L}_{j}}(1 \mid S_{j}, \overline{L}_{j})\Big)^{H_{j}}  
                    \times \Big(1- \beta_j(Z, S_j) \times f_{H_{j} \mid S_{j}, \overline{L}_{j}}(1 \mid S_{j}, \overline{L}_{j})\Big)^{1-H_j}
                }{
                f_{H_j \mid S_j, \overline{L}_j}(H_j \mid S_j, \overline{L}_j)
                } 
                \end{pmatrix*}^{\beth_{H,k}(Z)} \\
    \times &  \begin{pmatrix*}[l]  \frac{
                     \Big(1 - \beta_j(Z, S_j) \times f_{H_{j} \mid S_{j}, \overline{L}_{j}}(0 \mid S_{j}, \overline{L}_{j})\Big)^{H_{j}}  
                    \times \Big(\beta_j(Z, S_j) \times f_{H_{j} \mid R_{j}, \overline{L}_{j}}(0 \mid S_{j}, \overline{L}_{j})\Big)^{1-H_j}
                }{
                f_{H_j \mid S_j, \overline{L}_j}(H_j \mid S_j, \overline{L}_j)
                } 
                \end{pmatrix*}^{(1-\beth_{H,k}(Z))} \\
            \end{bmatrix*} .
\end{align}

Then, $\pi_{B,j}(Z)$ and $\pi_{H,j}(Z)$ are defined as in \eqref{eq; altgformBZ} and \eqref{eq; altgformHZ} and $\pi_{R,j}(Z)$ and $\pi_{S,j}(Z)$ are defined analogously. $\aleph$ and $\beth$ indicator functions for the MSM are defined as in \eqref{eq; alephB} - \eqref{eq; bethH}, except replacing $f_{W_j}^{g_z}(1)$ g-formulae expressions with $\pi_{W,j}(Z)$ expressions, for arbitrary variable $W_j$.

Finally, the MSM is written as in \eqref{eq; MSM} and, given the MSM holds, the g-formula for $Y_K$ is again re-written as in \eqref{eq; altgformYMSM}.

\subsubsection{Inverse Probability Weighted Estimation} \label{subsec: IPWarb}

Estimating equations are the same as in expressions \eqref{eq; esteq} and \eqref{eq; esteqU}, where estimated weights $W_{B,k}(Z,\hat{\eta}_B)$ and $W_{H,k}(Z, \hat{\eta}_H)$ and their components are analogously defined for the arbitrary intervention, again swapping $f_{B_j \mid R_j, \overline{L}_j}(B_j \mid R_j, \overline{L}_j; \hat{\eta}_B)$ for $f_{B_j \mid R_j, \overline{L}_j}(B_j \mid R_j, \overline{L}_j)$ and $f_{H_j \mid S_j, \overline{L}_j}(H_j \mid R_j, \overline{L}_j; \hat{\eta}_H)$ for $f_{H_j \mid S_j, \overline{L}_j}(H_j \mid R_j, \overline{L}_j)$ in all places. Estimators for $\hat{\pi}_{B,j}(Z, \hat{\eta})$ and $\hat{\pi}_{H,j}(Z, \hat{\eta})$ are defined as in Section \ref{subsec: IPWest}, and estimators for $\hat{\pi}_{R,j}(Z, \hat{\eta})$ and $\hat{\pi}_{S,j}(Z, \hat{\eta})$ are defined analagously  as follows:

\begin{align}
    & \hat{\pi}_{R,j}(Z, \hat{\eta}) = \frac{1}{n}\sum_{i=1}^n\big[R_{i,j}W_{H,i,j-1}(Z_i, \hat{\eta})W_{B, i, j-1}(Z_i, \hat{\eta})\big]
\end{align}

and 

\begin{align}
    & \hat{\pi}_{S,j}(Z, \hat{\eta}) = \frac{1}{n}\sum_{i=1}^n\big[S_{i,j}W_{B, i,j}(Z_i, \hat{\eta})W_{H, i,j-1}(Z_i, \hat{\eta})\big].
\end{align}.

Then, as before, if (i) the MSM is correctly specified; and (ii), the models $f_{B_j \mid R_j, \overline{L}_j}(B_j \mid R_j, \overline{L}_j; \eta_B)$ and $f_{H_j \mid S_j, \overline{L}_j}(H_j \mid S_j, \overline{L}_j; \eta_H)$ are correctly specified, then 

\begin{align}
   \mathbb{E}[U_k(\psi^*, \eta_B^*, \eta_H^*)] = 0 
\end{align}

for all $k$, with $\eta_B^*$ and $\eta_H^*$ the true values of $\eta_B$ and $\eta_H$ and the IPW estimator $\hat{\psi}$ consistent and asymptotically normal for $\psi^*$. Assuming the same models for $h\{\gamma(k,Z,V; \psi)\}$, $f_{B_k \mid R_k, \overline{L}_j}(B_k \mid 1, \overline{L}_k; \eta_B)$ and $f_{H_k \mid S_k, \overline{L}_j}(H_k \mid 1, \overline{L}_k; \eta_H)$ as in expressions \eqref{eq; pooledlogitY} - \eqref{eq; pooledlogitH}, we can solve the estimating equation with the following generalized algorithm, applied to a cloned subject-interval dataset, constructed as before:

\subsubsection*{Generalized IPW estimation algorithm for $\psi$} \label{subsubsec: IPW alg arb}

\begin{enumerate}
    \item [1.] Using subject-interval records with $Z=1$ and $R_k=1$, obtain $\hat{\eta}_B$ by fitting pooled logistic regression model \eqref{eq; pooledlogitB} with dependent variable $B_k$ and independent variables a specified function of $k=0,\dots, K$ and $\overline{L}_k$, corresponding to the choice of $\phi_B(\cdot)$.
    
    \item [2.] Similarly, using subject-interval records with $Z=1$ and $S_k=1$, obtain $\hat{\eta}_H$ by fitting a pooled logistic regression model \eqref{eq; pooledlogitH} with dependent variable $H_k$ and independent variables a specified function of $k=0,\dots, K$ and $\overline{L}_k$, corresponding to the choice of $\phi_H(\cdot)$.
    
    \item[3.] For all $z \in \mathcal{Z}$, set $\alpha_{0}(z, 1, \hat{\eta})$ and $\beta_{0}(z, 1, \hat{\eta})$ to 1. Obtain $\alpha_1(z, 1, \hat{\eta})$, $\aleph_{B,1}(z)$, and $\beth_{B,1}(z)$ by evaluating the estimated analogues of expression \eqref{eq; alphanew}, \eqref{eq; alephB}, and \eqref{eq; bethB}, noting that $P(B_1^{g_z+}=1)$ is defined by the intervention, and taking $\hat{\pi}_{B,1}(z, \hat{\eta}$ to be the proportion of individuals with $B_1=1$, $\frac{1}{n}\sum_{i=1}^nB_{i,1}$, and noting that $\hat{\pi}_{R,1}(z, \hat{\eta})=1$ by definition.
    
    \item[4.] For each subject's line 1, attach the suspected-superior treatment weight, $W_{B,1}$, calculated as: 
    
    $$
       \begin{bmatrix*}[l]
            &  \begin{pmatrix*}[l]  \frac{
                     \Big( \alpha_1(Z, R_1) \times \text{expit}\{\phi_B(1, \overline{L}_1; \hat{\eta_B})\}\Big)^{B_{1}}  
                    \times \Big(1- \alpha_j(Z, R_1) \times \text{expit}\{\phi_B(1, \overline{L}_1; \hat{\eta_B})\}\Big)^{1-B_1}
                }{
                \Big( \text{expit}\{\phi_B(1, \overline{L}_1; \hat{\eta_B})\}\Big)^{B_{1}}  
                    \times \Big(1- \text{expit}\{\phi_B(1, \overline{L}_1; \hat{\eta_B})\}\Big)^{1-B_1}
                } 
                \end{pmatrix*}^{\beth_{B,1}(Z)} \\
    \times &  \begin{pmatrix*}[l]  \frac{
                     \Big(1 - \alpha_1(Z, R_1) \times (1-\text{expit}\{\phi_B(1, \overline{L}_1; \hat{\eta_B})\})\Big)^{B_{1}}  
                    \times \Big(\alpha_j(Z, R_1) \times (1-\text{expit}\{\phi_B(1, \overline{L}_1; \hat{\eta_B})\})\Big)^{1-B_1}
                }{
                \Big( \text{expit}\{\phi_B(1, \overline{L}_1; \hat{\eta_B})\}\Big)^{B_{1}}  
                    \times \Big(1- \text{expit}\{\phi_B(1, \overline{L}_1; \hat{\eta_B})\}\Big)^{1-B_1}
                } 
                \end{pmatrix*}^{(1-\beth_{B,1}(Z))} \\
            \end{bmatrix*}.
$$

    \item[5.] For all $z \in \mathcal{Z}$, $s_1 \in \{0, 1\}$, then obtain $\beta_1(z, s_1, \hat{\eta})$, $\aleph_{H,1}(z)$, and $\beth_{H,1}(z)$.
    \item[6.] For each subject's line 1, attach the suspected-inferior treatment weight, $W_{H,1}$, calculated as: 
    
    $$
       \begin{bmatrix*}[l]
            &  \begin{pmatrix*}[l]  \frac{
                     \Big( \beta_1(Z, S_1) \times \text{expit}\{\phi_H(1, \overline{L}_1; \hat{\eta_H})\}\Big)^{H_{1}}  
                    \times \Big(1- \beta_1(Z, S_1) \times \text{expit}\{\phi_H(1, \overline{L}_1; \hat{\eta_H})\}\Big)^{1-H_1}
                }{
                \Big( \text{expit}\{\phi_H(1, \overline{L}_1; 
                    \hat{\eta_H})\}Big)^{H_{1}}  
                    \times \Big(1- \text{expit}\{\phi_H(1, \overline{L}_1; \hat{\eta_H})\}\Big)^{1-H_1}
                } 
                \end{pmatrix*}^{\beth_{H,1}(Z)} \\
    \times &  \begin{pmatrix*}[l]  \frac{
                     \Big(1 - \beta_1(Z, S_1) \times (1-\text{expit}\{\phi_H(1, \overline{L}_1; \hat{\eta_H})\})\Big)^{B_{1}}  
                    \times \Big(\beta_1(Z, S_1) \times (1-\text{expit}\{\phi_H(1, \overline{L}_1; \hat{\eta_H})\})\Big)^{1-H_1}
                }{
                \Big( \text{expit}\{\phi_H(1, \overline{L}_1; \hat{\eta_H})\}\Big)^{H_{1}}  
                    \times \Big(1- \text{expit}\{\phi_H(1, \overline{L}_1; \hat{\eta_H})\}\Big)^{1-H_1}
                } 
                \end{pmatrix*}^{(1-\beth_{H,1}(Z))} \\
            \end{bmatrix*}.
$$
    
    \item[7.] Iterate from $k=2,\dots K$:
    \begin{enumerate}
       \item [7.1.] For all $z \in \mathcal{Z}$, $r_k \in \{0, 1\}$, obtain $\alpha_{k}(z, r_k, \hat{\eta})$, $\aleph_{B,k}(z)$, and $\beth_{B,k}(z)$.
       \item[7.2.] Using subject-interval records on line $k$, attach the suspected-superior treatment weight, $W_{B,k}$, calculated as: 
       
$$ \prod_{j=1}^k
       \begin{bmatrix*}[l]
            &  \begin{pmatrix*}[l]  \frac{
                     \Big( \alpha_j(Z, R_j) \times \text{expit}\{\phi_B(j, \overline{L}_j; \hat{\eta_B})\}\Big)^{B_{j}}  
                    \times \Big(1- \alpha_j(Z, R_j) \times \text{expit}\{\phi_B(j, \overline{L}_j; \hat{\eta_B})\}\Big)^{1-B_j}
                }{
                \Big( \text{expit}\{\phi_B(j, \overline{L}_j; \hat{\eta_B})\}\Big)^{B_{j}}  
                    \times \Big(1- \text{expit}\{\phi_B(j, \overline{L}_j; \hat{\eta_B})\}\Big)^{1-B_j}
                } 
                \end{pmatrix*}^{\beth_{B,j}(Z)} \\
    \times &  \begin{pmatrix*}[l]  \frac{
                     \Big(1 - \alpha_j(Z, R_j) \times (1-\text{expit}\{\phi_B(j, \overline{L}_j; \hat{\eta_B})\})\Big)^{B_{j}}  
                    \times \Big(\alpha_j(Z, R_j) \times (1-\text{expit}\{\phi_B(j, \overline{L}_j; \hat{\eta_B})\})\Big)^{1-B_j}
                }{
                \Big( \text{expit}\{\phi_B(j, \overline{L}_j; \hat{\eta_B})\}\Big)^{B_{j}}  
                    \times \Big(1- \text{expit}\{\phi_B(j, \overline{L}_j; \hat{\eta_B})\}\Big)^{1-B_j}
                } 
                \end{pmatrix*}^{(1-\beth_{B,j}(Z))} \\
            \end{bmatrix*}.
$$
                
        \item [7.3.] For all $z \in \mathcal{Z}$, $s_k \in \{0, 1\}$, obtain $\beta_{k}(z, s_k, \hat{\eta})$, $\aleph_{H,k}(z)$, and $\beth_{H,k}(z)$.
        \item[7.4.] Using subject-interval records on line $k$, attach the suspected-superior treatment weight, $W_{H,k}$, calculated as:  
    
    $$
    \prod_{j=1}^k
       \begin{bmatrix*}[l]
            &  \begin{pmatrix*}[l]  \frac{
                     \Big( \beta_j(Z, S_j) \times \text{expit}\{\phi_H(j, \overline{L}_j; \hat{\eta_H})\}\Big)^{H_{j}}  
                    \times \Big(1- \beta_j(Z, S_j) \times \text{expit}\{\phi_H(j, \overline{L}_j; \hat{\eta_H})\}\Big)^{1-H_j}
                }{
                \Big( \text{expit}\{\phi_H(j, \overline{L}_j; 
                    \hat{\eta_H})\}Big)^{H_{j}}  
                    \times \Big(1- \text{expit}\{\phi_H(j, \overline{L}_j; \hat{\eta_H})\}\Big)^{1-H_j}
                } 
                \end{pmatrix*}^{\beth_{H,j}(Z)} \\
    \times &  \begin{pmatrix*}[l]  \frac{
                     \Big(1 - \beta_j(Z, S_j) \times (1-\text{expit}\{\phi_H(j, \overline{L}_j; \hat{\eta_H})\})\Big)^{B_{j}}  
                    \times \Big(\beta_j(Z, S_j) \times (1-\text{expit}\{\phi_H(j, \overline{L}_j; \hat{\eta_H})\})\Big)^{1-H_j}
                }{
                \Big( \text{expit}\{\phi_H(j, \overline{L}_j; \hat{\eta_H})\}\Big)^{H_{j}}  
                    \times \Big(1- \text{expit}\{\phi_H(j, \overline{L}_j; \hat{\eta_H})\}\Big)^{1-H_j}
                } 
                \end{pmatrix*}^{(1-\beth_{H,j}(Z))} \\
            \end{bmatrix*}.
$$

    \end{enumerate}
\item[8.] Using all subject-interval records in the cloned dataset, obtain $\hat{\psi}$ by fitting a \textit{weighted} pooled logistic regression model, with weights $W_{B,k}$ and $W_{H,k}$ defined in the previous steps, dependent variable $Y_k$ and independent variables a specified function of $k=1,\dots,K$ and $(Z, V)$ corresponding to the choice of $\gamma(\cdot)$.
\end{enumerate}\
\\

Our final IPW estimate of the g-formula for the risk of death by $K$ under regime $g_z$, $f^{g_z}_{Y_K}(1)$ defined by the arbitrary proportionally-representative interventions that constrain resources can then be obtained by the plug-in estimator of expression \eqref{eq; plugin} in the main text.

\subsection{Appendix A proofs}

\begin{lemma} \label{lemma; appA}
If $P(B^{g_z+}_j=1) \le P(B_j=1)$ and $P(H^{g_z+}_j=1) \le P(H_k=1)$ for all $j \in \{1,\dots, k\}$,  $\alpha_j(z)$ and $\beta_j(z)$ are defined as in expressions \eqref{eq; alphak} and \eqref{eq; betak}, and the identification conditions of expressions \eqref{ex1B}-\eqref{eq: positivityH} hold, then $P(B_{k+1}=1) \le P(R_{k+1}^{g_z+}=1)$ and $P(H_{k+1}=1) \le P(S_{k+1}^{g_z+}=1)$ for all $k \in \{1,\dots, K\}$.
\end{lemma}

\begin{proof}
 Since $P(B^{g_z+}_1=1) \le P(B_1=1)$ and $P(B^{g_z}_1=1)=P(B_1=1)$ by consistency (since $B^{g_z}_1$ occurs prior to any intervention) then $\alpha_1(z) \le 1$. Then, comparing the terms in

$$f^{g_z}_{H_1}(1) =\sum_{\overline{l}_1} P(H_1=1 \mid \overline{L}_1=\overline{l}_1, S_1=1)  \big(1-\alpha_{1}(z) \times f_{B_{1} \mid R_{1}, \overline{L}_{1}}(1 \mid 1, \overline{l}_{1})\big) P(L_1=l_1)$$

with the terms in

$$P(H_1=1) =\sum_{\overline{l}_1} P(H_1=1 \mid \overline{L}_1=\overline{l}_1, S_1=1)  \big(1- f_{B_{1} \mid R_{1}, \overline{L}_{1}}(1 \mid 1, \overline{l}_{1})\big) P(L_1=l_1)$$

We see that $f^{g_z}_{H_1}(1) \le P(H_1=1)$ whenever $\alpha_1(z) \le 1$. Since $\alpha_1(z) \le 1$, as shown above, then $\beta_1(z) \le 1$.

Similarly, comparing the terms in $f^{g_z}_{B_k}(1)$ with the terms in

\begin{align}
    P(B_k=1) = 
    & \sum_{\overline{l}_k} P(B_k=1 \mid \overline{L}_k=\overline{l}_k, R_k=1)   \nonumber\\
    & \times \prod_{m=1}^{k} \Big\{  P(L_m=l_m \mid R_m=1, \overline{L}_{m-1}=\overline{l}_{m-1}) \nonumber\\
    & \times P(Y_{m-1}=0 \mid H_{m-1}=0, S_{m-1}=1, \overline{L}_{m-1}=\overline{l}_{m-1}) \nonumber \\
    & \times  \big(1-f_{H_{m-1} \mid S_{m-1}, \overline{L}_{m-1}}(1 \mid 1, \overline{l}_{m-1})\big) \nonumber \\
    & \times  \big(1-f_{B_{m-1} \mid R_{m-1}, \overline{L}_{m-1}}(1 \mid 1, \overline{l}_{m-1})\big) \Big\} \nonumber,  
\end{align}

and likewise comparing the terms in $f^{g_z}_{H_k}(1)$, we see that $f^{g_z}_{B_k}(1) \le P(B_k=1)$ whenever $\alpha_m(z) \le 1$ and $\beta_m(z) \le 1$ for $m=1,\dots, k-1$. Similarly, $f^{g_z}_{H_k}(1) \le P(H_k=1)$. Arguing, iteratively, as before, from $k=2,...,K$, and using identification conditions of expressions \eqref{ex1B}-\eqref{eq: positivityH} so that $f^{g_z}_{H_k}(1) = P(H_k^{g_z}=1)$ and $f^{g_z}_{B_k}(1) = P(B_k^{g_z}=1)$, we see that $\alpha_k(z) \le 1$ and $\beta_k(z) \le 1$ for all $k \in \{2,...,K\}$. 

Then, since the above suffices to show that $P(B_k^{g_z+}=1) \le P(B_k^{g_z}=1)$ and $P(H_k^{g_z+}=1) \le P(H_k^{g_z}=1)$ for all $k \in \{1,...,K\}$, and since $P(B_k^{g_z}=1) \le P(R_k^{g_z+}=1)$ and $P(H_k^{g_z}=1) \le P(S_k^{g_z+}=1)$ for all $k \in \{1,...,K\}$ by definition, then it follows that $P(B_k^{g_z+}=1) \le P(R_k^{g_z+}=1)$ and $P(H_k^{g_z+}=1) \le P(S_k^{g_z+}=1)$ for all $k \in \{1,...,K\}$. 
\end{proof}.

\begin{corollary}. \label{corollary; appA} If $P(B^{g_z+}_k=1) \le P(B_k=1)$ and $P(H^{g_z+}_k=1) \le P(H_k=1)$ for all $k \in \{1,\dots, K\}$,  $\alpha_k(z)$ and $\beta_k(z)$ are defined as in expressions \eqref{eq; alphak} and \eqref{eq; betak}, and the identification conditions of expressions \eqref{ex1B}-\eqref{eq: positivityH} hold, then $0 \ge P(B_k^{g_z+}=1) \le 1$ and $0 \ge P(H_k^{g_z+}=1) \le 1$ for all $k \in \{1,\dots, K\}$. 
\end{corollary}

\begin{proof} 
This follows trivially from Theorem 1, since $P(S_k^{g_z+}=1) \le 1$ and $P(R_k^{g_z+}=1) \le 1$ and since all terms used to iteratively define $P(B^{g_z+}_k=1)$ and $P(H^{g_z+}_k=1)$ are positive.
\end{proof}. 

\clearpage

\section{Appendix B: Proving constraint satisfaction}

The constraint is satisfied for all intervals $t \in \{1,\dots, K\}$ if 
 \par\vspace{.5\topskip}\noindent\makebox[\textwidth]{\begin{minipage}{8in}
    \begin{equation*}
        \begin{split}
\mathbb{E}\Bigg[W^{g_z}_{B, t}W^{g_z}_{H, t-1}B_t\Bigg] = P(B^{g_z+}_t=1)
        \end{split}
    \end{equation*}
\end{minipage}}\vspace{\topskip}\par

and 

 \par\vspace{.5\topskip}\noindent\makebox[\textwidth]{\begin{minipage}{8in}
    \begin{equation*}
        \begin{split}
\mathbb{E}\Bigg[W^{g_z}_{H, t}W^{g_z}_{H, t}H_t\Bigg] = P(H^{g_z+}_t=1)
        \end{split}
    \end{equation*}
\end{minipage}}\vspace{\topskip}\par

for all $t \in \{1,\dots, K\}$

We demonstrate constraint satisfaction for joint intervention on $B_1$ and $H_1$, and leave the rest of the proof for the reader. 

\subsection{Suspected inferior treatment utlization in interval 1 under regime $g_z$, after intervention (i.e., $B_1^{g_z+}$)} \ \\

Substitute expression for weights:
{\tiny 
\begin{equation*}
        \begin{split}
& \mathbb{E}\Bigg[W^{g_z}_{B, 1}B_1\Bigg]  \\
        = & \mathbb{E}\begin{bmatrix*}[l]
              &  \Bigg(\frac{
                    \Big(\alpha_1(z, R_1)P(B_1=1 \mid R_1, L_1)\Big)^{B_1}
                    \Big(1-\alpha_1(z, R_1)P(B_1=1 \mid R_1, L_1)\Big)^{(1-B_1)}
        }{          \Big(P(B_1=1 \mid R_1, L_1)\Big)^{B_1}
                    \Big(1-P(B_1=1 \mid R_1, L_1)\Big)^{(1-B_1)}}\Bigg)^{\beth^{g_z}_{B,1}} \\
        \times 
              & \Bigg(\frac{
                    \Big(1-\alpha_1(z, R_1)(1-P(B_1=1 \mid R_1, L_1))\Big)^{B_1}
                    \Big(\alpha_1(z, R_1)(1-P(B_1=1 \mid R_1, L_1))\Big)^{(1-B_1)}
        }{          \Big(P(B_1=1 \mid R_1, L_1)\Big)^{B_1} 
                    \Big(1-P(B_1=1 \mid R_1, L_1)\Big)^{(1-B_1)}}\Bigg)^{(1-\beth^{g_z}_{B,1})} \\
        \times 
              & B_1
    \end{bmatrix*}
        \end{split}
    \end{equation*}
}%

Substitute expression for $\alpha_1(z, R_1)$:
{\tiny 
\begin{equation*}
        \begin{split}
        =  \mathbb{E}\begin{bmatrix*}[l]
        & \begin{pmatrix*}[l]
         &  \begin{pmatrix*}[l] 
                 \Bigg[ \Bigg(
                    \frac{
                         q^{g_z}_1 \times P(B_1=1)
                          }{
                        P(B_1^{g_z}=1)
                          }
                \Bigg)^{\beth^{g_z}_{B,1}}                           
                \Bigg(
                    \frac{
                        1- 
                        \frac{
                            q^{g_z}_1 \times P(B_1=1)
                              }{
                            P(R_1^{g_z+}=1)
                              }
                          }{
                        P(B^{g_z}_1= 0\mid R^{g_z+}_1=1)
                          }
                \Bigg)^{(1- \beth^{g_z}_{B,1})}      
                \Bigg(
                    \aleph^{g_z}_{B,1}
                \Bigg) \Bigg]^{R_1}
                P(B_1=1 \mid R_1,L_1)
           \end{pmatrix*}^{B_1}       \\
       \times 
          & \begin{pmatrix*}[l] 
             1- 
                \Bigg[ \Bigg(
                    \frac{
                         q^{g_z}_1 \times P(B_1=1)
                         }{
                        P(B_1^{g_z}=1)
                         }
                \Bigg)^{\beth^{g_z}_{B,1}}                           
                \Bigg(
                    \frac{
                        1- 
                        \frac{
                            q^{g_z}_1 \times P(B_1=1)
                              }{
                            P(R_1^{g_z+}=1)
                              }                         }{
                        P(B^{g_z}_1= 0\mid R^{g_z+}_1=1)
                         }
                \Bigg)^{(1- \beth^{g_z}_{B,1})  }     
                \Bigg(
                    \aleph^{g_z}_{B,1}
                \Bigg) \Bigg]^{R_1}
                P(B_1=1 \mid R_1,  L_1)
            \end{pmatrix*} ^{(1-B_1)}    \\
       \times & 
            \frac{
                   1
                 }{          
                P(B_1=1 \mid R_1,  L_1)^{B_1}
                P(B_1=0 \mid R_1,  L_1)^{(1-B_1)}}
        \end{pmatrix*}^{\beth^{g_z}_{B,1}} \\
 \times & \begin{pmatrix*}[l]
           & \begin{pmatrix*}[l]1- 
                \Bigg[\Bigg(
                    \frac{
                        q^{g_z}_1 \times P(B_1=1)
                          }{
                        P(B_1^{g_z}=1)
                          }
                 \Bigg)^{\beth^{g_z}_{B,1}}                         \times
                 \Bigg(
                    \frac{
                        1- 
                        \frac{
                            q^{g_z}_1 \times P(B_1=1)
                              }{
                            P(R_1^{g_z+}=1)
                              }                         }{
                        P(B^{g_z}_1= 0\mid R^{g_z+}_1=1)
                         }
                 \Bigg)^{(1- \beth^{g_z}_{B,1})  } \times  
                 \Bigg(
                    \aleph^{g_z}_{B,1}
                 \Bigg) \Bigg]^{R_1}
                 P(B_1=0 \mid R_1,  L_1)
            \end{pmatrix*}^{B_1} \\
    \times & \begin{pmatrix*}[l] 
                \Bigg[\Bigg(
                    \frac{
                         q^{g_z}_1 \times P(B_1=1)
                        }{
                        P(B_1^{g_z}=0)
                        }\Bigg)^{\beth^{g_z}_{B,1}}                 \times
                \Bigg(
                    \frac{
                        1- 
                        \frac{
                            q^{g_z}_1 \times P(B_1=1)
                              }{
                            P(R_1^{g_z+}=1)
                              }                         }{
                        P(B^{g_z}_1= 0\mid R^{g_z+}_1=1)
                         }
                \Bigg)^{(1- \beth^{g_z}_{B,1})  }   \times  
                \Bigg(
                    \aleph^{g_z}_{B,1}
                \Bigg) \Bigg]^{R_1}
                P(B_1=0 \mid R_1,  L_1)
            \end{pmatrix*}^{(1-B_1)} \\
    \times & \frac{
                    1
                 }{ 
                    P(B_1=1 \mid R_1,  L_1)^{B_1}
                    P(B_1=0 \mid R_1,  L_1)^{(1-B_1)}
                 }
    \end{pmatrix*}^{(1-\beth^{g_z}_{B,1})} \\
  \times  &   B_1 \end{bmatrix*} \\
        \end{split}
    \end{equation*}
}%

Evaluating total expectation, conditional on $R_1$ and $L_1$. Since $B_1$ is binary, terms exponentiated by $(1-B_1)$ drop out of the sum, and since $B_1 =1$ implies $R_1=1$, then terms exponentiated by $R_1$ are retained and the $R_1$ exponent is dropped. Furthermore, terms exponentiated by the indicator functions $(1-\beth^{g_z}_{B,1})$ that are themselves exponentiated by the indicator functions $\beth^{g_z}_{B,1}$, and vice versa, are removed, as are the inner indicator functions that themselves agree with the outer indicator functions. Finally, the alpha term is removed from terms themselves exponentiated by $\beth^{g_z}_{B,1}$ (since $\beth^{g_z}_{B,1}=1$ implies $\aleph^{g_z}_{B,1}=1$).

{\tiny 
\begin{equation*}
        \begin{split}
        =  \sum_{l_1} \begin{Bmatrix*}[l]
        & \begin{pmatrix*}[l]
                    \frac{
                    \Bigg(
                        \frac{
                             q^{g_z}_1 \times P(B_1=1)
                              }{
                            P(B_1^{g_z}=1)
                              }
                    \Bigg)\times  
                    P(B_1=1 \mid R_1=1,l_1) 
                 }{
                    P(B_1=1 \mid R_1=1,  l_1)
                 }
        \end{pmatrix*}^{\beth^{g_z}_{B,1}} \\
 \times & \begin{pmatrix*}[l]
 \frac{1-
                 \Bigg(
                    \frac{
                        1- 
                        \frac{
                            q^{g_z}_1 \times P(B_1=1)
                              }{
                            P(R_1^{g_z+}=1)
                              }                         
                         }{
                        P(B^{g_z}_1= 0\mid R^{g_z+}_1=1)
                         }
                 \Bigg) \times  
                 \Bigg(
                    \aleph^{g_z}_{B,1}
                 \Bigg) 
                 P(B_1=0 \mid R_1=1,  l_1)
                 }{ 
                    P(B_1=1 \mid R_1=1,  l_1)
                 }
    \end{pmatrix*}^{(1-\beth^{g_z}_{B,1})} \\
   \times  &    P(B_1=1 \mid R_1=1, l_1)P(R_1=1 \mid l_1) P(L_1=l_1) \end{Bmatrix*} \\
        \end{split}
    \end{equation*}
}%

Since by consistency $B_1^{g_z}=B_1$ and $R_1^{g_z+}=R_1=1$, then,
{\tiny 
\begin{equation*}
        \begin{split}
        =  \sum_{l_1} \begin{Bmatrix*}[l]
        & \begin{pmatrix*}[l]
            \frac{
                    \Bigg(
                        \frac{
                            q^{g_z}_1 \times P(B_1=1)
                              }{
                            P(B_1=1)
                              }
                    \Bigg)\times  
                    P(B_1=1 \mid l_1) 
                 }{
                    P(B_1=1 \mid l_1)
                 }
        \end{pmatrix*}^{\beth^{g_z}_{B,1}} \\
 \times & \begin{pmatrix*}[l]
            \frac{1-
                 \Bigg(
                    \frac{
                        1- q^{g_z}_1 \times P(B_1=1)  
                         }{
                        P(B_1= 0)
                         }
                 \Bigg) \times  
                 \Bigg(
                    \aleph^{g_z}_{B,1}
                 \Bigg) 
                 P(B_1=0 \mid   l_1)
                 }{ 
                    P(B_1=1 \mid  l_1)
                 }
    \end{pmatrix*}^{(1-\beth^{g_z}_{B,1})} \\
   \times  &    P(B_1=1 \mid  l_1) P(L_1=l_1) \end{Bmatrix*} \\
        \end{split}
    \end{equation*}
}%

Re-arranging and cancelling terms,

{\tiny 
\begin{equation*}
        \begin{split}
        =  \sum_{l_1} \begin{Bmatrix*}[l]
        & \begin{pmatrix*}[l]
            q^{g_z}_1 \times 
        \end{pmatrix*}^{\beth^{g_z}_{B,1}} \\
 \times & \begin{pmatrix*}[l]
            \frac{1-
                 \Bigg(
                    \frac{
                        1- q^{g_z}_1 \times P(B_1= 1)
                         }{
                        P(B_1= 0)
                         }
                 \Bigg) \times  
                 \Bigg(
                    \aleph^{g_z}_{B,1}
                 \Bigg) 
                 P(B_1=0 \mid l_1)
                 }{ 
                    P(B_1=1 \mid  l_1)
                 }
    \end{pmatrix*}^{(1-\beth^{g_z}_{B,1})} \\
   \times  &    P(B_1=1 \mid l_1) P(L_1=l_1) \end{Bmatrix*} \\
        \end{split}
    \end{equation*}
}%

Noting that when $\beth^{g_z}_{B,1}=0$ and $\aleph^{g_z}_{B,1}=1$: 

{\tiny 
\begin{equation*}
        \begin{split}
        \mathbb{E}\Bigg[W^{g_z}_{B, 1}B_1\Bigg] =  \sum_{l_1} & \begin{Bmatrix*}[l]
        & \begin{pmatrix*}[l]
            1-
                 \Bigg(
                    \frac{
                        1- q^{g_z}_1 \times P(B_1= 1)
                         }{
                        P(B_1= 0)
                         }
                 \Bigg) \times  
                 P(B_1=0 \mid  l_1)        \end{pmatrix*} \\
 \times &  P(L_1=l_1) 
 \end{Bmatrix*} \\
         =  \sum_{l_1} & \begin{Bmatrix*}[l]
        & 
             P(L_1=l_1)-
                 \Bigg(
                    \frac{
                        1- q^{g_z}_1 \times P(B_1= 1)
                         }{
                        P(B_1= 0)
                         }
                 \Bigg) \times  
                 P(B_1=0 \mid l_1) P(L_1=l_1)        \\
 \end{Bmatrix*} \\
         =  & 1 - (1- q^{g_z}_1 \times P(B_1= 1))      \\
         = & q^{g_z}_1 \times P(B_1= 1)
        \end{split}
    \end{equation*}
}%

Therefore, 

\begin{itemize}
    \item [1.]  If $\beth^{g_z}_{B,1}=1$  then $\mathbb{E}\Bigg[W^{g_z}_{B, 1}B_1\Bigg]=q^{g_z}_1 \times \mathbb{E}[B_1]$ and 
    \item [2.]  If $\beth^{g_z}_{B,1}=0$ and $\aleph^{g_z}_{B,1}=1$ then $\mathbb{E}\Bigg[W^{g_z}_{B, 1}B_1\Bigg] =q^{g_z}_1 \times \mathbb{E}[B_1]$
    \item [3.]  If $\beth^{g_z}_{B,1}=0$ and $\aleph^{g_z}_{B,1}=0$ then $\mathbb{E}\Bigg[W^{g_z}_{B, 1}B_1\Bigg]= \mathbb{E}[R_1]=\mathbb{E}[R_1^{g_z+}]$
\end{itemize}

And so the constraint is satisfied.

\subsection{Natural suspected inferior treatment utilization in interval 1 under regime $g_z$ (i.e., $H_1^{g_z}$)}

Substitute expression for weights and $\alpha_1(z, R_1)$:

{\tiny 
\begin{equation*}
        \begin{split}
        & \mathbb{E}\Bigg[W^{g_z}_{B, 1}H_1\Bigg]  \\
        = & \mathbb{E}\begin{bmatrix*}[l]
        & \begin{pmatrix*}[l]
         &  \begin{pmatrix*}[l] 
                 \Bigg[ \Bigg(
                    \frac{
                         q^{g_z}_1 \times P(B_1=1)
                          }{
                        P(B_1^{g_z}=1)
                          }
                \Bigg)^{\beth^{g_z}_{B,1}}                           
                \Bigg(
                    \frac{
                        1- 
                        \frac{
                            q^{g_z}_1 \times P(B_1=1)
                              }{
                            P(R_1^{g_z+}=1)
                              }
                          }{
                        P(B^{g_z}_1= 0\mid R^{g_z+}_1=1)
                          }
                \Bigg)^{(1- \beth^{g_z}_{B,1})  }      
                \Bigg(
                    \aleph^{g_z}_{B,1}
                \Bigg) \Bigg]^{R_1}
                P(B_1=1 \mid R_1,L_1)
           \end{pmatrix*}^{B_1}       \\
       \times 
          & \begin{pmatrix*}[l] 
             1- 
                \Bigg[ \Bigg(
                    \frac{
                         q^{g_z}_1 \times P(B_1=1)
                         }{
                        P(B_1^{g_z}=1)
                         }
                \Bigg)^{\beth^{g_z}_{B,1}}                           
                \Bigg(
                    \frac{
                        1- 
                        \frac{
                            q^{g_z}_1 \times P(B_1=1)
                              }{
                            P(R_1^{g_z+}=1)
                              }                         }{
                        P(B^{g_z}_1= 0\mid R^{g_z+}_1=1)
                         }
                \Bigg)^{(1- \beth^{g_z}_{B,1})  }     
                \Bigg(
                    \aleph^{g_z}_{B,1}
                \Bigg) \Bigg]^{R_1}
                P(B_1=1 \mid R_1,  L_1)
            \end{pmatrix*} ^{(1-B_1)}    \\
       \times & 
            \frac{
                   1
                 }{          
                P(B_1=1 \mid R_1,  L_1)^{B_1}
                P(B_1=0 \mid R_1,  L_1)^{(1-B_1)}}
        \end{pmatrix*}^{\beth^{g_z}_{B,1}} \\
 \times & \begin{pmatrix*}[l]
           & \begin{pmatrix*}[l]1- 
                \Bigg[\Bigg(
                    \frac{
                        P(B_1^{g_z}=1)
                          }{
                        P(B_1^{g_z}=1)
                          }
                 \Bigg)^{\beth^{g_z}_{B,1}}                         \times
                 \Bigg(
                    \frac{
                        1- 
                        \frac{
                            q^{g_z}_1 \times P(B_1=1)
                              }{
                            P(R_1^{g_z+}=1)
                              }                         }{
                        P(B^{g_z}_1= 0\mid R^{g_z+}_1=1)
                         }
                 \Bigg)^{(1- \beth^{g_z}_{B,1})  } \times  
                 \Bigg(
                    \aleph^{g_z}_{B,1}
                 \Bigg) \Bigg]^{R_1}
                 P(B_1=0 \mid R_1,  L_1)
            \end{pmatrix*}^{B_1} \\
    \times & \begin{pmatrix*}[l] 
                \Bigg[\Bigg(
                    \frac{
                         q^{g_z}_1 \times P(B_1=1)
                        }{
                        P(B_1^{g_z}=0)
                        }\Bigg)^{\beth^{g_z}_{B,1}}                 \times
                \Bigg(
                    \frac{
                        1- 
                        \frac{
                            q^{g_z}_1 \times P(B_1=1)
                              }{
                            P(R_1^{g_z+}=1)
                              }                         }{
                        P(B^{g_z}_1= 0\mid R^{g_z+}_1=1)
                         }
                \Bigg)^{(1- \beth^{g_z}_{B,1})  }   \times  
                \Bigg(
                    \aleph^{g_z}_{B,1}
                \Bigg) \Bigg]^{R_1}
                P(B_1=0 \mid R_1,  L_1)
            \end{pmatrix*}^{(1-B_1)} \\
    \times & \frac{
                    1
                 }{ 
                    P(B_1=1 \mid R_1,  L_1)^{B_1}
                    P(B_1=0 \mid R_1,  L_1)^{(1-B_1)}
                 }
    \end{pmatrix*}^{(1-\beth^{g_z}_{B,1})} \\
  \times  &   H_1 \end{bmatrix*} \\
        \end{split}
    \end{equation*}
}%

Evaluating total expectation, conditional on  $B_1$, $R_1$ and $L_1$. Since  $H_1$ is binary, and since $H_1=1$ implies $B_1 =0$ and $R_1=1$, then terms exponentiated by $B_1$ drop out of the sum, and terms exponentiated by $R_1$ are retained and the $R_1$ exponent is dropped. As in previous steps, terms exponentiated by indicator functions are removed or retained, as appropriate. 

{\tiny 
\begin{equation*}
        \begin{split}
        =  \sum_{l_1} \begin{Bmatrix*}[l]
        & \begin{pmatrix*}[l]
                    \frac{1-
                    \Bigg(
                        \frac{
                             q^{g_z}_1 \times P(B_1=1)
                              }{
                            P(B_1^{g_z}=1)
                              }
                    \Bigg) 
                    P(B_1=1 \mid R_1=1,l_1) 
                 }{
                    P(B_1=0 \mid R_1=1,  l_1)
                 }
        \end{pmatrix*}^{\beth^{g_z}_{B,1}} \\
 \times & \begin{pmatrix*}[l]
 \frac{
                 \Bigg(
                    \frac{
                        1- 
                        \frac{
                            q^{g_z}_1 \times P(B_1=1)
                              }{
                            P(R_1^{g_z+}=1)
                              }                         
                         }{
                        P(B^{g_z}_1= 0\mid R^{g_z+}_1=1)
                         }
                 \Bigg) \times  
                 \Bigg(
                    \aleph^{g_z}_{B,1}
                 \Bigg) 
                 P(B_1=0 \mid R_1=1,  l_1)
                 }{ 
                    P(B_1=0 \mid R_1=1,  l_1)
                 }
    \end{pmatrix*}^{(1-\beth^{g_z}_{B,1})} \\
   \times  & P(H_1=1 \mid B_1=0, R_1=1, l_1) P(B_1=0 \mid R_1=1, l_1)P(R_1=1 \mid l_1) P(L_1=l_1) \end{Bmatrix*} \\
        \end{split}
    \end{equation*}
}%

Since by consistency $B_1^{g_z}=B_1$ and $R_1^{g_z+}=R_1=1$, and re-arranging and cancelling terms:

{\tiny 
\begin{equation*}
        \begin{split}
        =  \sum_{l_1} & \begin{Bmatrix*}[l]
        & \begin{pmatrix*}[l]
            1-q^{g_z}_1 \times  P(B_1=1 \mid  l_1)
        \end{pmatrix*}^{\beth^{g_z}_{B,1}} \\
 \times & \begin{pmatrix*}[l]
         \frac{
            1- q^{g_z}_1 \times P(B_1=1 )
             }{
            P(B_1= 0)
             }
          \Bigg(
            \aleph^{g_z}_{B,1}
         \Bigg) 
         P(B_1=0 \mid   l_1)
    \end{pmatrix*}^{(1-\beth^{g_z}_{B,1})} \\
   \times  & P(H_1=1 \mid B_1=0,  l_1) P(L_1=l_1) \end{Bmatrix*} \\
        \end{split}
    \end{equation*}
}%

Noting that when $\beth^{g_z}_{B,1}=0$ and $\aleph^{g_z}_{B,1}=1$, then 

{\tiny 
\begin{equation*}
        \begin{split}
        \mathbb{E}\Bigg[W^{g_z}_{B, 1}H_1\Bigg] = & \sum_{l_1} \begin{Bmatrix*}[l]
        & \begin{pmatrix*}[l]
         \frac{
            1- q^{g_z}_1 \times P(B_1=1 )
             }{
            P(B_1= 0)
             }
         P(B_1=0 \mid  l_1)
    \end{pmatrix*} \\
   \times  & P(H_1=1 \mid B_1=0,  l_1)  P(L_1=l_1) \end{Bmatrix*} \\
        =  & \frac{
            1- q^{g_z}_1 \times P(B_1=1 )
             }{
            P(B_1= 0)
             } P(H_1=1) \\
        =  & \frac{
            1- q^{g_z}_1 \times P(B_1=1)
             }{
            P(B_1= 0)
             } P(H_1=1 \mid B=0 ) P(B_1= 0) \\
        =  & P(H_1=1 \mid B_1=0) (1- q^{g_z}_1 \times P(B_1=1))  \\
        =  & (1- q^{g_z}_1 \times )P(H_1=1 \mid B_1=0) + q^{g_z}_1 \times P(H_1=1)\\
        \end{split}
    \end{equation*}
}%

Noting that when $\beth^{g_z}_{B,1}=1$ and $\aleph^{g_z}_{B,1}=1$, then 
{\tiny\begin{equation*} 
   \begin{split}
    \mathbb{E}\Bigg[W^{g_z}_{B, 1}H_1\Bigg] = & \sum_{l_1}\Bigg\{  P(H_1=1 \mid B_1=0, L_1=l_1)\Bigg(1- q^{g_z}_1 \times P(B_1=1 \mid  l_1)\Bigg) P(L_1=l_1)\Bigg\}\\
    = & \begin{Bmatrix*}[l]
       & (1-q^{g_z}_1 \times )\sum_{l_1}\Bigg\{P(H_1=1 \mid B_1=0, L_1=l_1)P(L_1=l_1)\Bigg\} \\
       & + q^{g_z}_1 \times P(H_1=1) 
    \end{Bmatrix*}
 \end{split}
\end{equation*}}%

Therefore, 

\begin{itemize}
    \item [4.]  If $\beth^{g_z}_{B,1}=1$ and $\aleph^{g_z}_{B,1}=1$ then:
{\tiny\begin{equation*} \mathbb{E}\Bigg[W^{g_z}_{B, 1}H_1\Bigg]=\sum_{l_1}\begin{Bmatrix*}[l] & P(H_1=1 \mid B_1=0, L_1=l_1) \Bigg(1- q^{g_z}_1 \times P(B_1=1 \mid  l_1)\Bigg) P(L_1=l_1)\end{Bmatrix*}\end{equation*}}%
    \item [5.]  If $\beth^{g_z}_{B,1}=0$ and $\aleph^{g_z}_{B,1}=1$ then:
{\tiny\begin{equation*} \mathbb{E}\Bigg[W^{g_z}_{B, 1}H_1\Bigg]=P(H_1=1 \mid B_1=0) \big(1- q^{g_z}_1 \times P(B_1=1)\big)\end{equation*}}%
\item [6.]  If $\beth^{g_z}_{B,1}=0$ and $\aleph^{g_z}_{B,1}=0$ then:
{\tiny\begin{equation*} \mathbb{E}\Bigg[W^{g_z}_{B, 1}H_1\Bigg]=0\end{equation*}}%
\end{itemize}

\subsection{Suspected inferior treatment utilization in interval 1 under regime $g_z$, after intervention (i.e., $H_1^{g_z+}$)}

Substitute expression for weights, $\alpha_1(z, R_1)$ and $\beta_1(z, S_1)$. 

{\tiny 
\begin{equation*}
        \begin{split}
        &  \mathbb{E}\Bigg[W^{g_z}_{B, 1}W^{g_z}_{H, 1}H_1\Bigg]  \\
        = & \mathbb{E}\begin{bmatrix*}[l]
        & \begin{pmatrix*}[l]
         &  \begin{pmatrix*}[l] 
                 \Bigg[ \Bigg(
                    \frac{
                         q^{g_z}_1 \times P(B_1=1)
                          }{
                        P(B_1^{g_z}=1)
                          }
                \Bigg)^{\beth^{g_z}_{B,1}}                           
                \Bigg(
                    \frac{
                        1- 
                        \frac{
                            q^{g_z}_1 \times P(B_1=1)
                              }{
                            P(R_1^{g_z+}=1)
                              }
                          }{
                        P(B^{g_z}_1= 0\mid R^{g_z+}_1=1)
                          }
                \Bigg)^{(1- \beth^{g_z}_{B,1}) }      
                \Bigg(
                    \aleph^{g_z}_{B,1}
                \Bigg) \Bigg]^{R_1}
                P(B_1=1 \mid R_1,L_1)
           \end{pmatrix*}^{B_1}       \\
       \times 
          & \begin{pmatrix*}[l] 
             1- 
                \Bigg[ \Bigg(
                    \frac{
                         q^{g_z}_1 \times P(B_1=1)
                         }{
                        P(B_1^{g_z}=1)
                         }
                \Bigg)^{\beth^{g_z}_{B,1}}                           
                \Bigg(
                    \frac{
                        1- 
                        \frac{
                            q^{g_z}_1 \times P(B_1=1)
                              }{
                            P(R_1^{g_z+}=1)
                              }                         }{
                        P(B^{g_z}_1= 0\mid R^{g_z+}_1=1)
                         }
                \Bigg)^{(1- \beth^{g_z}_{B,1})   }     
                \Bigg(
                    \aleph^{g_z}_{B,1}
                \Bigg) \Bigg]^{R_1}
                P(B_1=1 \mid R_1,  L_1)
            \end{pmatrix*} ^{(1-B_1)}    \\
       \times & 
            \frac{
                   1
                 }{          
                P(B_1=1 \mid R_1,  L_1)^{B_1}
                P(B_1=0 \mid R_1,  L_1)^{(1-B_1)}}
        \end{pmatrix*}^{\beth^{g_z}_{B,1}} \\
 \times & \begin{pmatrix*}[l]
           & \begin{pmatrix*}[l]1- 
                \Bigg[\Bigg(
                    \frac{
                        P(B_1^{g_z}=1)
                          }{
                        P(B_1^{g_z}=1)
                          }
                 \Bigg)^{\beth^{g_z}_{B,1}}                         \times
                 \Bigg(
                    \frac{
                        1- 
                        \frac{
                            q^{g_z}_1 \times P(B_1=1)
                              }{
                            P(R_1^{g_z+}=1)
                              }                         }{
                        P(B^{g_z}_1= 0\mid R^{g_z+}_1=1)
                         }
                 \Bigg)^{(1- \beth^{g_z}_{B,1}) } \times  
                 \Bigg(
                    \aleph^{g_z}_{B,1}
                 \Bigg) \Bigg]^{R_1}
                 P(B_1=0 \mid R_1,  L_1)
            \end{pmatrix*}^{B_1} \\
    \times & \begin{pmatrix*}[l] 
                \Bigg[\Bigg(
                    \frac{
                         q^{g_z}_1 \times P(B_1=1)
                        }{
                        P(B_1^{g_z}=0)
                        }\Bigg)^{\beth^{g_z}_{B,1}}                 \times
                \Bigg(
                    \frac{
                        1- 
                        \frac{
                            q^{g_z}_1 \times P(B_1=1)
                              }{
                            P(R_1^{g_z+}=1)
                              }                         }{
                        P(B^{g_z}_1= 0\mid R^{g_z+}_1=1)
                         }
                \Bigg)^{(1- \beth^{g_z}_{B,1}) }   \times  
                \Bigg(
                    \aleph^{g_z}_{B,1}
                \Bigg) \Bigg]^{R_1}
                P(B_1=0 \mid R_1,  L_1)
            \end{pmatrix*}^{(1-B_1)} \\
    \times & \frac{
                    1
                 }{ 
                    P(B_1=1 \mid R_1,  L_1)^{B_1}
                    P(B_1=0 \mid R_1,  L_1)^{(1-B_1)}
                 }
    \end{pmatrix*}^{(1-\beth^{g_z}_{B,1})} \\
  \times  & \begin{pmatrix*}[l]
         &  \begin{pmatrix*}[l] 
                 \Bigg[ \Bigg(
                    \frac{
                         m^{g_z}_1 \times P(H_1=1)
                          }{
                        P(H_1^{g_z}=1)
                          }
                \Bigg)^{\beth^{g_z}_{H,1}}                           
                \Bigg(
                    \frac{
                        1- 
                        \frac{
                            m^{g_z}_1 \times P(H_1=1)
                              }{
                            P(R_1^{g_z+}=1)
                              }
                          }{
                        P(H^{g_z}_1= 0\mid S^{g_z+}_1=1)
                          }
                \Bigg)^{(1- \beth^{g_z}_{H,1})  }      
                \Bigg(
                    \aleph^{g_z}_{H,1}
                \Bigg) \Bigg]^{S_1}
                P(H_1=1 \mid S_1,L_1)
           \end{pmatrix*}^{H_1}       \\
       \times 
          & \begin{pmatrix*}[l] 
             1- 
                \Bigg[ \Bigg(
                    \frac{
                         m^{g_z}_1 \times P(H_1=1)
                         }{
                        P(H_1^{g_z}=1)
                         }
                \Bigg)^{\beth^{g_z}_{H,1}}                           
                \Bigg(
                    \frac{
                        1- 
                        \frac{
                            m^{g_z}_1 \times P(H_1=1)
                              }{
                            P(R_1^{g_z+}=1)
                              }                         }{
                        P(H^{g_z}_1= 0\mid R^{g_z+}_1=1)
                         }
                \Bigg)^{(1- \beth^{g_z}_{H,1})   }     
                \Bigg(
                    \aleph^{g_z}_{H,1}
                \Bigg) \Bigg]^{S_1}
                P(H_1=1 \mid S_1,  L_1)
            \end{pmatrix*} ^{(1-H_1)}    \\
       \times & 
            \frac{
                   1
                 }{          
                P(H_1=1 \mid S_1,  L_1)^{H_1}
                P(H_1=0 \mid S_1,  L_1)^{(1-H_1)}}
        \end{pmatrix*}^{\beth^{g_z}_{H,1}} \\
 \times & \begin{pmatrix*}[l]
           & \begin{pmatrix*}[l]1- 
                \Bigg[\Bigg(
                    \frac{
                        P(H_1^{g_z}=1)
                          }{
                        P(H_1^{g_z}=1)
                          }
                 \Bigg)^{\beth^{g_z}_{H,1}}                         \times
                 \Bigg(
                    \frac{
                        1- 
                        \frac{
                            m^{g_z}_1 \times P(H_1=1)
                              }{
                            P(R_1^{g_z+}=1)
                              }                         }{
                        P(H^{g_z}_1= 0\mid S^{g_z+}_1=1)
                         }
                 \Bigg)^{(1- \beth^{g_z}_{H,1})  } \times  
                 \Bigg(
                    \aleph^{g_z}_{H,1}
                 \Bigg) \Bigg]^{S_1}
                 P(H_1=0 \mid S_1,  L_1)
            \end{pmatrix*}^{H_1} \\
    \times & \begin{pmatrix*}[l] 
                \Bigg[\Bigg(
                    \frac{
                         m^{g_z}_1 \times P(H_1=1)
                        }{
                        P(H_1^{g_z}=0)
                        }\Bigg)^{\beth^{g_z}_{H,1}}                 \times
                \Bigg(
                    \frac{
                        1- 
                        \frac{
                            m^{g_z}_1 \times P(H_1=1)
                              }{
                            P(R_1^{g_z+}=1)
                              }                         }{
                        P(H^{g_z}_1= 0\mid S^{g_z+}_1=1)
                         }
                \Bigg)^{(1- \beth^{g_z}_{H,1}) }   \times  
                \Bigg(
                    \aleph^{g_z}_{H,1}
                \Bigg) \Bigg]^{S_1}
                P(H_1=0 \mid S_1,  L_1)
            \end{pmatrix*}^{(1-H_1)} \\
    \times & \frac{
                    1
                 }{ 
                    P(H_1=1 \mid S_1,  L_1)^{H_1}
                    P(H_1=0 \mid S_1,  L_1)^{(1-H_1)}
                 }
    \end{pmatrix*}^{(1-\beth^{g_z}_{H,1})} \\
  \times  &   H_1 \end{bmatrix*} \\
        \end{split}
    \end{equation*}
}%

Evaluating total expectation, conditional on  $B_1$, $R_1$ and $L_1$. Since $H_1$ is binary, and since $H_1=1$ implies $S_1 =1$, $B_1=0$ and $R_1=1$, then terms exponentiated by $(1-S_1)$ and by $B_1$ drop out of the sum, and terms exponentiated by $R_1$ and $S_1$ are retained and the $R_1$ and $S_1$ exponents are dropped. As in previous steps, terms exponentiated by indicator functions are removed or retained, as appropriate. Noting also that $B_1^{g_z}=B_1$ and $R_1^{g_z+}=R_1=1$, then,

{\tiny 
\begin{equation*}
        \begin{split}
        = & \sum_{l_1}\begin{Bmatrix*}[l]
        & \begin{pmatrix*}[l]
            \frac{
                   1-  \Bigg(
                    \frac{
                         q^{g_z}_1 \times P(B_1=1)
                         }{
                        P(B_1=1)
                         }
                \Bigg)
                P(B_1=1 \mid  l_1)
                 }{
                 P(B_1=0 \mid l_1)
                 }
        \end{pmatrix*}^{\beth^{g_z}_{B,1}} \\
 \times & \begin{pmatrix*}[l]
             \frac{
                \Bigg(
                    \frac{
                        1- q^{g_z}_1 \times P(B_1=1)                        }{
                        P(B_1=1)
                         }
                \Bigg)   \times  
                \Bigg(
                    \aleph^{g_z}_{B,1}
                \Bigg)
                P(B_1=0 \mid l_1)
                 }{ 
                    P(B_1=0 \mid l_1)
                 }
    \end{pmatrix*}^{(1-\beth^{g_z}_{B,1})} \\
  \times  & \begin{pmatrix*}[l]
            \frac{
                \Bigg(
                    \frac{
                         m^{g_z}_1 \times P(H_1=1)
                          }{
                        P(H_1^{g_z}=1)
                          }
                \Bigg)     
                P(H_1=1 \mid B_1=0,l_1)
                 }{          
                P(H_1=1 \mid B_1=0,  l_1)
                }
        \end{pmatrix*}^{\beth^{g_z}_{H,1}} \\
 \times & \begin{pmatrix*}[l]
         & \frac{1-
                 \Bigg(
                    \frac{
                        1- 
                        \frac{
                            m^{g_z}_1 \times P(H_1=1)
                              }{
                            P(B_1^{g_z}=1)
                              }                         }{
                        P(H^{g_z}_1= 0\mid B^{g_z}_1=0)
                         }
                 \Bigg) \times  
                 \Bigg(
                    \aleph^{g_z}_{H,1}
                 \Bigg) 
                 P(H_1=0 \mid B_1=0,  l_1)
                 }{ 
                    P(H_1=1 \mid B_1=0,  l_1)
                 }
    \end{pmatrix*}^{(1-\beth^{g_z}_{H,1})} \\
  \times  &  P(H_1=1 \mid B_1=0,  l_1) P(B_1=0 \mid  l_1) P(L_1=l_1) \end{Bmatrix*} \\
        \end{split}
    \end{equation*}
}%

Noting that $P(H_1=1 \mid B_1=0) = \frac{P(H_1=1)}{P(B_1=0)}$, and re-arranging terms and cancelling some terms as in previous steps.

{\tiny 
\begin{equation*}
        \begin{split}
        = & \sum_{l_1}\begin{Bmatrix*}[l]
        & \begin{pmatrix*}[l]
            \frac{
                   1-  q^{g_z}_1 \times P(B_1=1 \mid   l_1)
                 }{
                 P(B_1=0 \mid  l_1)
                 }
        \end{pmatrix*}^{\beth^{g_z}_{B,1}} \\
 \times & \begin{pmatrix*}[l]
                \Bigg(
                    \frac{
                        1- q^{g_z}_1 \times P(B_1=1))
                         }{
                        P(B_1= 0)
                         }
                \Bigg)   \times  
                \Bigg(
                    \aleph^{g_z}_{B,1}
                \Bigg)
    \end{pmatrix*}^{(1-\beth^{g_z}_{B,1})} \\
  \times  & \begin{pmatrix*}[l]
            \frac{
                \Bigg(
                    \frac{
                         m^{g_z}_1 \times P(H_1=1)
                          }{
                        P(H_1^{g_z}=1)
                          }
                \Bigg)     
                P(H_1=1 \mid B_1=0,l_1)
                 }{          
                P(H_1=1 \mid B_1=0,  l_1)
                }
        \end{pmatrix*}^{\beth^{g_z}_{H,1}} \\
 \times & \begin{pmatrix*}[l]
         & \frac{1-
                 \Bigg(
                    \frac{
                        1- 
                        \frac{
                            m^{g_z}_1 \times P(H_1=1)
                              }{
                            P(B_1^{g_z}=0)
                              }                         }{
                        1- 
                        \frac{
                            P(H_1^{g_z}=1)
                              }{
                            P(B_1^{g_z}=0)
                              }
                         }
                 \Bigg) \times  
                 \Bigg(
                    \aleph^{g_z}_{H,1}
                 \Bigg) 
                 P(H_1=0 \mid B_1=0,  l_1)
                 }{ 
                    P(H_1=1 \mid B_1=0,  l_1)
                 }
    \end{pmatrix*}^{(1-\beth^{g_z}_{H,1})} \\
  \times  &  P(H_1=1 \mid B_1=0, l_1) P(B_1=0 \mid  l_1) P(L_1=l_1) \end{Bmatrix*} \\
        \end{split}
    \end{equation*}
}%

Noting that when $\beth^{g_z}_{B,1}=1$ and $\beth^{g_z}_{H,1}=1$, then

{\tiny 
\begin{equation*}
        \begin{split}
      \mathbb{E}\Bigg[W^{g_z}_{B, 1}W^{g_z}_{H, 1}H_1\Bigg]  = & \sum_{l_1}\begin{Bmatrix*}[l]
        & \begin{pmatrix*}[l]
            \frac{
                   1-  q^{g_z}_1 \times P(B_1=1 \mid  l_1)
                 }{
                 P(B_1=0 \mid   l_1)
                 }
        \end{pmatrix*} \\
\times  & \begin{pmatrix*}[l]
                \frac{
                     m^{g_z}_1 \times P(H_1=1)
                      }{
                    \sum_{l_1}\begin{Bmatrix*}[l] & P(H_1=1 \mid  B_1=0, L_1=l_1)\Bigg(1- q^{g_z}_1 \times P(B_1=1 \mid  l_1)\Bigg) P(L_1=l_1)\end{Bmatrix*}
                      }
        \end{pmatrix*} \\
  \times  &  P(H_1=1 \mid  B_1=0, l_1) P(B_1=0 \mid l_1) P(L_1=l_1) \end{Bmatrix*} \\
        = & \sum_{l_1}\begin{Bmatrix*}[l]
        & \begin{pmatrix*}[l]
                \frac{
                     m^{g_z}_1 \times P(H_1=1)
                      }{
                    \sum_{l_1}\begin{Bmatrix*}[l] & P(H_1=1 \mid  B_1=0, L_1=l_1) \Bigg(1- q^{g_z}_1 \times P(B_1=1 \mid  l_1)\Bigg) P(L_1=l_1)\end{Bmatrix*}
                      }
        \end{pmatrix*} \\
  \times  &  P(H_1=1 \mid B_1=0, l_1) \Bigg(1- q^{g_z}_1 \times P(B_1=1 \mid  l_1)\Bigg)P(L_1=l_1)
  \end{Bmatrix*} \\
    =& m^{g_z}_1 \times P(H_1=1)
        \end{split}
    \end{equation*}
}%

Noting that when $\beth^{g_z}_{B,1}=1$ and $\beth^{g_z}_{H,1}=0$ and $\aleph^{g_z}_{H,1}=1$, then

{\tiny 
\begin{equation*}
        \begin{split}
              &      \frac{
                        1- 
                        \frac{
                            m^{g_z}_1 \times P(H_1=1)
                              }{
                            P(B_1^{g_z}=0)
                              }                         }{
                        1- 
                        \frac{
                            P(H_1^{g_z}=1)
                              }{
                            P(B_1^{g_z}=0)
                              }
                         } \\
= &  \frac{
                        1- 
                        \frac{
                            m^{g_z}_1 \times P(H_1=1)
                              }{
                            1-q^{g_z}_1 \times P(B_1=1)
                              }                         }{
                        1- 
                        \frac{
                            \sum_{l_1}\begin{Bmatrix*}[l] & P(H_1=1 \mid B_1=0, L_1=l_1) \Bigg(1- q^{g_z}_1 \times P(B_1=1 \mid l_1)\Bigg) P(L_1=l_1)\end{Bmatrix*}
                              }{
                            1-q^{g_z}_1 \times P(B_1=1)
                              }
                         }\\
= &  \frac{
                        1-q^{g_z}_1 \times P(B_1=1) -
                            m^{g_z}_1 \times P(H_1=1)
                        }{
                        1-q^{g_z}_1 \times P(B_1=1) -
                            (1-q)\sum_{l_1}\{ P(H_1=1 \mid B_1=0, L_1=l_1)P(L_1=l_1) \} - q^{g_z}_1 \times P(H_1=1)
                         }
        \end{split}
    \end{equation*}
}%

And thus:

{\tiny 
\begin{equation*}
        \begin{split}
   \mathbb{E}\Bigg[W^{g_z}_{B, 1}W^{g_z}_{H, 1}H_1\Bigg] = & \sum_{l_1}\begin{Bmatrix*}[l]
& \begin{pmatrix*}[l]
         & \frac{1-
                 \Bigg(
                    \frac{
                        1- 
                        \frac{
                            m^{g_z}_1 \times P(H_1=1)
                              }{
                            P(B_1^{g_z}=1)
                              }                         }{
                        1- 
                        \frac{
                            P(H_1^{g_z}=1)
                              }{
                            P(B_1^{g_z}=1)
                              }
                         }
                 \Bigg) \times  
                 P(H_1=0 \mid B_1=0,  l_1)
                 }{ 
                    P(H_1=1 \mid B_1=0,  l_1)
                 }
    \end{pmatrix*} \\
    \times  &  P(H_1=1 \mid  B_1=0, l_1) \Bigg(1- q^{g_z}_1 \times P(B_1=1 \mid  l_1)\Bigg) P(L_1=l_1)
  \end{Bmatrix*} \\
        \end{split}
    \end{equation*}
}%

{\tiny 
\begin{equation*}
        \begin{split}
    = & \sum_{l_1}\begin{Bmatrix*}[l]
& \begin{pmatrix*}[l]
         & 1-
                 \Bigg(
                    \frac{
                        1- 
                        \frac{
                            m^{g_z}_1 \times P(H_1=1)
                              }{
                            P(B_1^{g_z}=1)
                              }                         }{
                        1- 
                        \frac{
                            P(H_1^{g_z}=1)
                              }{
                            P(B_1^{g_z}=1)
                              }
                         }
                 \Bigg) \times  
                 P(H_1=0 \mid B_1=0,  l_1)
    \end{pmatrix*} \\
    \times  &  \Bigg(1- q^{g_z}_1 \times P(B_1=1 \mid  l_1)\Bigg) P(L_1=l_1)
  \end{Bmatrix*} \\
        \end{split}
    \end{equation*}
}%

{\tiny 
\begin{equation*}
        \begin{split}
    = & \sum_{l_1}\begin{Bmatrix*}[l]
    \times  &  \Bigg(1- q^{g_z}_1 \times P(B_1=1 \mid  l_1)\Bigg) P(L_1=l_1)
  \end{Bmatrix*} \\
 - & \sum_{l_1}\begin{Bmatrix*}[l]
        & \begin{pmatrix*}[l]
                    \frac{
                        1- 
                        \frac{
                            m^{g_z}_1 \times P(H_1=1)
                              }{
                            P(B_1^{g_z}=1)
                              }                         }{
                        1- 
                        \frac{
                            P(H_1^{g_z}=1)
                              }{
                            P(B_1^{g_z}=1)
                              }
                         }
            \end{pmatrix*} \\
    \times  &  \Bigg(1- P(H_1=1 \mid B_1=0,  l_1)\Bigg)\Bigg((1- q^{g_z}_1 \times ) + q^{g_z}_1 \times P(B_1=0 \mid  l_1)\Bigg) P(L_1=l_1)
  \end{Bmatrix*} \\
        \end{split}
    \end{equation*}
}%

{\tiny 
\begin{equation*}
        \begin{split}
    = & \Bigg(1 - q^{g_z}_1 \times P(B_1=1) \Bigg) \\
 - & \sum_{l_1}\begin{Bmatrix*}[l]
        & \begin{pmatrix*}[l]
            \frac{
                        1-q^{g_z}_1 \times P(B_1=1) -
                            m^{g_z}_1 \times P(H_1=1)
                        }{
                        1-q^{g_z}_1 \times P(B_1=1) -
                            (1-q)\sum_{l_1}\{ P(H_1=1 \mid B_1=0, L_1=l_1)P(L_1=l_1) \} - q^{g_z}_1 \times P(H_1=1)
                         }
            \end{pmatrix*} \\
    \times    &\begin{pmatrix*}[l] 
    & (1- q^{g_z}_1 \times P(B_1=1)) \\
-   & (1- q^{g_z}_1 \times )\sum_{l_1}P(H_1=1 \mid  B_1=0,  L_1=l_1)P(L_1=l_1) \\
-   & q^{g_z}_1 \times P(H_1=1)
        \end{pmatrix*}
  \end{Bmatrix*} \\
        \end{split}
    \end{equation*}
}%

{\tiny 
\begin{equation*}
        \begin{split}
    = & \Bigg(1 - q^{g_z}_1 \times P(B_1=1) \Bigg) \\
 - & \begin{pmatrix*}[l]
        1-q^{g_z}_1 \times P(B_1=1) - m^{g_z}_1 \times P(H_1=1)
  \end{pmatrix*} \\
  =& m^{g_z}_1 \times P(H_1=1)
        \end{split}
    \end{equation*}
}%

Noting that when $\beth^{g_z}_{B,1}=0$ and $\beth^{g_z}_{H,1}=0$ and $\aleph^{g_z}_{H,1}=1$, then

{\tiny 
\begin{equation*}
        \begin{split}
              &      \frac{
                        1- 
                        \frac{
                            m^{g_z}_1 \times P(H_1=1)
                              }{
                            P(B_1^{g_z}=0)
                              }                         }{
                        1- 
                        \frac{
                            P(H_1^{g_z}=1)
                              }{
                            P(B_1^{g_z}=0)
                              }
                         } \\
= &  \frac{
                        1- 
                        \frac{
                            m^{g_z}_1 \times P(H_1=1)
                              }{
                            1-q^{g_z}_1 \times P(B_1=1)
                              }                         }{
                        1- 
                        \frac{ P(H_1=1 \mid B_1=0) (1- q^{g_z}_1 \times P(B_1=1))
                              }{
                            1-q^{g_z}_1 \times P(B_1=1)
                              }
                         }\\
= &  \frac{
                        1-q^{g_z}_1 \times P(B_1=1) -
                            m^{g_z}_1 \times P(H_1=1)
                        }{
                        1-q^{g_z}_1 \times P(B_1=1) - (1-q^{g_z}_1 \times )P(H_1=1 \mid B_1=0) -q^{g_z}_1 \times P(H_1=1)
                         }
        \end{split}
    \end{equation*}
}%

{\tiny 
\begin{equation*}
        \begin{split}
= &    \frac{
                           1-q^{g_z}_1 \times P(B_1=1)- m^{g_z}_1 \times P(H_1=1)
                              }{
                           P(H_1=0 \mid B_1=0) (1-q^{g_z}_1 \times P(B_1=1))
                         }\\
        \end{split}
    \end{equation*}
}%
Thus:

{\tiny 
\begin{equation*}
        \begin{split}
      \mathbb{E}\Bigg[W^{g_z}_{B, 1}W^{g_z}_{H, 1}H_1\Bigg]  = & \sum_{l_1}\begin{Bmatrix*}[l]
        & \begin{pmatrix*}[l]
                \Bigg(
                    \frac{
                        1- q^{g_z}_1 \times P(B_1=1)
                         }{
                        P(B_1= 0)
                         }
                \Bigg)   \times  
    \end{pmatrix*} \\
 \times & \begin{pmatrix*}[l]
         & \frac{1-
                 \Bigg(
                    \frac{
                        1- 
                        \frac{
                            m^{g_z}_1 \times P(H_1=1)
                              }{
                            P(B_1^{g_z}=0)
                              }                         }{
                        1- 
                        \frac{
                            P(H_1^{g_z}=1)
                              }{
                            P(B_1^{g_z}=0)
                              }
                         }
                 \Bigg) \times  
                 P(H_1=0 \mid B_1=0,  l_1)
                 }{ 
                    P(H_1=1 \mid B_1=0,  l_1)
                 }
    \end{pmatrix*} \\
  \times  &  P(H_1=1 \mid B_1=0, l_1) P(B_1=0 \mid l_1)P(L_1=l_1) \end{Bmatrix*} \\
        \end{split}
    \end{equation*}
}%

{\tiny 
\begin{equation*}
        \begin{split}
        =         & \begin{pmatrix*}[l]
                    \frac{
                        1- q^{g_z}_1 \times P(B_1=1)
                         }{
                        P(B_1= 0)
                         }
    \end{pmatrix*} \\ 
    \times& \sum_{l_1}\begin{Bmatrix*}[l]
    & \begin{pmatrix*}[l]
         & 1-
                 \Bigg(
                    \frac{
                        1- 
                        \frac{
                            m^{g_z}_1 \times P(H_1=1)
                              }{
                            P(B_1^{g_z}=0)
                              }                         }{
                        1- 
                        \frac{
                            P(H_1^{g_z}=1)
                              }{
                            P(B_1^{g_z}=0)
                              }
                         }
                 \Bigg) \times  
                 P(H_1=0 \mid B_1=0,  l_1)
    \end{pmatrix*} \\
  \times  &  P(B_1=0 \mid l_1)P(L_1=l_1) \end{Bmatrix*} \\
        =         & \begin{pmatrix*}[l]
                    \frac{
                        1- q^{g_z}_1 \times P(B_1=1)
                         }{
                        P(B_1= 0)
                         }
    \end{pmatrix*} \\ 
    \times& \sum_{l_1}\begin{Bmatrix*}[l]
  & \begin{pmatrix*}[l]
         & P(B_1=0 \mid l_1)P(L_1=l_1)-
                 \Bigg(
                    \frac{
                        1 - 
                        \frac{
                            m^{g_z}_1 \times P(H_1=1)
                              }{
                            P(B_1^{g_z}=0)
                              }                         }{
                        1- 
                        \frac{
                            P(H_1^{g_z}=1)
                              }{
                            P(B_1^{g_z}=0)
                              }
                         }
                 \Bigg) \times  
                 P(H_1=0 \mid B_1=0,  l_1)P(B_1=0 \mid l_1)P(L_1=l_1)
    \end{pmatrix*}    \end{Bmatrix*} \\
        =         & \begin{pmatrix*}[l]
                    \frac{
                        1- q^{g_z}_1 \times P(B_1=1)
                         }{
                        P(B_1= 0)
                         }
    \end{pmatrix*} \\ 
    \times& \begin{pmatrix*}[l]
         & P(B_1=0) -
                 P(H_1=0)\frac{
                           1-q^{g_z}_1 \times P(B_1=1)- m^{g_z}_1 \times P(H_1=1)
                              }{
                           P(H_1=0 \mid B_1=0) (1-q^{g_z}_1 \times P(B_1=1))
                         }
    \end{pmatrix*} \\
        = &\begin{pmatrix*}[l]
         & 1- q^{g_z}_1 \times P(B_1=1) -
                 (1- q^{g_z}_1 \times P(B_1=1))P(H_1=0)\frac{
                           1-q^{g_z}_1 \times P(B_1=1)- m^{g_z}_1 \times P(H_1=1)
                              }{
                           P(H_1=0 \mid B_1=0)P(B_1= 0) (1-q^{g_z}_1 \times P(B_1=1))
                         }
    \end{pmatrix*} \\
        = &\begin{pmatrix*}[l]
         & 1- q^{g_z}_1 \times P(B_1=1) -
                 (1- q^{g_z}_1 \times P(B_1=1))P(H_1=0)\frac{
                           1-q^{g_z}_1 \times P(B_1=1)- m^{g_z}_1 \times P(H_1=1)
                              }{
                           P(H_1=0 ) (1-q^{g_z}_1 \times P(B_1=1))
                         }
    \end{pmatrix*} \\
        = &\begin{pmatrix*}[l]
         & 1- q^{g_z}_1 \times P(B_1=1) - \Big(1-q^{g_z}_1 \times P(B_1=1)- m^{g_z}_1 \times P(H_1=1)\Big)
    \end{pmatrix*} \\
= & m^{g_z}_1 \times P(H_1=1)
        \end{split}
    \end{equation*}
}%

Noting that when $\beth^{g_z}_{B,1}=0$ and $\beth^{g_z}_{H,1}=1$ and $\aleph^{g_z}_{H,1}=1$, then

{\tiny 
\begin{equation*}
        \begin{split}
       \mathbb{E}\Bigg[W^{g_z}_{B, 1}W^{g_z}_{H, 1}H_1\Bigg] = & \sum_{l_1}\begin{Bmatrix*}[l]
         & \begin{pmatrix*}[l]
                    \frac{
                        1- q^{g_z}_1 \times P(B_1=1 \mid R_1=1))
                         }{
                        P(B_1= 0\mid R_1=1)
                         }
    \end{pmatrix*} \\
  \times  & \begin{pmatrix*}[l]
                    \frac{
                         m^{g_z}_1 \times P(H_1=1)
                          }{
                        P(H_1^{g_z}=1)
                          }
        \end{pmatrix*} \\
  \times  &  P(H_1=1 \mid B_1=0, l_1) P(B_1=0 \mid  l_1) P(L_1=l_1) \end{Bmatrix*}\\
        = &\begin{pmatrix*}[l]
         & \begin{pmatrix*}[l]
                    \frac{
                        1- q^{g_z}_1 \times P(B_1=1)
                         }{
                        P(B_1= 0)
                         }
    \end{pmatrix*} \\
  \times  & \begin{pmatrix*}[l]
                    \frac{
                         m^{g_z}_1 \times P(H_1=1)
                          }{
                        P(H_1^{g_z}=1)
                          }
        \end{pmatrix*} \\
  \times  &  P(H_1=1) \end{pmatrix*} \\
        = &\begin{pmatrix*}[l]
         & \begin{pmatrix*}[l]
                    \frac{
                        1- q^{g_z}_1 \times P(B_1=1)
                         }{
                        P(B_1= 0)
                         }
    \end{pmatrix*} \\
  \times  & \begin{pmatrix*}[l]
                    \frac{
                         m^{g_z}_1 \times P(H_1=1)
                          }{
                        P(H_1=1 \mid B_1=0)\big(1- q^{g_z}_1 \times P(B_1=1)\big)
                          }
        \end{pmatrix*} \\
  \times  &  P(H_1=1) \end{pmatrix*} \\
        = &\begin{pmatrix*}[l]
         & \begin{pmatrix*}[l]
                    1- q^{g_z}_1 \times P(B_1=1)
    \end{pmatrix*} \\
  \times  & \begin{pmatrix*}[l]
                    \frac{
                         m^{g_z}_1 \times P(H_1=1)
                          }{
                        P(H_1=1)\big(1- q^{g_z}_1 \times P(B_1=1)\big)
                          }
        \end{pmatrix*} \\
  \times  &  P(H_1=1) \end{pmatrix*} \\
 = & m^{g_z}_1 \times P(H_1=1)
        \end{split}
    \end{equation*}
}%

Therefore: 

\begin{itemize}
    \item [1.]  If $\beth^{g_z}_{B,1}=1$ and $\beth_1^H=1$ then $\mathbb{E}\Bigg[W^{g_z}_{B, 1}W^{g_z}_{H, 1}H_1\Bigg]=m^{g_z}_1 \times \mathbb{E}[H_1]$ and 
    \item [2.]  If $\beth^{g_z}_{B,1}=0$ and $\beth_1^H=1$ then $\mathbb{E}\Bigg[W^{g_z}_{B, 1}W^{g_z}_{H, 1}H_1\Bigg]=m^{g_z}_1 \times \mathbb{E}[H_1]$ and 
    \item [3.]  If $\beth^{g_z}_{B,1}=1$ and $\beth_1^H=0$ and $\aleph^{g_z}_{H,1}=1$ then $\mathbb{E}\Bigg[W^{g_z}_{B, 1}W^{g_z}_{H, 1}H_1\Bigg]=m^{g_z}_1 \times \mathbb{E}[H_1]$ and 
    \item [4.]  If $\beth^{g_z}_{B,1}=0$ and and $\aleph^{g_z}_{B,1}=1$ $\beth^{g_z}_{B,1}=0$ then $\mathbb{E}\Bigg[W^{g_z}_{B, 1}W^{g_z}_{H, 1}H_1\Bigg]=m^{g_z}_1 \times \mathbb{E}[H_1]$ and 
    \item [5.]  If $\aleph^{g_z}_{B,1}=0$ then $\mathbb{E}\Bigg[W^{g_z}_{B, 1}W^{g_z}_{H, 1}H_1\Bigg]=0$ and 
    \item [6.]  If $\aleph^{g_z}_{H,1}=0$ then $\mathbb{E}\Bigg[W^{g_z}_{B, 1}W^{g_z}_{H, 1}H_1\Bigg]=P(B_1^{g_z}=0)=\mathbb{E}(S_1^{g_z}=1)$ 
\end{itemize}

And so the constraint is satisfied.

\clearpage

\section{Appendix C: Relationship between proportionally-representative interventions and deterministic regimes} 

Section 2 of Young, et. al (2014) \cite{young2014identification} provides a succinct review of the relationship between deterministic and random (i.e. stochastic) regimes. Specifically, a regime $g_z$ can be considered deterministic for some treatment $B_k$ with support $\mathcal{B}$, if the following condition holds for all $k$:

\begin{align}
    & \exists \ b_k \in \mathcal{B} \ s.t. \ f_{B_k^{g_z+} \mid \overline{L}_k^{g_z}, \overline{R}_{k}^{g_z+}}(a_k \mid \overline{L}_k, \overline{R}_{k}) = 1,  w.p. 1. \nonumber
\end{align}

In words, a regime $g_z$ can be considered deterministic if every individual, possibly conditional on their treatment and covariate history, receives some treatment with certainty under the regime. Otherwise, a regime is stochastic, when at least one individual, possibly conditional on their treatment and covariate history, could possibly receive more than one treatment level under the regime. 

Obviously, proportionally representative interventions are always stochastic, unless one the following condition holds for all $k$: 

\begin{align}
    & q_k^{g_z} \times P(B_k=1) > P(R_k^{g_z+}=1),
\end{align}

or

\begin{align}
    & q_k^{g_z} \times P(H_k=1) = 0.
\end{align}

In words the former condition means that there are more treatment resources available than their are treatment eligible individuals, and the latter condition means that there are no treatment resources available under regime $g_z$. If the former condition holds, for a particular $k$, then all eligible individuals will receive the treatment with certainty. We refer to this condition as one in which treatment resources are `practically unlimited'. If the latter condition holds, for a particular $k$, then all individuals will receive no treatment with certainty. We refer to this condition as one in which treatment resources are `abolished'. Under these conditions, proportionally representative interventions are static deterministic interventions, according to the definitions in Young, et. al (2014) \cite{young2014identification}. As such, an expected potential outcome under an arbitrary deterministic regime can be understood as a proportionally representative intervention in which treatment resources are either abolished or made practically unlimited for each subgroup of the target population. As a consequence, we can appreciate that the average outcome observed in one arm of an ideal randomized controlled trial, where participants are deterministically assigned treatment $A=a$, identifies the average potential outcome under a proportionally representative intervention in which treatment level $A=a$ is made to be practically unlimited for all subgroups. When such an intervention on treatment resources is not feasible, then that estimand will not be relevant for policy-making. 
\clearpage

\section{Appendix D: G-formula proof}

Here we provide a proof for the identification formula for expected potential outcomes under proportionally-representative interventions for limited resources, under the 'natural course' - that is, under an additional hypothetical intervention to prevent censoring in all individuals (i.e. $\overline{c}_K=0$). Thus, our proof also covers settings in which censoring can be present, and so we generalize the identifiability conditions from the main text to allow for censoring. Let $C_k$ be an indicator for censoring in interval $k$. By definition, all individuals are uncensored in interval 0, so $C_0=0$, and individuals that had previously been censored, stay censored, such that if $C_k=1$ then $\overline{C}_k=1$. We define a topological order within each interval that includes censoring as $\Big(L_k, B_k, H_k, C_k, Y_k\Big)$. Further, we redefine all regimes $g_z$, $z \in \mathcal{Z}$ to involve some proportionally-representative intervention, \textit{and} an intervention on $\overline{C}^{g_z}_K$ such that $\overline{C}^{g_z+}_K=0$ for all individuals.

Finally, for the sake of the proof, we consider the stochastic regime $g_z$ for $B_k$, defined by the intervention density $f_{B_k^{g_z+} \mid \overline{L}_k^{g_z}, R_{k}^{g_z+}, C_{k-1}^{g_z+}}(\cdot \mid \cdot)$ to be produced by some deterministic regime $g^-_{z}$, so that $B_k^{g_Z+}=g^-_{z, B, k}(\overline{L}_k^{g_z}, R_{k}^{g_z+}, C_{k-1}^{g_z+}, V^{g_z}_{B,k})$, and thus: 

\begin{align}
    & f_{B_k^{g_z+} \mid \overline{L}_k^{g_z}, R_{k}^{g_z+}, C_{k-1}^{g_z+}, V^{g_z}_{B,k}}(B_k^{g_z+} \mid \overline{L}_k^{g_z}, R_{k}^{g_z+}, C_{k-1}^{g_z+}, V^{g_z}_{B,k}) =  \nonumber \\
    & I(B_k^{g_z+} = g^-_{z, B, k}(\overline{L}_k^{g_z}, R_{k}^{g_z+}, C_{k-1}^{g_z+}, V^{g_z}_{B,k})  \nonumber,
\end{align}

and where $V_k^{g_z}$ is completely exogenous with respect to all other variables in the observed data, and $g^-_{z, k}(\cdot)$ is specified precisely so that:

\begin{align}
    & \int_{v_{B,k}}I(B_k^{g_z+} = g^-_{z,B, k}(\overline{L}_k^{g_z}, R_{k}^{g_z+}, C_{k-1}^{g_z+}, v_{B,k}))f_{V^{g_z}_{B,k}}(v_{B,k}) =  \label{eq; stochastic def}\\
    & f_{B_k^{g_z+} \mid \overline{L}_k^{g_z}, R_{k}^{g_z+}, C_{k-1}^{g_z+}}(B_k^{g_z+} \mid \overline{L}_k^{g_z}, R_{k}^{g_z+}, C_{k-1}^{g_z+}) \nonumber .
\end{align}

That is, under regime $g_z$, the suspected inferior treatment, $B_k^{g_Z+}$ is assigned a value as a deterministic function of $\{\overline{L}_k^{g_z}, R_{k}^{g_z+}, C_{k-1}^{g_z+}, V^{g_z}_{B,k}\}$ according to $g^-_{z, k}(\cdot)$, and $g^-_{z, k}(\cdot)$ is precisely specified such that marginalizaing over the joint distribution of $B_k^{g_Z+}$ and  $V^{g_z}_{B,k}$, conditional on past treatment and covariate history, equals the defining intervention distribution of the stochastic intervention. Likewise, there exist $V^{g_z}_{B,k}$ and $V^{g_z}_{H,k}$ with these properties for all $k=0,\dots,K$.

\subsection{Exchangeability 1}

\begin{align}
    \underline{Y}^{g_z}_t \independent I(B_t^{g_z}=b_t) \mid \overline{L}_t^{g_z}=\overline{l}_t, \overline{C}_{t-1}^{g_z}=\overline{Y}_{t-1}^{g_z}=0,  \overline{H}_{t-1}^{g_z}=\overline{h}_{t-1}, \overline{B}_{t-1}^{g_z}=\overline{b}_{t-1}, \label{Cex1B}
\end{align}

for $\{\overline{b}_{t}, \overline{l}_t, \overline{h}_{t-1}  \mid P(  \overline{B}^{g_z+}=\overline{b}_{t}, 
\overline{L}_{t}^{g_z}=\overline{l}_{t}, 
\overline{C}_{t-1}^{g_z+}=\overline{Y}_{t-1}^{g_z}=0, 
\overline{H}_{t-1}^{g_z+}=\overline{h}_{t-1})>0\}$,  $t\in\{1,\dots,k\}$, and:
\\

\begin{align}
    \underline{Y}^{g_z}_t \independent I(H_t^{g_z}=h_t) \mid \overline{B}_{t}^{g_z}=\overline{b}_{t}, \overline{L}_t^{g_z}=\overline{l}_t, \overline{C}_{t-1}^{g_z}=\overline{Y}_{t-1}^{g_z}=0, \overline{H}_{t-1}^{g_z}=\overline{h}_{t-1}, \label{Cex1H}
\end{align}

for $\{\overline{h}_{t}, \overline{b}_{t}, \overline{l}_{t}  \mid P(  \overline{H}_{t}^{g_z+}=\overline{h}_{t}, 
\overline{B}_{t}^{g_z+}=\overline{b}_{t}, 
\overline{L}_{t}^{g_z}=\overline{l}_{t}, 
\overline{C}_{t-1}^{g_z+}=\overline{Y}_{t-1}^{g_z}=0)>0\}$, $t\in\{1,\dots,k\}$, and:

\begin{align}
    \underline{Y}^{g_z}_t \independent I(C_{t}^{g_z}=0) \mid \overline{H}_{t}^{g_z}=\overline{h}_{t}, \overline{B}_{t}^{g_z}=\overline{b}_{t}, \overline{L}_{t}^{g_z}=\overline{l}_{t}, \overline{Y}_{t-1}^{g_z}=\overline{C}_{t-1}^{g_z}=0,  \label{Cex1C}
\end{align}

for $\{\overline{h}_{t}, \overline{b}_{t}, \overline{l}_{t}  \mid P(  \overline{H}_{t}^{g_z+}=\overline{h}_{t}, 
\overline{B}_{t}^{g_z+}=\overline{b}_{t}, 
\overline{L}_{t}^{g_z}=\overline{l}_{t}, 
\overline{C}_{t}^{g_z+}=\overline{Y}_{t-1}^{g_z}=0)>0\}$, $t\in\{1,\dots,k\}$.

Sequential exchangeability conditions for the outcomes $Y_t$, with respect to past treatment, in the main text are extended to include sequential exchangeability with respect to censoring. 

\subsection{Exchangeability 2}

\begin{align}
    \underline{B}^{g_z}_t \independent I(B_{t-1}^{g_z}=b_{t-1}) \mid \overline{L}_{t-1}^{g_z}=\overline{l}_{t-1}, \overline{C}_{t-2}^{g_z}=\overline{Y}_{t-2}^{g_z}=0,  \overline{H}_{t-2}^{g_z}=\overline{h}_{t-2}, \overline{B}_{t-2}^{g_z}=\overline{b}_{t-2}, \label{Cex2BB}
\end{align}

for $\{\overline{b}_{t-1}, \overline{l}_{t-1}, \overline{h}_{t-2}  \mid P(  \overline{B}_{t-1}^{g_z+}=\overline{b}_{t-1}, 
\overline{L}_{t-1}^{g_z}=\overline{l}_{t-1}, 
\overline{C}_{t-2}^{g_z+}=\overline{Y}_{t-2}^{g_z}=0, 
\overline{H}_{t-2}^{g_z+}=\overline{h}_{t-2})>0\}$,  $t\in\{1,\dots,k\}$, and:
\\

\begin{align}
    \underline{B}^{g_z}_t \independent I(H_{t-1}^{g_z}=h_{t-1}) \mid \overline{B}_{t-1}^{g_z}=\overline{b}_{t-1}, \overline{L}_{t-1}^{g_z}=\overline{l}_{t-1}, \overline{C}_{t-2}^{g_z+}=\overline{Y}_{t-2}^{g_z}=0, \overline{H}_{t-2}^{g_z}=\overline{h}_{t-2}, \label{Cex2BH}
\end{align}

for $\{\overline{h}_{t-1}, \overline{b}_{t-1}, \overline{l}_{t-1}  \mid P(  \overline{H}_{t-1}^{g_z+}=\overline{h}_{t-1}, 
\overline{B}_{t-1}^{g_z+}=\overline{b}_{t-1}, 
\overline{L}_{t-1}^{g_z}=\overline{l}_{t-1}, 
\overline{C}_{t-2}^{g_z+}=\overline{Y}_{t-2}^{g_z}=0)>0\}$,  $t\in\{1,\dots,k\}$, and:
\\

\begin{align}
    \underline{B}^{g_z}_t \independent  I(C_{t-1}^{g_z}=0) \mid  \overline{H}_{t-1}^{g_z}=\overline{h}_{t-1}, \overline{B}_{t-1}^{g_z}=\overline{b}_{t-1}, \overline{L}_{t-1}^{g_z}=\overline{l}_{t-1}, \overline{Y}_{t-2}^{g_z}=\overline{C}_{t-2}^{g_z}=0, \label{Cex2BC}
\end{align}

for $\{\overline{h}_{t-1}, \overline{b}_{t-1}, \overline{l}_{t-1}  \mid P(  \overline{H}_{t-1}^{g_z+}=\overline{h}_{t-1}, 
\overline{B}_{t-1}^{g_z+}=\overline{b}_{t-1}, 
\overline{L}_{t-1}^{g_z}=\overline{l}_{t-1}, 
\overline{C}_{t-1}^{g_z+}=\overline{Y}_{t-2}^{g_z}=0)>0\}$, $t\in\{1,\dots,k\}$, and:
\\

\begin{align}
    \underline{H}^{g_z}_t \independent I(B_{t}^{g_z}=b_{t}) \mid \overline{L}_{t}^{g_z}=\overline{l}_{t}, \overline{C}_{t-1}^{g_z}=\overline{Y}_{t-1}^{g_z}=0, \overline{H}_{t-1}^{g_z}=\overline{h}_{t-1}, \overline{B}_{t-1}^{g_z}=\overline{b}_{t-1}, \label{Cex2HB}
\end{align}

for $\{\overline{b}_{t}, \overline{l}_{t}, \overline{h}_{t-1}  \mid P(  \overline{B}_{t}^{g_z+}=\overline{b}_{t}, 
\overline{L}_{t}^{g_z}=\overline{l}_{t}, 
\overline{C}_{t-1}^{g_z+}=\overline{Y}_{t-1}^{g_z}=0, 
\overline{H}_{t-1}^{g_z+}=\overline{h}_{t-1})>0\}$,  $t\in\{1,\dots,k\}$, and:
\\

\begin{align}
    \underline{H}^{g_z}_t \independent I(H_{t-1}^{g_z}=h_{t-1}) \mid  \overline{B}_{t-1}^{g_z}=\overline{b}_{t-1}, \overline{L}_{t-1}^{g_z}=\overline{l}_{t-1}, \overline{C}_{t-2}^{g_z}=\overline{Y}_{t-2}^{g_z}=0, \overline{H}_{t-2}^{g_z}=\overline{h}_{t-2},\label{Cex2HH}
\end{align}

for $\{\overline{h}_{t-1},  \overline{b}_{t-1},  \overline{l}_{t-1}  \mid P(  \overline{H}_{t-1}^{g_z+}=\overline{h}_{t-1}, 
\overline{B}_{t-1}^{g_z+}=\overline{b}_{t-1}, 
\overline{L}_{t-1}^{g_z}=\overline{l}_{t-1}, 
\overline{C}_{t-2}^{g_z+}=\overline{Y}_{t-2}^{g_z}=0)>0\}$,  $t\in\{1,\dots,k\}$, and:
\\

\begin{align}
    \underline{H}^{g_z}_t \independent  I(C_{t-1}^{g_z}=0) \mid \overline{H}_{t-1}^{g_z}=\overline{h}_{t-1}, \overline{B}_{t-1}^{g_z}=\overline{b}_{t-1}, \overline{L}_{t-1}^{g_z}=\overline{l}_{t-1}, \overline{C}_{t-2}^{g_z}=\overline{Y}_{t-2}^{g_z}=0, \label{Cex2HC}
\end{align}

for $\{\overline{h}_{t-1}, \overline{b}_{t-1}, \overline{l}_{t-1}  \mid P(
\overline{H}_{t-1}^{g_z+}=\overline{h}_{t-1}, 
\overline{B}_{t-1}^{g_z+}=\overline{b}_{t-1}, 
\overline{L}_{t-1}^{g_z}=\overline{l}_{t-1}, 
\overline{C}_{t-1}^{g_z+}=\overline{Y}_{t-2}^{g_z}=0)>0\}$, $t\in\{1,\dots,k\}$.
\\

Sequential exchangeability conditions for the treatment resources, with respect to past treatment, in the main text are extended to include sequential exchangeability with respect to censoring.

\subsection{Consistency}

\begin{align}
    & \text{if } \overline{B}_{t}=\overline{B}_{t}^{g_z+} \text{, }  \overline{H}_{t}=\overline{H}_{t}^{g_z+}  \text{ and } \overline{C}_{t}=0 \nonumber\\
    & \text{then } Y_t = Y_t^{g_z}, L_{t+1} = L_{t+1}^{g_z}, and \text{, and } B_{t+1} = B^{g_z}_{t+1},
    \label{Cass: consistency B and Y}
\end{align} 

and,

\begin{align}
    & \text{if } \overline{H}_{t}=\overline{H}_{t}^{g_z+} \text{, } \overline{C}_{t}=0 \text{, and }  \overline{B}_{t+1}=\overline{B}_{t+1}^{g_z+}   \nonumber\\
    & \text{then } H_{t+1} = H^{g_z}_{t+1},
    \label{Cass: consistency H1}
\end{align}

and,

\begin{align}
    & \text{if } \overline{C}_{t-1}=0 \text{, } \overline{B}_{t}=\overline{B}_{t}^{g_z+} \text{, and }  \overline{H}_{t}=\overline{H}_{t}^{g_z+}      \nonumber\\
    & \text{then } C_{k+1} = C^{g_z}_{t},
    \label{Cass: consistency H}
\end{align}

for all $t\in\{0,\dots,K\}$. 

Consistency statements are updated to reflect the additional hypothetical intervention to prevent censoring.

\subsection{Positivity}

\begin{align}
    & f_{R_t^{g_z+}, \overline{L}_t^{g_z}, C_{t-1}^{g_z+}}(1, \overline{L}_t, 0)>0\text{ and } f_{B_t^{g_z+} \mid R_t^{g_z+}, \overline{L}_t^{g_z}, C_{t-1}^{g_z+}}(B_t \mid 1, \overline{L}_t, 0)>0 \implies \nonumber  \\
    & \quad f_{B_t \mid R_t, \overline{L}_t, C_{t-1}}(B_t \mid 1, \overline{L}_t, 0)>0\text{, w.p.1} \label{eq: CpositivityB},
\end{align}

and: 

\begin{align}
    & f_{S_t^{g_z+}, \overline{L}_t^{g_z}, C_{t-1}^{g_z+}}(1, \overline{L}_t, 0)>0 \text{ and } f_{H_t^{g_z+} \mid S_t^{g_z+}, \overline{L}_t^{g_z}, C_{t-1}^{g_z+}}(H_t \mid 1, \overline{L}_t, 0)>0 \implies \nonumber  \\
    &   \quad f_{H_t \mid S_t, \overline{L}_t, C_{t-1}}(H_t \mid 1, \overline{L}_t, 0)>0\text{, w.p.1},  \label{eq: CpositivityH} 
\end{align}

and:

\begin{align}
    & f_{S_t^{g_z+}, \overline{L}_t^{g_z}, C_{t-1}^{g_z+}}(1, \overline{L}_t, 0)>0  \implies \nonumber  \\
    &   \quad f_{C_t \mid S_t, \overline{L}_t, C_{t-1}}(0 \mid 1, \overline{L}_t, 0)>0\text{, w.p.1},  \label{eq: CpositivityC} 
\end{align}

Positivity conditions are extended so that if there exists in interval $t$ some censoring-eligible individuals (alive, untreated, and uncensored) with covariate history $\overline{l}_t$ who are uncensored under regime ${g_z+}$, then there must be some such individuals in the unintervened world.

\begin{theorem}
If conditions \eqref{Cex1B}-\eqref{eq: CpositivityC} hold, then $\mathbb{E}( Y^{g_z}_K )$ is identified from the non-extended g-formula of Robins (1986) for $Y_K$, $f^{g_z}_{Y_K}(1)$: equal to

\begin{align}
  & \sum_{\overline{l}_K} \sum_{\overline{h}_K} \sum_{\overline{b}_K} \sum_{k=1}^K P(Y_k=1 \mid C_k=0, \overline{H}_k=\overline{h}_k, \overline{B}_k=\overline{b}_k, \overline{L}_k=\overline{l}_k, Y_{k-1}=0) \label{eq; CgformY} \\
    & \times \prod_{j=1}^{k} \Big\{f_{H_j^{g_z+} \mid \overline{B}_{j}^{g_z+}, \overline{L}_j^{g_z}, Y_{j-1}^{g_z}, C_{j-1}^{g_z+}, \overline{H}_{j-1}^{g_z+}}(h_j \mid \overline{b}_{j}, \overline{l}_j, 0, 0, \overline{h}_{j-1})\nonumber \\
    & \times f_{B_j^{g_z+} \mid \overline{L}_j^{g_z}, Y_{j-1}^{g_z}, C_{j-1}^{g_z+}, \overline{H}_{j-1}^{g_z+}, \overline{B}_{j-1}^{g_z+}}(b_j \mid \overline{l}_j, 0, 0, \overline{h}_{j-1}, \overline{b}_{j-1})\nonumber \\
    & \times P(L_j=l_j \mid Y_{j-1}=0, C_{j-1}=0, \overline{H}_{j-1}=\overline{h}_{j-1}, \overline{B}_{j-1}=\overline{b}_{j-1}, \overline{L}_{j-1}=\overline{l}_{j-1}) \nonumber\\
    & \times P(Y_{j-1}=0 \mid C_{j-1}=0, \overline{H}_{j-1}=\overline{h}_{j-1}, \overline{B}_{j-1}=\overline{b}_{j-1}, \overline{L}_{j-1}=\overline{l}_{j-1}, Y_{j-2}=0)\Big\} \nonumber,
\end{align}

where  

\begin{align}
    & f_{B_j^{g_z+} \mid \overline{L}_j^{g_z}, Y_{j-1}^{g_z}, C_{j-1}^{g_z+},  \overline{H}_{j-1}^{g_z+}, \overline{B}_{j-1}^{g_z+}}(b_j \mid \overline{l}_j, 0, 0, \overline{h}_{j-1}, \overline{b}_{j-1}) = \\
    & \Big(\alpha_j(z) \times f_{B_j \mid \overline{L}_j, Y_{j-1}, C_{j-1}, \overline{H}_{j-1}, \overline{B}_{j-1}}(1 \mid \overline{l}_j, 0, 0, \overline{h}_{j-1}, \overline{b}_{j-1})\Big)^{b_j} \nonumber \\
    \times &  \Big(1-\alpha_j(z) \times f_{B_j \mid \overline{L}_j, Y_{j-1}, C_{j-1}, \overline{H}_{j-1}, \overline{B}_{j-1}}(0 \mid \overline{l}_j, 0, 0, \overline{h}_{j-1}, \overline{b}_{j-1})\Big)^{1-b_j} \nonumber.
\end{align}, and

\begin{align}
    & f_{H_j^{g_z+} \mid \overline{B}_{j}^{g_z+}, \overline{L}_j^{g_z}, Y_{j-1}^{g_z}, C_{j-1}^{g_z+}, \overline{H}_{j-1}^{g_z+}}(h_j \mid \overline{b}_{j}, \overline{l}_j, 0, 0, \overline{h}_{j-1}) = \\
    & \Big(\beta_j(z) \times f_{H_j \mid \overline{B}_{j}, \overline{L}_j, Y_{j-1}, C_{j-1}, \overline{H}_{j-1}}(1 \mid \overline{b}_{j}, \overline{l}_j, 0, 0, \overline{h}_{j-1})\Big)^{h_j} \nonumber \\
    \times &  \Big(1-\beta_j(z) \times f_{H_j \mid \overline{B}_{j}, \overline{L}_j, Y_{j-1}, C_{j-1}, \overline{H}_{j-1}}(0 \mid \overline{b}_{j}, \overline{l}_j, 0, 0, \overline{h}_{j-1})\Big)^{1-h_j} \nonumber 
\end{align}, 

and $\alpha_1(z, 1)$ is identified by noting the equality $P(B_k^{g_z}=1) = P(B_k=1)$, $\alpha_k(z, R_k)$ and $\alpha_k(z, S_k)$, are identified by the functionals in expressions \eqref{eq; alphanew} and  \eqref{eq; betanew}, and $\aleph$ and $\beth$ indicator functions are identified by the functionals in expressions \eqref{eq; alephB} - \eqref{eq; bethH}, except replacing $P(B_k^{g_z}=1)$ with $f^{g_z}_{B_k}(1)$, $P(H_k^{g_z}=1)$ with $f^{g_z}_{H_k}(1)$, and $P(R_k^{g_z+}=1)$ with $f^{g_z}_{R_k}(1)$ and $P(S_k^{g_z+}=1)$ with $f^{g_z}_{S_k}(1)$, where: 

\begin{align}
    f^{g_z}_{B_j}(1) = 
    & \sum_{\overline{l}_j} P(B_j=1 \mid \overline{L}_j=\overline{l}_j, R_j=1, C_{j-1}=0) \label{eq; CgformB} \\
    & \times \prod_{m=1}^{j} \Big\{  P(L_m=l_m \mid R_m=1, C_{m-1}=0, \overline{L}_{m-1}=\overline{l}_{m-1}) \nonumber\\
    & \times P(Y_{m-1}=0 \mid C_{m-1}=0, H_{m-1}=0, S_{m-1}=1, \overline{L}_{m-1}=\overline{l}_{m-1}) \nonumber \\
    & \times  f_{H_{m-1}^{g_z+} \mid \overline{B}_{m-1}^{g_z+}, \overline{L}_{m-1}^{g_z}, Y_{m-2}^{g_z}, C_{m-2}^{g_z+}, \overline{H}_{m-2}^{g_z+}}(0 \mid \overline{b}_{m-1}, \overline{l}_{m-1}, 0, 0, \overline{h}_{m-2})  \nonumber \\
    & \times  f_{B_{m-1}^{g_z+} \mid \overline{L}_{m-1}^{g_z}, Y_{m-2}^{g_z}, C_{m-2}^{g_z+},  \overline{H}_{m-2}^{g_z+}, \overline{B}_{m-2}^{g_z+}}(0 \mid \overline{l}_{m-1}, 0, 0, \overline{h}_{m-2}, \overline{b}_{m-2}) \Big\} \nonumber,  
\end{align}

and:

\begin{align}
    f^{g_z}_{H_j}(1)=
    & \sum_{\overline{l}_j} P(H_j=1 \mid \overline{L}_j=\overline{l}_j, S_j=1, C_{j-1}=0) \label{eq; CgformH} \\
    & \times \prod_{m=1}^{j} \Big\{ f_{B_m^{g_z+} \mid \overline{L}_m^{g_z}, Y_{m-1}^{g_z}, C_{m-1}^{g_z+},  \overline{H}_{j-m}^{g_z+}, \overline{B}_{m-1}^{g_z+}}(0 \mid \overline{l}_j, 0, 0, \overline{h}_{m-1}, \overline{b}_{m-1})  \nonumber \\
    & \times  P(L_m=l_m \mid R_m=1, C_{m-1}=0, \overline{L}_{m-1}=\overline{l}_{m-1}) \nonumber\\
    & \times P(Y_{m-1}=0 \mid C_{m-1}=0, H_{m-1}=0, S_{m-1}=1, \overline{L}_{m-1}=\overline{l}_{m-1}) \nonumber \\
    & \times  f_{H_{m-1}^{g_z+} \mid \overline{B}_{m-1}^{g_z+}, \overline{L}_{m-1}^{g_z}, Y_{m-2}^{g_z}, C_{m-2}^{g_z+}, \overline{H}_{m-2}^{g_z+}}(0 \mid \overline{b}_{m-1}, \overline{l}_{m-1}, 0, 0, \overline{h}_{m-2})  \Big\}\nonumber,
\end{align}

and:

\begin{align}
    f^{g_z}_{R_k}(1) = 
    & \sum_{\overline{l}_{k-1}} P(Y_{k-1}=0 \mid C_{k-1}=0,  \overline{H}_{k-1}=0, \overline{B}_{k-1}=0, \overline{L}_{k-1}=\overline{l}_{k-1}, Y_{k-2}=0) \label{eq; CgformR} \\
    & \times  \prod_{m=1}^{k-1} \Big\{f_{H_m^{g_z+} \mid \overline{B}_{m}^{g_z+}, \overline{L}_m^{g_z}, Y_{m-1}^{g_z}, C_{m-1}^{g_z+}, \overline{H}_{m-1}^{g_z+}}(0 \mid \overline{b}_{m}, \overline{l}_m, 0, 0, \overline{h}_{m-1})  \nonumber \\
    & \times  f_{B_m^{g_z+} \mid \overline{L}_m^{g_z}, Y_{m-1}^{g_z}, C_{m-1}^{g_z+},  \overline{H}_{j-m}^{g_z+}, \overline{B}_{m-1}^{g_z+}}(0 \mid \overline{l}_j, 0, 0, \overline{h}_{m-1}, \overline{b}_{m-1})  \nonumber \\
    & \times P(L_m=l_m \mid R_m=1, C_{m-1}=0, \overline{L}_{m-1}=\overline{l}_{m-1}) \nonumber\\
    & \times P(Y_{m-1}=0 \mid C_{m-1}=0, H_{m-1}=0, S_{m-1}=1, \overline{L}_{m-1}=\overline{l}_{m-1}) \Big\}\nonumber,
\end{align}

and:

\begin{align}
    f^{g_z}_{S_k}(1) = 
    & \sum_{\overline{l}_{k}} f_{B_k^{g_z+} \mid \overline{L}_k^{g_z}, Y_{k-1}^{g_z}, C_{k-1}^{g_z+},  \overline{H}_{k-1}^{g_z+}, \overline{B}_{k-1}^{g_z+}}(0 \mid \overline{l}_k, 0, 0, \overline{h}_{k-1}, \overline{b}_{k-1}) \label{eq; CgformS} \\
    & \times P(L_k=l_k \mid R_k=1, C_{k-1}=0, \overline{L}_{k-1}=\overline{l}_{k-1}) \\
    & \times \prod_{m=1}^{k-1} \Big\{ P(Y_{m}=0 \mid C_m=0, H_{m}=0, S_{m}=1, \overline{L}_{m}=\overline{l}_{m}) \nonumber \\
    & \times  f_{H_m^{g_z+} \mid \overline{B}_{m}^{g_z+}, \overline{L}_m^{g_z}, Y_{m-1}^{g_z}, C_{m-1}^{g_z+}, \overline{H}_{m-1}^{g_z+}}(0 \mid \overline{b}_{m}, \overline{l}_m, 0, 0, \overline{h}_{m-1}) \nonumber \\
    & \times  f_{B_m^{g_z+} \mid \overline{L}_m^{g_z}, Y_{m-1}^{g_z}, C_{m-1}^{g_z+},  \overline{H}_{j-m}^{g_z+}, \overline{B}_{m-1}^{g_z+}}(0 \mid \overline{l}_j, 0, 0, \overline{h}_{m-1}, \overline{b}_{m-1}) \nonumber \\
    & \times   P(L_m=l_m \mid R_m=1, C_{m-1}=0, \overline{L}_{m-1}=\overline{l}_{m-1}) \Big\} \nonumber.
\end{align}

\label{theorem: identification g formula with censoring}
\end{theorem}

\begin{proof}
Assume that conditions \eqref{Cex1B}-\eqref{eq: CpositivityC} hold. Using laws of probability:

\begin{align}
\mathbb{E}( Y^{g_z}_K ) = 
    & \sum_{k=1}^K P(Y_k^{g_z}=1 \mid  Y_{k-1}^{g_z}=0) \prod_{j=1}^{k} \Big\{P(Y_{j-1}^{g_z}=0 \mid Y_{j-2}^{g_z}=0)\Big\} \nonumber \\
     = &  \sum_{k=1}^K P(Y_k^{g_z}=1, \overline{Y}_{k-1}^{g_z}=0)  \nonumber
\end{align}

Then $\mathbb{E}( Y^{g_z}_K )$ is identified if each element of the sum is identified. We provide the proof for $k=1$, and leave the rest for the reader. 

Since $\overline{Y}_{0}^{g_z}=0$ for all individuals,

\begin{align}
& P(Y_1^{g_z}=1, \overline{Y}_{0}^{g_z}=0) =  P(Y_1^{g_z}=1). \nonumber
\end{align}

Using laws of probability,

{\tiny 
\begin{equation*}
        \begin{split}
    & = \sum_{l_1} \sum_{h_1} \sum_{b_1}P(Y_1^{g_z}=1 \mid C_1^{g_z+}=0, H^{g_z+}_1=h_1, B_1^{g_z+}=b_1, L_1^{g_z}=l_1) \nonumber\\
    & \times  f_{H_1^{g_z+} \mid B_{1}^{g_z+}, L_1^{g_z}}(h_1 \mid b_{1}, l_1)\nonumber \\
    & \times f_{B_1^{g_z+} \mid L_1^{g_z}}(b_1 \mid l_1)\nonumber \\
    & \times P(L_1^{g_z}=l_1).\nonumber
        \end{split}
    \end{equation*}
}%

Using the complete exogeneity of $\overline{V}^{g_z}_{B,1}$, and $\overline{V}^{g_z}_{H,1}$, we find that

{\tiny 
\begin{equation*}
        \begin{split}
    & = \int_{v_{H,1}}\int_{v_{B,1}} \sum_{l_1} \sum_{h_1} \sum_{b_1} P(Y_1^{g_z}=1 \mid H^{g_z+}_1=h_1, B_1^{g_z+}=b_1, L_1^{g_z}=l_1, V^{g_z}_{H,1}=v_{H,1}, V^{g_z}_{B,1}=v_{B,1}) \nonumber\\
    & \times  f_{H_1^{g_z+} \mid B_{1}^{g_z+}, L_1^{g_z}, V^{g_z}_{H,1}}(h_1 \mid b_{1}, l_1, v_{H,1})f(v_{H,1})\nonumber \\
    & \times f_{B_1^{g_z+} \mid L_1^{g_z}, V^{g_z}_{B,1}}(b_1 \mid l_1, v_{B,1})f(v_{B,1})\nonumber \\
    & \times P(L_1^{g_z}=l_1) \nonumber
        \end{split}
    \end{equation*}
}%

Noting that $B_1^{g_z+}$ and $H_1^{g_z+}$ are constants conditional on treatment, and covariate histories and on $V^{g_z}_{B,1}$ and $V^{g_z}_{H,1}$, respectively,

{\tiny 
\begin{equation*}
        \begin{split}
    & = \int_{v_{H,1}}\int_{v_{B,1}}\sum_{l_1} \sum_{h_1} \sum_{b_1} P(Y_1^{g_z}=1 \mid L_1^{g_z}=l_1, V^{g_z}_{H,1}=v_{H,1}, V^{g_z}_{B,1}=v_{B,1}) \nonumber\\
    & \times  f_{H_1^{g_z+} \mid B_{1}^{g_z+}, L_1^{g_z}, V^{g_z}_{H,1}}(h_1 \mid b_{1}, l_1, v_{H,1})f(v_{H,1})\nonumber \\
    & \times f_{B_1^{g_z+} \mid L_1^{g_z}, V^{g_z}_{B,1}}(b_1 \mid l_1, v_{B,1})f(v_{B,1})\nonumber \\
    & \times P(L_1^{g_z}=l_1) \nonumber
        \end{split}
    \end{equation*}
}%

Using the complete exogeneity of $V^{g_z}_{B,1}$, and $V^{g_z}_{H,1}$ again, 

{\tiny 
\begin{equation*}
        \begin{split}
    & = \int_{v_{B,1}}  \int_{v_{H,1}}\sum_{l_1} \sum_{h_1} \sum_{b_1} P(Y_1^{g_z}=1 \mid L_1^{g_z}=l_1) \nonumber\\
    & \times  f_{H_1^{g_z+} \mid B_{1}^{g_z+}, L_1^{g_z}, V^{g_z}_{H,1}}(h_1 \mid b_{1}, l_1, v_{H,1})f(v_{H,1})\nonumber \\
    & \times f_{B_1^{g_z+} \mid L_1^{g_z}, V^{g_z}_{B,1}}(b_1 \mid l_1, v_{B,1})f(v_{B,1})\nonumber \\
    & \times P(L_1^{g_z}=l_1)
        \end{split}
    \end{equation*}
}%

Defining $\mathcal{B}^{g_z}_{1,pos} \times \mathcal{H}^{g_z}_{1,pos} \times \mathcal{L}^{g_z}_{1,pos} = \{h_{1}, b_{1}, l_1  \mid P(H_1^{g_z+}=h_1, B_{1}^{g_z+}=b_{1}, L_{1}^{g_z}=l_1)>0\}$, that is, the support of $B_1^{g_z+}$, $H_1^{g_z+}$, and $L_1^{g_z}$ under regime $g_z$, then

{\tiny 
\begin{equation*}
        \begin{split}
    & = \int_{v_{B,1}}  \int_{v_{H,1}}\sum_{\mathcal{B}^{g_z}_{1,pos} \times \mathcal{H}^{g_z}_{1,pos} \times \mathcal{L}^{g_z}_{1,pos}} P(Y_1^{g_z}=1 \mid L_1^{g_z}=l_1) \nonumber\\
    & \times  f_{H_1^{g_z+} \mid B_{1}^{g_z+}, L_1^{g_z}, V^{g_z}_{H,1}}(h_1 \mid b_{1}, l_1, v_{H,1})f(v_{H,1})\nonumber \\
    & \times f_{B_1^{g_z+} \mid L_1^{g_z}, V^{g_z}_{B,1}}(b_1 \mid l_1, v_{B,1})f(v_{B,1})\nonumber \\
    & \times P(L_1^{g_z}=l_1) \nonumber
        \end{split}
    \end{equation*}
}%

Sequentially using the weaker exchangeability conditions with respect to $Y_k^{g_z}$ for $B_k$, $H_k$, and $C_k$, of expressions \eqref{Cex1B} - \eqref{Cex1C}, respectively, and noting that all individuals have $\overline{C}_K^{g_z+}=0$,

{\tiny 
\begin{equation*}
        \begin{split}
    & =  \int_{v_{B,1}}  \int_{v_{H,1}}\sum_{\mathcal{B}^{g_z}_{1,pos} \times \mathcal{H}^{g_z}_{1,pos} \times \mathcal{L}^{g_z}_{1,pos}} P(Y_1^{g_z}=1 \mid C^{g_z}_1=0, H_1^{g_z} = g^-_{z, H, 1}(l_1, v_{B,1}, v_{H,1}), B_1^{g_z} = g^-_{z, B, 1}(l_1, v_{B,1}), L_1^{g_z}=l_1) \nonumber\\
    & \times f_{H_1^{g_z+} \mid B_{1}^{g_z+}, L_1^{g_z}, V^{g_z}_{H,1}, C^{g_z+}_0}(h_1 \mid b_{1}, l_1, v_{H,1}, 0)f(v_{H,1})\nonumber \\
    & \times f_{B_1^{g_z+} \mid L_1^{g_z}, V^{g_z}_{B,1}, C^{g_z+}_0}(b_1 \mid l_1, v_{B,1},  0)f(v_{B,1})\nonumber \\
    & \times P(L_1^{g_z}=l_1) \nonumber
        \end{split}
    \end{equation*}
}%

Sequentially using the consistency conditions for $L_k$, $B_k$, $H_k$, $C_k$, and $Y_k$ of expressions \eqref{Cass: consistency B and Y}  and \eqref{Cass: consistency H}, respectively,

{\tiny 
\begin{equation*}
        \begin{split}
    & =  \int_{v_{B,1}}  \int_{v_{H,1}}\sum_{l_1} \sum_{h_1} \sum_{b_1} P(Y_1=1 \mid C_1=0, H_1 = g^-_{z, H, 1}(l_1, v_{B,1}, v_{H,1}), B_1 = g^-_{z, B, 1}(l_1, v_{B,1}), L_1=l_1) \nonumber\\
   & \times f_{H_1^{g_z+} \mid B_{1}^{g_z+}, L_1^{g_z}, V^{g_z}_{H,1}, C^{g_z+}_0}(h_1 \mid b_{1}, l_1, v_{H,1}, 0)f(v_{H,1})\nonumber \\
    & \times f_{B_1^{g_z+} \mid L_1^{g_z}, V^{g_z}_{B,1}, C^{g_z+}_0}(b_1 \mid l_1, v_{B,1},  0)f(v_{B,1})\nonumber \\
    & \times P(L_1=l_1) 
        \end{split}
    \end{equation*}
}%

Noting that $f_{H_1^{g_z+} \mid B_{1}^{g_z+}, L_1^{g_z}, V^{g_z}_{H,1}, C^{g_z+}_0}(h_1 \mid v_{H,1}, b_{1}, l_1, 0) = I(h_1 = g^-_{z, H, 1}(l_1, v_{B,1}, v_{H,1}))$, and $f_{B_1^{g_z+} \mid L_1^{g_z}, V^{g_z}_{B,1}, C^{g_z+}_0}(b_1 \mid v_{B,1}, l_1, 0) = I(b_1 = g^-_{z, B, 1}(l_1, v_{B,1}))$, so: 

{\tiny 
\begin{equation*}
        \begin{split}
    & =   \sum_{l_1} \sum_{h_1} \sum_{b_1} P(Y_1=1 \mid C_1=0, H_1 = h_1, B_1 = b_1, L_1=l_1) \nonumber\\
    & \times \int_{v_{H,1}} f_{H_1^{g_z+} \mid B_{1}^{g_z+}, L_1^{g_z}, V^{g_z}_{H,1}, C^{g_z+}_0}(h_1 \mid b_{1}, l_1, v_{H,1}, 0)f(v_{H,1})\nonumber \\
    & \times \int_{v_{B,1}} f_{B_1^{g_z+} \mid L_1^{g_z}, V^{g_z}_{B,1}, C^{g_z+}_0}(b_1 \mid l_1, v_{B,1},  0)f(v_{B,1})\nonumber \\
    & \times P(L_1=l_1) \nonumber
        \end{split}
    \end{equation*}
}%

And using the definition of stochastic regimes in expression \eqref{eq; stochastic def},

{\tiny 
\begin{equation*}
        \begin{split}
    & =  \int_{v_{B,1}}  \int_{v_{H,1}}\sum_{l_1} \sum_{h_1} \sum_{b_1} P(Y_1=1 \mid C_1=0, H_1 = h_1, B_1 = b_1, L_1=l_1) \nonumber\\
    & \times f_{H_1^{g_z+} \mid B_{1}^{g_z+}, L_1^{g_z}, C^{g_z+}_0}(h_1 \mid b_{1}, l_1, 0)\nonumber \\
    & \times f_{B_1^{g_z+} \mid L_1^{g_z}, C^{g_z+}_0}(b_1 \mid l_1, v_{B,1},  0)\nonumber \\
    & \times P(L_1=l_1) \nonumber
        \end{split}
    \end{equation*}
}%

Now, all that is left is identifying the intervention distributions $f_{B_1^{g_z+} \mid L_1^{g_z}, C^{g_z+}_0}(\cdot \mid \cdot)$ and $f_{H_1^{g_z+} \mid B_{1}^{g_z+}, L_1^{g_z}, C^{g_z+}_0}(\cdot \mid \cdot)$, which are defined in terms of factual distributions and user-specified constraints, with the exception of marginal distributions of the natural values of treatment and of treatment eligibility under regime $g_z$. Both of these terms are identified using steps analogous to the above, particularly employing weaker exchangeability conditions with respect to $B_k^{g_z}$ for $B_{k-1}$, $H_{k-1}$, and $C_{k-1}$, and with respect to $H_k^{g_z}$ for $B_k$, $H_{k-1}$, and $C_{k-1}$ of expressions \eqref{Cex2BB} - \eqref{Cex2HC}, respectively, as needed. 

Then, we repeat for intervals $k=2,\dots, K$ and combine terms to yield $\mathbb{E}( Y^{g_z}_K ) = f^{g_z}_{Y_K}(1)$, with $P(B_j^{g_z}=1) = f^{g_z}_{B_j}(1)$, $P(H_j^{g_z}=1) = f^{g_z}_{H_j}(1)$, $P(R_j^{g_z}=1) = f^{g_z}_{R_j}(1)$, and $P(S_j^{g_z}=1) = f^{g_z}_{S_j}(1)$. 

\end{proof}

\clearpage

\section{Appendix E: Equivalence of alternative g-formula representation} 

In this section we prove that the g-formula representations of expressions \eqref{eq; gformY} and \eqref{eq; altgformY} are equivalent, under the positivity conditions in expressions \eqref{eq: positivityB} and \eqref{eq: positivityH}. We will rely on the following Lemma.

\begin{lemma}

\begin{align*}
& \mathbb{E}[(1-Y_{k-1})W^{g_z}_{H, k}W^{g_z}_{B, k} \mid V] =  \\
& \mathbb{E}[(1-Y_{k-2})W^{g_z}_{H, k-1}W^{g_z}_{B, k-1} \mid V] -\mathbb{E}[Y_{k-1}(1-Y_{k-2})W^{g_z}_{H, k-1}W^{g_z}_{B, k-1} \mid V]
\end{align*}

 where $W^{g_z}_{B, k-1}$ and  $W^{g_z}_{H, k-1}$ are defined as in expressions \eqref{eq; weight B} and \eqref{eq; weight H}, respectively.
\label{lemma; gformequivlemma}
\end{lemma}

\begin{proof}

First, given the positivity conditions \eqref{eq: positivityB} and \eqref{eq: positivityH}, we have that

\begin{align*}
& \mathbb{E}[(1-Y_{k-1})W^{g_z}_{H, k}W^{g_z}_{B, k} \mid V] \\
= & \mathbb{E}
\begin{bmatrix*}[l]
   & (1-Y_{k-1}) \\
   & \times W^{g_z}_{H, k-1} \mathbb{E}\begin{bmatrix*}[l]\frac{f_{H_k^{g_z+} \mid \overline{B}_{k}^{g_z+}, \overline{L}_k^{g_z}, Y_{k-1}^{g_z}, \overline{H}_{k-1}^{g_z+}}(H_k \mid \overline{B}_{k}, \overline{L}_k, 0, \overline{H}_{k-1})}{f_{H_k \mid \overline{B}_{k}, \overline{L}_k, Y_{k-1}, \overline{H}_{k-1}}(H_k \mid \overline{B}_{k}, \overline{L}_k, 0, \overline{H}_{k-1})} \mid \overline{B}_{k}, \overline{L}_k, \overline{H}_{k-1}\end{bmatrix*} \\
   & \times W^{g_z}_{B, k-1} \mathbb{E}\begin{bmatrix*}[l]\frac{f_{B_k^{g_z+} \mid \overline{L}_k^{g_z}, Y_{k-1}^{g_z}, \overline{H}_{k-1}^{g_z+}, \overline{B}_{k-1}^{g_z+}}(B_k \mid \overline{L}_k, 0, \overline{H}_{k-1}, \overline{B}_{k-1})}{f_{B_k \mid \overline{L}_k, Y_{k-1}, \overline{H}_{k-1}, \overline{B}_{k-1}}(B_k \mid \overline{L}_k, 0, \overline{H}_{k-1}, \overline{B}_{k-1})} \mid \overline{L}_k, \overline{H}_{k-1}, \overline{B}_{k-1}\end{bmatrix*} \\
   & \mid V, 
    \end{bmatrix*} \\
= & \mathbb{E}[(1-Y_{k-1})W^{g_z}_{H, k-1}W^{g_z}_{B, k-1} \mid V]
\end{align*}

Since event $Y_{k-1}=0$ implies joint event $(Y_{k-1}=0, Y_{k-2}=0)$, then:

\begin{align*}
  & \mathbb{E}[(1-Y_{k-1})W^{g_z}_{H, k-1}W^{g_z}_{B, k-1} \mid V] \\
= & \mathbb{E}[(1-Y_{k-1})(1-Y_{k-2})W^{g_z}_{H, k-1}W^{g_z}_{B, k-1} \mid V] \\
= & \mathbb{E}[(1-Y_{k-2})W^{g_z}_{H, k-1}W^{g_z}_{B, k-1} -  Y_{k-1}(1-Y_{k-2})W^{g_z}_{H, k-1}W^{g_z}_{B, k-1}\mid V] \\
= & \mathbb{E}[(1-Y_{k-2})W^{g_z}_{H, k-1}W^{g_z}_{B, k-1}\mid V]  -         
    \mathbb{E}[Y_{k-1}(1-Y_{k-2})W^{g_z}_{H, k-1}W^{g_z}_{B, k-1}\mid V] 
\end{align*}

\end{proof}

\begin{theorem}

Define $\lambda_{Y,k}^{g_z}(V)$ as in expression \eqref{eq; lambda}. It follows that the g-formula of expression \eqref{eq; gformY}

\begin{align*}
    & \sum_{\overline{l}_K} \sum_{\overline{h}_K} \sum_{\overline{b}_K} \sum_{k=1}^K P(Y_k=1 \mid \overline{H}_k=\overline{h}_k, \overline{B}_k=\overline{b}_k, \overline{L}_k=\overline{l}_k, Y_{k-1}=0) \\
    & \times \prod_{j=1}^{k} \Big\{f_{H_j^{g_z+} \mid \overline{B}_{j}^{g_z+}, \overline{L}_j^{g_z}, Y_{j-1}^{g_z}, \overline{H}_{j-1}^{g_z+}}(h_j \mid \overline{b}_{j}, \overline{l}_j, 0, \overline{h}_{j-1})\nonumber \\
    & \times f_{B_j^{g_z+} \mid \overline{L}_j^{g_z}, Y_{j-1}^{g_z}, \overline{H}_{j-1}^{g_z+}, \overline{B}_{j-1}^{g_z+}}(b_j \mid \overline{l}_j, 0, \overline{h}_{j-1}, \overline{b}_{j-1})\nonumber \\
    & \times P(L_j=l_j \mid Y_{j-1}=0, \overline{H}_{j-1}=\overline{h}_{j-1}, \overline{B}_{j-1}=\overline{b}_{j-1}, \overline{L}_{j-1}=\overline{l}_{j-1}) \nonumber\\
    & \times P(Y_{j-1}=0 \mid \overline{H}_{j-1}=\overline{h}_{j-1}, \overline{B}_{j-1}=\overline{b}_{j-1}, \overline{L}_{j-1}=\overline{l}_{j-1}, Y_{j-2}=0)\Big\} \nonumber,
\end{align*}

is equivalent to the g-formula of expression \eqref{eq; altgformY}

\begin{align*}
   & \sum_{v}\sum_{k=1}^K \lambda_{Y,k}^{g_z}(v) \prod_{j=1}^{k-1}[1-\lambda_{Y,j}^{g_z}(v)]f(v), 
\end{align*}

\label{theorem; gformequivtheorem}
\end{theorem}

\begin{proof}
By definition of $\lambda_{Y,k}^{g_z}(V)$, we can re-write the alternate g-formula expression of \eqref{eq; altgformY} as

\begin{align*}
   \sum_{v}\sum_{k=1}^K & \mathbb{E}[Y_k(1-Y_{k-1})W^{g_z}_{H, k}W^{g_z}_{B, k} \mid V=v] \\
   & \times \prod_{j=1}^{k}
        \frac{\mathbb{E}[(1-Y_{j-2})W^{g_z}_{H, j-1}W^{g_z}_{B, j-1}\mid V=v]  -
              \mathbb{E}[Y_{j-1}(1-Y_{j-2})W^{g_z}_{H, j-1}W^{g_z}_{B, j-1}\mid V=v]}
             {\mathbb{E}[(1-Y_{j-1})W^{g_z}_{H, j}W^{g_z}_{B, j}\mid V=v]}f(v),
\end{align*}

which by Lemma \ref{lemma; gformequivlemma}

\begin{align*}
   =\sum_{v}\sum_{k=1}^K & \mathbb{E}[Y_k(1-Y_{k-1})W^{g_z}_{H, k}W^{g_z}_{B, k} \mid V=v] f(v).
\end{align*}

By laws of probability, and the positivity conditions \eqref{eq: positivityB} and \eqref{eq: positivityH}, the last expression is equivalent to the g-formula  of \eqref{eq; gformY}.
\end{proof}

\clearpage

\section{Appendix F: Extension to censoring} 

In Appendix D, we provide identification results in the general setting in which censoring is present and in which our target estimands are expected potential outcomes under proportionally-representative interventions for limited resources, under the 'natural course'; that is, under an additional hypothetical intervention to prevent censoring in all individuals (i.e. $\overline{c}_K=0$). As in Appendix D consider $C_k$ to be an indicator for censoring in interval $k$, and a topological order within each interval of $\Big(L_k, B_k, H_k, C_k, Y_k\Big)$. Also, as in Appendix D redefine all regimes $g_z$, $z \in \mathcal{Z}$, to involve some proportionally-representative intervention, \textit{and} an intervention on $\overline{C}^{g_z}_K$ such that $\overline{C}^{g_z+}_K=0$ for all individuals.

In this setting, the constraints are specified as in expressions \eqref{eq; czB} and \eqref{eq; czH}, and these constraints motivate stochastic intervention distributions  as in \eqref{eq; newintB} and \eqref{eq; newintH}, except conditional on being uncensored:

\begin{align}
    f_{B_k^{g_z+} \mid R_k^{g_z+}, C_k^{g_z+}, \overline{L}_k^{g_z}}(1 \mid R_k, 0, \overline{L}_k) 
        =      & \big(\alpha_{k}(z, R_k) \times f_{B_k \mid R_k, C_k, \overline{L}_k}(1 \mid R_k, 0, \overline{L}_k)\big)
                    ^{\beth^{g_z}_{B,k}} \label{eq; CnewintB} \\
        \times & \big(1- \alpha_{k}(z, R_k) \times f_{B_k \mid R_k, C_k, \overline{L}_k}(0 \mid R_k, 0, \overline{L}_k)\big)
                    ^{1-\beth^{g_z}_{B,k}} \nonumber,
\end{align}

and

\begin{align}
    f_{H_k^{g_z+} \mid S_k^{g_z+}, C_k^{g_z+}, \overline{L}_k^{g_z}}(1 \mid S_k, 0, \overline{L}_k) 
        =      & \big(\beta_{k}(z, S_k) \times f_{H_k \mid S_k, \overline{L}_k}(1 \mid S_k, 0, \overline{L}_k)\big)
                    ^{\beth^{g_z}_{H,k}} \label{eq; CnewintH}\\
        \times & \big(1- \beta_{k}(z, S_k) \times f_{H_k \mid S_k, C_k, \overline{L}_k}(0 \mid S_k, 0, \overline{L}_k)\big)
                    ^{1-\beth^{g_z}_{H,k}} \nonumber,
\end{align}

with $\alpha_k(z, R_k)$ and $\beta_k(z, S_k)$ defined identically.

\subsection{Identification}

Identification conditions and subsequent g-formulae identification results are provided in Appendix D.

\subsubsection{Alternative g-formulae representation}

As in the main text, the g-formula can be represented as in expression \eqref{eq; altgformY}, except $\lambda_{Y,k}^{g_z}(V)$ is redefined to be

\begin{align}
    & \lambda_{Y,k}^{g_z}(V) = \frac{
        \mathbb{E}\big[Y_k(1-Y_{k-1})W_{H,k}^{g_z}W_{B,k}^{g_z}W_{C,k}^{g_z} \mid V\big]
                    }{
        \mathbb{E}\big[(1-Y_{k-1})W_{H,k}^{g_z}W_{B,k}^{g_z}W_{C,k}^{g_z}\mid V \big]           
                    }, \label{eq; lambdaC}
\end{align}

where $W^{g_z}_{B,k}$ and $W^{g_z}_{H,k}$ are redefined conditional on censoring, as in:

\begin{align}
        W_{B,k}(Z)=    \prod_{j=1}^{k}
            \begin{bmatrix*}[l]
            &  \begin{pmatrix*}[l]  \frac{
                     \Big( \alpha_j(Z, R_j) \times f_{B_{j} \mid R_{j}, C_{j-1}, \overline{L}_{j}}(1 \mid R_{j}, C_{j-1}, \overline{L}_{j})\Big)^{B_{j}}  
                    \times \Big(1- \alpha_j(Z, R_j) \times f_{B_{j} \mid R_{j}, C_{j-1}, \overline{L}_{j}}(1 \mid R_{j}, C_{j-1},  \overline{L}_{j})\Big)^{1-B_j}
                }{
                f_{B_j \mid R_j, C_{j-1}, \overline{L}_j}(B_j \mid R_j, C_{j-1}, \overline{L}_j)
                } 
                \end{pmatrix*}^{\beth_{B,k}(Z)} \\
    \times &  \begin{pmatrix*}[l]  \frac{
                     \Big(1 - \alpha_j(Z, R_j) \times f_{B_{j} \mid R_{j}, C_{j-1}, \overline{L}_{j}}(0 \mid R_{j}, C_{j-1}, \overline{L}_{j})\Big)^{B_{j}}  
                    \times \Big(\alpha_j(Z, R_j) \times f_{B_{j} \mid R_{j}, C_{j-1}, \overline{L}_{j}}(0 \mid R_{j}, C_{j-1}, \overline{L}_{j})\Big)^{1-B_j}
                }{
                f_{B_j \mid R_j, C_{j-1}, \overline{L}_j}(B_j \mid R_j, C_{j-1}, \overline{L}_j)
                } 
                \end{pmatrix*}^{(1-\beth_{B,k}(Z))} \\
            \end{bmatrix*} \label{eq; Cnewweight B}.
\end{align}

and 

\begin{align}
        W_{H,k}(Z)=    \prod_{j=1}^{k}
            \begin{bmatrix*}[l]
            &  \begin{pmatrix*}[l]  \frac{
                     \Big( \beta_j(Z, S_j) \times f_{H_{j} \mid S_{j}, C_{j-1}, \overline{L}_{j}}(1 \mid S_{j}, C_{j-1}, \overline{L}_{j})\Big)^{H_{j}}  
                    \times \Big(1- \beta_j(Z, S_j) \times f_{H_{j} \mid S_{j}, C_{j-1}, \overline{L}_{j}}(1 \mid S_{j}, C_{j-1}, \overline{L}_{j})\Big)^{1-H_j}
                }{
                f_{H_j \mid S_j, C_{j-1}, \overline{L}_j}(H_j \mid S_j, C_{j-1}, \overline{L}_j)
                } 
                \end{pmatrix*}^{\beth_{H,k}(Z)} \\
    \times &  \begin{pmatrix*}[l]  \frac{
                     \Big(1 - \beta_j(Z, S_j) \times f_{H_{j} \mid S_{j}, C_{j-1}, \overline{L}_{j}}(0 \mid S_{j}, C_{j-1}, \overline{L}_{j})\Big)^{H_{j}}  
                    \times \Big(\beta_j(Z, S_j) \times f_{H_{j} \mid S_{j}, C_{j-1}, \overline{L}_{j}}(0 \mid S_{j}, C_{j-1}, \overline{L}_{j})\Big)^{1-H_j}
                }{
                f_{H_j \mid S_j, C_{j-1}, \overline{L}_j}(H_j \mid S_j, C_{j-1}, \overline{L}_j)
                } 
                \end{pmatrix*}^{(1-\beth_{H,k}(Z))} \\
            \end{bmatrix*} \label{eq; Cnewweight H}.
\end{align}

and $W_{C,k}^{g_z}$ is defined as

\begin{align}
        W_{C,k}^{g_z}=    \prod_{j=1}^{k}\frac{
                     I(C_j=0)}{
                f_{C_j \mid H_j, S_j, L_j, C_{j-1}}(C_j \mid H_j, S_j, L_j, C_{j-1})
                }.
                 \label{eq; weight C}
\end{align}

The g-formulae for $B_k$ and $H_k$ are expressed identically as in \eqref{eq; altgformB} and \eqref{eq; altgformH}, and the additional g-formulae for $R_k$ and $S_k$ as in \eqref{eq; altgformR} and \eqref{eq; altgformS}, except all densities condition on being uncensored.

\subsection{Inverse Probability Weighted Estimation of Risk under Proportionally Representative Interventions}
\subsubsection{Marginal Structural Models}

Consider a cloned dataset as in the subsection \ref{subsec; AppAMSM} of Appendix A. Then $\lambda_{Y,k}^{g_z}(V)$, is analagously expressed as in \eqref{eq; altgformYZ}, with $W_{B,k}(Z)$ and $W_{H,k}(Z)$ analagously redefined as in \eqref{eq; newweight B} and \eqref{eq; newweight H}. The intervention to abolish censoring is invariant across regimes indexed by $z$, so the weight is not redefined. The MSM is specified identically as in Appendix A.

\subsubsection{Inverse Probability Weighted Estimation}

Let $\hat{\psi}$ be the solution to the estimating equation

\begin{align}
   & \sum_{i=1}^n\sum_{z}\sum_{k=1}^K U_{i,k}(\psi, \hat{\eta}_B, \hat{\eta}_H, \hat{\eta}_C), \label{eq; Cesteq}
\end{align},

with respect to $\psi$, where 

\begin{align}
   U_k(\psi, \hat{\eta}_B, \hat{\eta}_H, \hat{\eta}_C) = & [Y_k-h\{\gamma(k,Z,V; \psi)\}] \\
                           & \times (1-Y_{k-1})W_{B,k}(Z,\hat{\eta}_B)W_{H,k}(Z, \hat{\eta}_H)W^{g_z}_{C,k}(\hat{\eta}_C). \nonumber \label{eq; CesteqU}
\end{align}

Estimated weights $W_{B,k}(Z,\hat{\eta}_B)$,  $W_{H,k}(Z, \hat{\eta}_H)$, and $W^{g_z}_{C,k}(\hat{\eta}_C)$, and their components are defined to those in Appendix A, swapping parameters models of conditional treatment and censoring with there estimated counterparts. Estimators for $\hat{\pi}_{B,j}(Z, \hat{\eta})$ and $\hat{\pi}_{H,j}(Z, \hat{\eta})$, $\hat{\pi}_{R,j}(Z, \hat{\eta})$, and $\hat{\pi}_{S,j}(Z, \hat{\eta})$ are defined analogously  as in Section \ref{subsec: IPWarb}, except by additionally weighting by estimated censoring weights, For example,

\begin{align}
    & \hat{\pi}_{B,j}(Z, \hat{\eta}) = \frac{1}{n}\sum_{i=1}^n\big[B_{i,j}W_{H,i,j-1}(Z_i, \hat{\eta})W_{B, i, j-1}(Z_i, \hat{\eta})W_{C,i,j-1}^{g_z}(\hat{\eta})\big].
\end{align}

Then, as in Section \ref{subsec; AppAMSM}, if (i) the MSM is correctly specified; and (ii), the models $f_{B_j \mid R_j, C_{j-1}, \overline{L}_j}(B_j \mid R_j, C_{j-1}, \overline{L}_j; \eta_B)$, $f_{H_j \mid S_j, C_{j-1}, \overline{L}_j}(H_j \mid S_j, C_{j-1}, \overline{L}_j; \eta_H)$, and $f_{C_j \mid H_j, S_j, L_j, C_{j-1}}(C_j \mid H_j, S_j, L_j, C_{j-1}; \eta_C)$ are correctly specified, then we have 

\begin{align}
   \mathbb{E}[U_k(\psi^*, \eta_B^*, \eta_H^*, \eta_C^*)] = 0 
\end{align}

for all $k$, with $\eta_B^*$, $\eta_H^*$, and $\eta_C^*$ the true values of $\eta_B$, $\eta_H$, and $\eta_C$ and the IPW estimator $\hat{\psi}$ consistent and asymptotically normal for $\psi^*$.

Here, we assume pooled logistic models for $f_{B_k \mid R_k, C_{k-1}, \overline{L}_k}(B_k \mid 1, 0, \overline{L}_k; \eta_B)$ and $f_{H_k \mid S_k, C_{k-1}, \overline{L}_k}(H_k \mid 1, 0, \overline{L}_k; \eta_H)$, and $f_{C_k \mid H_k, S_k, C_{k-1}, \overline{L}_k}(C_k \mid 0, 1, 0, \overline{L}_k; \eta_C)$, that is,

\begin{align}
   f_{B_k \mid R_k, C_{k-1}, \overline{L}_k}(B_k \mid 1, 0, \overline{L}_k; \eta_B) = \text{expit}\{\phi_B(k, \overline{L}_k; \eta_B)\} \label{eq; CpooledlogitB}
\end{align}

and

\begin{align}
   f_{H_k \mid S_k, C_{k-1}, \overline{L}_k}(H_k \mid 1, 0, \overline{L}_k; \eta_H) = \text{expit}\{\phi_H(k, \overline{L}_k; \eta_H)\} \label{eq; CpooledlogitH}
\end{align}

and 

\begin{align}
   f_{C_k \mid H_k, S_k, C_{k-1}, \overline{L}_k}(C_k \mid 0, 1, 0, \overline{L}_k; \eta_C) = \text{expit}\{\phi_C(k, \overline{L}_k; \eta_C)\} \label{eq; CpooledlogitC}
\end{align}

with $\phi_B$, $\phi_H$, and $\phi_C$ specified functions of $(k, \overline{L}_k)$, differentiable with respect to $\eta_B$, $\eta_H$, and $\eta_C$. Note that $f_{B_k \mid R_k, C_{k-1}, \overline{L}_j}(0 \mid 0, 0, \overline{L}_k; \eta_B) = f_{H_k \mid S_k, C_{k-1}, \overline{L}_j}(0 \mid 0, 0, \overline{L}_k; \eta_H)=1$ by definition (that is, previously treated or deceased uncensored individuals will be untreated in interval $k$, with probability 1). Further, note that $f_{B_k \mid R_k, C_{k-1}, \overline{L}_j}(0 \mid R_k, 1, \overline{L}_k; \eta_B) = f_{H_k \mid S_k, C_{k-1}, \overline{L}_j}(0 \mid S_k, 1, \overline{L}_k; \eta_B) = f_{C_k \mid H_k, S_k, C_{k-1}, \overline{L}_k}(1 \mid 0, S_k, 1, \overline{L}_k; \eta_C)=1$ by definition (that is, censored individuals will be untreated and stay censored in interval $k$ with probability 1). Finally, note that $f_{C_k \mid H_k, S_k, C_{k-1}, \overline{L}_k}(0 \mid 0, 0, 0, \overline{L}_k; \eta_C)=1$ in this particular setting (that is, previously treated or deceased uncensored individuals will remain uncensored in interval $k$, with probability 1). This last quality of the conditional censoring density is particular to the study in the applied example, where previously treated individuals will never be treated again, and where the only other time-varying variable after treatment is death, which is measured with 100\% reliability.

Assuming the same model for $h\{\gamma(k,Z,V; \psi)\}$ as in the main text, we can solve the estimating equation with the following generalized algorithm, applied to a cloned subject-interval dataset, constructed as in section \ref{subsubsec: IPW alg arb}:

\textbf{Generalized IPW estimation algorithm for $\psi$}

\begin{enumerate}
    \item [1.] Using subject-interval records with $Z=1$ and $R_k=1$ and $C_k=0$, obtain $\hat{\eta}_B$ by fitting pooled logistic regression model \eqref{eq; CpooledlogitB} with dependent variable $B_k$ and independent variables a specified function of $k=0,\dots, K$ and $\overline{L}_k$, corresponding to the choice of $\phi_B(\cdot)$.
    
    \item [2.] Using subject-interval records with $Z=1$ and $S_k=1$ and $C_k=0$, obtain $\hat{\eta}_H$ by fitting a pooled logistic regression model \eqref{eq; CpooledlogitH} with dependent variable $H_k$ and independent variables a specified function of $k=0,\dots, K$ and $\overline{L}_k$, corresponding to the choice of $\phi_H(\cdot)$.
    
    \item [3.] Using subject-interval records with $Z=1$ and $S_k=1$ and $H_k=0$ and $C_{k-1}=0$, obtain $\hat{\eta}_C$ by fitting a pooled logistic regression model \eqref{eq; CpooledlogitC} with dependent variable $C_k$ and independent variables a specified function of $k=0,\dots, K$ and $\overline{L}_k$, corresponding to the choice of $\phi_C(\cdot)$.
    
    \item [4.] For each subject's line $k$, attach the suspected-superior treatment weight, $W_{C,k}$, calculated as:

$$
            \prod_{j=1}^k\frac{1}{1-\text{expit}\{\phi_C(j, \overline{L}_j; \hat{\eta_C})\}}
$$

    \item[5.] For all $z \in \mathcal{Z}$, $r_1 \in \{0, 1\}$, set $\alpha_{0}(z, r_0, \hat{\eta})$ and $\beta_{0}(z, s_0, \hat{\eta})$ to 1. Obtain $\alpha_1(z, r_1, \hat{\eta})$, $\aleph_{B,1}(z)$, and $\beth_{B,1}(z)$ by evaluating the estimated analogues of expression \eqref{eq; alphanew}, \eqref{eq; alephB}, and \eqref{eq; bethB}, noting that $P(B_1^{g_z+}=1)$ is defined by the intervention, and taking $\hat{\pi}_{B,1}(z, \hat{\eta})$ to be the proportion of individuals with $B_1=1$, $\frac{1}{n}\sum_{i=1}^nB_{i,1}$, and noting that $\hat{\pi}_{R,1}(z, \hat{\eta})=1$ by definition.
    
    \item[6.] For each subject's line 1, attach the suspected-superior treatment weight, $W_{B,1}$, calculated as: 
    
    $$ 
       \begin{bmatrix*}[l]
            &  \begin{pmatrix*}[l]  \frac{
                     \Big( \alpha_1(Z, R_1) \times \text{expit}\{\phi_B(1, \overline{L}_1; \hat{\eta_B})\}\Big)^{B_{1}}  
                    \times \Big(1- \alpha_j(Z, R_1) \times \text{expit}\{\phi_B(1, \overline{L}_1; \hat{\eta_B})\}\Big)^{1-B_1}
                }{
                \Big( \text{expit}\{\phi_B(1, \overline{L}_1; \hat{\eta_B})\}\Big)^{B_{1}}  
                    \times \Big(1- \text{expit}\{\phi_B(1, \overline{L}_1; \hat{\eta_B})\}\Big)^{1-B_1}
                } 
                \end{pmatrix*}^{\beth_{B,1}(Z)} \\
    \times &  \begin{pmatrix*}[l]  \frac{
                     \Big(1 - \alpha_1(Z, R_1) \times (1-\text{expit}\{\phi_B(1, \overline{L}_1; \hat{\eta_B})\})\Big)^{B_{1}}  
                    \times \Big(\alpha_j(Z, R_1) \times (1-\text{expit}\{\phi_B(1, \overline{L}_1; \hat{\eta_B})\})\Big)^{1-B_1}
                }{
                \Big( \text{expit}\{\phi_B(1, \overline{L}_1; \hat{\eta_B})\}\Big)^{B_{1}}  
                    \times \Big(1- \text{expit}\{\phi_B(1, \overline{L}_1; \hat{\eta_B})\}\Big)^{1-B_1}
                } 
                \end{pmatrix*}^{(1-\beth_{B,1}(Z))} \\
            \end{bmatrix*}.
$$

    \item[7.] For all $z \in \mathcal{Z}$, $s_1 \in \{0, 1\}$, then obtain $\beta_1(z, s_1, \hat{\eta})$, $\aleph_{H,1}(z)$, and $\beth_{H,1}(z)$.
    \item[8.] For each subject's line 1, attach the suspected-inferior treatment weight, $W_{H,1}$, calculated as: 
    
    $$
       \begin{bmatrix*}[l]
            &  \begin{pmatrix*}[l]  \frac{
                     \Big( \beta_1(Z, S_1) \times \text{expit}\{\phi_H(1, \overline{L}_1; \hat{\eta_H})\}\Big)^{H_{1}}  
                    \times \Big(1- \beta_1(Z, S_1) \times \text{expit}\{\phi_H(1, \overline{L}_1; \hat{\eta_H})\}\Big)^{1-H_1}
                }{
                \Big( \text{expit}\{\phi_H(1, \overline{L}_1; 
                    \hat{\eta_H})\}Big)^{H_{1}}  
                    \times \Big(1- \text{expit}\{\phi_H(1, \overline{L}_1; \hat{\eta_H})\}\Big)^{1-H_1}
                } 
                \end{pmatrix*}^{\beth_{H,1}(Z)} \\
    \times &  \begin{pmatrix*}[l]  \frac{
                     \Big(1 - \beta_1(Z, S_1) \times (1-\text{expit}\{\phi_H(1, \overline{L}_1; \hat{\eta_H})\})\Big)^{B_{1}}  
                    \times \Big(\beta_1(Z, S_1) \times (1-\text{expit}\{\phi_H(1, \overline{L}_1; \hat{\eta_H})\})\Big)^{1-H_1}
                }{
                \Big( \text{expit}\{\phi_H(1, \overline{L}_1; \hat{\eta_H})\}\Big)^{H_{1}}  
                    \times \Big(1- \text{expit}\{\phi_H(1, \overline{L}_1; \hat{\eta_H})\}\Big)^{1-H_1}
                } 
                \end{pmatrix*}^{(1-\beth_{H,1}(Z))} \\
            \end{bmatrix*}.
$$
    
    \item[9.] Iterate from $k=2,\dots K$:
    \begin{enumerate}
       \item [9.1.] For all $z \in \mathcal{Z}$, $r_k \in \{0, 1\}$, obtain $\alpha_{k}(z, r_k, \hat{\eta})$, $\aleph_{B,k}(z)$, and $\beth_{B,k}(z)$.
       \item [9.2.] Using subject-interval records on line $k$, attach the suspected-superior treatment weight, $W_{B,k}$, calculated as: 
       
    $$ \prod_{j=1}^k
       \begin{bmatrix*}[l]
            &  \begin{pmatrix*}[l]  \frac{
                     \Big( \alpha_j(Z, R_j) \times \text{expit}\{\phi_B(j, \overline{L}_j; \hat{\eta_B})\}\Big)^{B_{j}}  
                    \times \Big(1- \alpha_j(Z, R_j) \times \text{expit}\{\phi_B(j, \overline{L}_j; \hat{\eta_B})\}\Big)^{1-B_j}
                }{
                \Big( \text{expit}\{\phi_B(j, \overline{L}_j; \hat{\eta_B})\}\Big)^{B_{j}}  
                    \times \Big(1- \text{expit}\{\phi_B(j, \overline{L}_j; \hat{\eta_B})\}\Big)^{1-B_j}
                } 
                \end{pmatrix*}^{\beth_{B,j}(Z)} \\
    \times &  \begin{pmatrix*}[l]  \frac{
                     \Big(j - \alpha_j(Z, R_j) \times (1-\text{expit}\{\phi_B(j, \overline{L}_j; \hat{\eta_B})\})\Big)^{B_{j}}  
                    \times \Big(\alpha_j(Z, R_j) \times (1-\text{expit}\{\phi_B(j, \overline{L}_j; \hat{\eta_B})\})\Big)^{1-B_j}
                }{
                \Big( \text{expit}\{\phi_B(j, \overline{L}_j; \hat{\eta_B})\}\Big)^{B_{j}}  
                    \times \Big(1- \text{expit}\{\phi_B(j, \overline{L}_j; \hat{\eta_B})\}\Big)^{1-B_j}
                } 
                \end{pmatrix*}^{(1-\beth_{B,j}(Z))} \\
            \end{bmatrix*}.
$$
                
        \item [9.3.] For all $z \in \mathcal{Z}$, $s_k \in \{0, 1\}$, obtain $\beta_{k}(z, s_k, \hat{\eta})$, $\aleph_{H,k}(z)$, and $\beth_{H,k}(z)$.
        \item [9.4.] Using subject-interval records on line $k$ with $Z=z$, attach the suspected-superior treatment weight, $W_{H,k}$, calculated as:  
    
    $$
    \prod_{j=1}^k
       \begin{bmatrix*}[l]
            &  \begin{pmatrix*}[l]  \frac{
                     \Big( \beta_j(Z, S_j) \times \text{expit}\{\phi_H(j, \overline{L}_j; \hat{\eta_H})\}\Big)^{H_{j}}  
                    \times \Big(1- \beta_j(Z, S_j) \times \text{expit}\{\phi_H(j, \overline{L}_j; \hat{\eta_H})\}\Big)^{1-H_j}
                }{
                \Big( \text{expit}\{\phi_H(j, \overline{L}_j; 
                    \hat{\eta_H})\}Big)^{H_{j}}  
                    \times \Big(1- \text{expit}\{\phi_H(j, \overline{L}_j; \hat{\eta_H})\}\Big)^{1-H_j}
                } 
                \end{pmatrix*}^{\beth_{H,j}(Z)} \\
    \times &  \begin{pmatrix*}[l]  \frac{
                     \Big(j - \beta_j(Z, S_j) \times (1-\text{expit}\{\phi_H(j, \overline{L}_j; \hat{\eta_H})\})\Big)^{B_{j}}  
                    \times \Big(\beta_j(Z, S_j) \times (1-\text{expit}\{\phi_H(j, \overline{L}_j; \hat{\eta_H})\})\Big)^{1-H_j}
                }{
                \Big( \text{expit}\{\phi_H(j, \overline{L}_j; \hat{\eta_H})\}\Big)^{H_{j}}  
                    \times \Big(1- \text{expit}\{\phi_H(j, \overline{L}_j; \hat{\eta_H})\}\Big)^{1-H_j}
                } 
                \end{pmatrix*}^{(1-\beth_{H,j}(Z))} \\
            \end{bmatrix*}.
$$
        
    \end{enumerate}
\item[10.] Using all subject-interval records in the cloned dataset, obtain $\hat{\psi}$ by fitting a \textit{weighted} pooled logistic regression model, with weights $W_{B,k}$ and $W_{H,k}$ and $W_{C,k}$ defined in the previous steps, dependent variable $Y_k$ and independent variables a specified function of $k=1,\dots,K$ and $(Z, V)$ corresponding to the choice of $\gamma(\cdot)$.
\end{enumerate}\
\\

Our final IPW estimate of the g-formula for the risk of death by $K$ under regime $g_z$, $f^{g_z}_{Y_K}(1)$ defined by the arbitrary proportionally-representative interventions that constrain resources under the natural course can then be obtained by the plug-in estimator of expression \eqref{eq; plugin} in the main text.

\clearpage

\section{Appendix G: Model specifications in Data Analysis}

In constructing the denominator of the weights, we assumed that

\begin{align}
    \phi_B(k, \overline{L}_k; \eta_B) = \eta_{B,0}+\eta_{B, 1}^{T}g(k) + \phi^{\sim}(k, \overline{L}_k: \eta_B^{\sim}),
\end{align}

where $g(k)$ is specified as a natural cubic spline function with internal knots at $\xi_I=(2,4,12,24,54)$, boundary knots at $\xi_B=(1, 120)$, and truncated power basis functions $t_I(k, \xi_I)= \begin{Bmatrix*}[l] & (k-\xi_I)^3 \text{ if } k>\xi_I \\ & 0 \text{ otherwise.} \end{Bmatrix*}$, and $t_B(k, \xi_B)= \begin{Bmatrix*}[l] & (k-\xi_B) \text{ if } k>\xi_B \\ & 0 \text{ otherwise.} \end{Bmatrix*}$: 

\begin{align}
    g(k) = \{k, k^2, K^3, t_I(k, \xi_{I, 1}),\dots, t_I(k, \xi_{I, 5}), t_B(k, \xi_{B, 1}), t_B(k, \xi_{B, 2}) \},
\end{align}

and assumed

\begin{align}
    \phi_H(k, \overline{L}_k; \eta_H) = \eta_{H,0}+\eta_{H, 1}^Tg(k) + \phi^{\sim}(k, \overline{L}_k: \eta_H^{\sim}).
\end{align}.

Further, since all regimes involved interventions to abolish censoring, we constructed censoring weights, and, as such, specified censoring models, where 

\begin{align}
    \phi_C(k, \overline{L}_k; \eta_C) = \eta_{C,0}+\eta_{C, 1}^Tg(k) + \phi^{\sim}(k, \overline{L}_k: \eta_C^{\sim}).
\end{align}

The functional form of $\phi^{\sim}(\cdot)$ is common across treatment and censoring models and consists of the product of a model-specific vector of parameters with a vector of features defined by $L_k$. When $k=1$, these features are defined in Table \ref{table; feature table}, where categorical features are transformed into a subvector of dummy variables, and where continuous features are specified as natural cubic splines with internal knots at their 35$^{th}$ and 65$^{th}$ percentiles, and boundary knots at their 5$^{th}$ and 95$^{th}$ percentiles, parameterized analagous to $g(k)$ above. Time varying features include MELD score, specified as a natural cubic spline function identically to other continuous variables, and an interaction between Baseline MELD exception status and $g(k)$.

\clearpage
\tiny
\begin{longtable}{lrr}
\caption{Table of Baseline Covariates}\\
  \hline
  \specialcell{Variable} & \specialcell{Values} \\ 
  \hline
  Baseline MELD &  Continuous  \\ 
  Baseline MELD exception & Yes/No \\
  Status 1 & Yes/No \\
  Gender & Male/Female \\
  Race (categorical) & Pacific Islander \\
       & Hispanic or Latino\\
       & Asian\\
       & Black or African American\\
       & Native American\\
       & White\\
       & Multi-Racial\\
  Year of Listing (categorical) & 2005-2015 \\
  Age (years) & Continuous \\
  Height (cm) & Continuous \\
  Weight (kg) & Continuous \\
  Willingness to... &  \\
  \hskip .5cm     Accept Incompatible Blood Type & Yes/No \\
  \hskip .5cm     Accept Extra Corporeal Liver & Yes/No \\
  \hskip .5cm     Accept Liver Segment & Yes/No \\
  \hskip .5cm     Accept HBV Positive Donor & Yes/No \\
  \hskip .5cm     Accept HCV Positive Donor & Yes/No \\
  Patient On Life Support & Yes/No \\
  Functional Status (categorical) & Requires No Assistance  \\
                    & Some assistance\\
                    & Total assistance\\
  Primary Diagnosis (categorical) & Cholestatic \\
                    & Fulminant Hepatic Failure\\
                    & Malignant Neoplasm\\
                    & Metabolic\\
                    & Non-cholestatic/Other\\
  Spontaneous Bacterial Peritonitis & Yes/No \\
  History of Portal Vein Thrombosis & Yes/No \\
  History of TIPSS & Yes/No \\
   \hline
\hline \label{table; feature table}
\end{longtable}
\normalsize

\clearpage

\end{document}